\documentclass[a4paper,11pt]{article}
\input{Preamble.tex}

\begin{document}

\title{Fast FPT-Approximation of Branchwidth}

\author{
Fedor V. Fomin\thanks{University of Bergen, Norway. \texttt{fomin@ii.uib.no}}
 \and Tuukka Korhonen\thanks{University of Bergen, Norway. \texttt{tuukka.korhonen@uib.no}}
}

\maketitle

\thispagestyle{empty}

\begin{abstract}
%
%
Branchwidth determines how graphs, and more generally, arbitrary connectivity (basically symmetric and submodular) functions could be decomposed into a tree-like structure by specific cuts. We develop a  general framework for designing fixed-parameter tractable (FPT)  2-approximation algorithms for branchwidth of connectivity functions.  
The first ingredient of our framework is combinatorial. 
We prove a structural theorem establishing that either a sequence of particular refinement operations could decrease the width of a branch decomposition or that the width of the decomposition is already within a factor of 2 from the optimum. 
The second ingredient is an efficient implementation of the refinement operations for branch decompositions that support efficient dynamic programming.
We present two concrete applications of our general framework.
\begin{itemize}
\item
An algorithm that for a given $n$-vertex graph $G$ and integer $k$ in time  $2^{2^{\OO(k)}} n^2$ either constructs a rank decomposition of $G$ of width at most $2k$ or concludes that the rankwidth of $G$ is more than $k$. It also yields a $(2^{2k+1}-1)$-approximation algorithm for cliquewidth within the same time complexity, which in turn, improves to $f(k)n^2$ the running times of various algorithms on graphs of cliquewidth $k$. Breaking the ``cubic barrier'' for rankwidth and cliquewidth  was an open problem in the area.
 
\item An algorithm that for a given $n$-vertex graph $G$ and integer $k$ in time  $2^{\OO(k)} n$ either constructs a branch decomposition of $G$ of width at most $2k$ or concludes that the branchwidth of $G$ is more than $k$. 
This improves over the 3-approximation that follows from the recent treewidth 2-approximation of Korhonen [FOCS 2021].
\end{itemize}
\end{abstract}

\newpage
\pagestyle{plain}
\setcounter{page}{1}

\section{Introduction}
The branchwidth of  a connectivity function $f$ (that is,   $f(\emptyset) = 0$, and $f$ is  symmetric and submodular) was introduced by Robertson and Seymour in \cite{RobertsonS91}.   Let $V$ be a finite set and $f : 2^V \rightarrow \mathbb{Z}_{\ge 0}$ be a connectivity function on $V$.
A \emph{branch decomposition} of   $f $ is a pair $(T, L)$, where  $T$ is a cubic tree (the degree of each non-leaf node of $T$ is  3) and $L$ is a bijection mapping  $V$ to the leaves of $T$. (If $|V|\leq 1$, then $f$ admits no branch decomposition.) For every edge $e$ of $T$,  
the connected components of $T\setminus\{ e\}$, the graph obtained from $T$ by deleting $e$,  induce a partition $(X, Y )$ of the set of leaves of $T$.  The \emph{width} of  $e$ is $f(e) = f(L^{-1}(X))=f(L^{-1}(Y))$. The \emph{width of $(T , L)$}  is the maximum width of all edges of $T$. The \emph{branchwidth} $\bw(f )$ of $f$ is the minimum width of a branch decomposition of $f$. In this paper, we develop a framework for designing fixed-parameter tractable (\FPT) 2-approximation algorithms for computing branch decompositions of connectivity functions. We provide a detailed overview of the framework in \Cref{subsec_framework}.   Here we discuss some of its concrete algorithmic consequences. 

\medskip\noindent\textbf{Rankwidth.}
Rankwidth was introduced by Oum and Seymour \cite{OumS06}. The rank decomposition of a graph $G$  is the branch decomposition of the following connectivity function $f$  defined on the   vertex set $V(G)$. 
For a graph $G$ and a pair $A,B$ of disjoint subsets of $V(G)$, let $G[A,B]$ be the bipartite graph induced by edges between $A$ and $B$.
Let  $M_G[A,B]$ be the  $|A| \times |B|$ 0-1~matrix representing $G[A,B]$.
Then the value $f(A)=f(V(G) \setminus A)=\rank(M_G[A,V(G) \setminus A])$ is the GF(2)-rank of  $M_G[A,V(G) \setminus A]$.
%
 By making use of our algorithmic framework for connectivity functions, we prove the following theorem about  approximation of rankwidth.
 \begin{restatable}{theorem}{thmrankwidthalgorithm}
 \label{theorem_rw_algo}
For integer $k$ and an $n$-vertex graph $G$, there is an algorithm that in time 
 $2^{2^{\OO(k)}} n^2$ either computes a rank decomposition of $G$ of width at most $2k$, or correctly concludes that the rankwidth of $G$ is more than $k$.
 \end{restatable}
 Several algorithms computing rankwidth exactly or approximately are known in the literature. After a number of  improvements, the best  running times of algorithms computing rankwidth are of the form $\OO(f(k)\cdot n^3)$ \cite{HlinenyO08,Jeong0O18,Oum08,OumS06,OumS07}. As Oum wrote in his survey \cite{Oum17}, 
 \emph{``However, for rank-width, it seems quite non-trivial to reduce $n^3$ into $n^c$ for some $c<3$."} \Cref{theorem_rw_algo} affirmatively answers the open question of Oum \cite[Question~3]{Oum17}, who asked whether 
 there exists an algorithm with functions $f(k)$, $g(k)$ and a constant $c<3$ that finds a rank decomposition of width at most $f(k)$ or confirms that the rankwidth of a  graph is larger than $k$, in time $\OO(g(k)\cdot n^c)$.
Pipelined with the previous work on rankwidth and cliquewidth, \Cref{theorem_rw_algo}  has several important consequences.

 \medskip\noindent\textbf{Cliquewidth.}
A common approach in graph algorithms is to decompose a graph by making use of small separations. Perhaps the most popular measure of graph decomposition is the treewidth of a graph \cite[Chapter~7]{cygan2015parameterized}. 
Many \NP-hard optimization problems can be solved efficiently on graphs of bounded treewidth. The seminal result of Courcelle \cite{Courcelle92a,Courcelle97} (see also \cite{ArnborgLS91,BoriePT92,Courcelle:2012book}) combined with the algorithm of Bodlaender~\cite{Bodlaender96}, states that every decision
 problem on graphs expressible in
Monadic Second Order Logic ($\MSO_2$) is solvable in linear time on graphs of bounded treewidth.  
However, the average vertex degree of a graph of treewidth $k$ does not exceed $k$, limiting the algorithmic applicability of treewidth to ``sparse'' graph classes.
Arguably, the most successful project of extending the meta-algorithmic results from graphs of bounded treewidth to ``non-sparse'' graphs is by making use of the cliquewidth of a graph defined by Courcelle, Engelfriet, and  Rozenberg~\cite{CourcelleER93}. 
We give the formal definition due to its technicality in Appendix. 
 Intuitively, a graph is of cliquewidth at most $k$ if it can be built from single vertices following a $k$-expression, which identifies how to systematically join the already constructed parts of the graph. Moreover, in each constructed part, the vertices can be partitioned into at most $k$ types such that vertices of the same type will be indistinguishable in later steps of the construction. 
Cliquewidth generalizes treewidth in the following sense. 
 Every graph $G$ of treewidth at most $k$ has cliquewidth at most $\OO(2^k)$.   On the other hand, there are graphs, like the $n$-vertex complete graph $K_n$ whose treewidth is $n-1$ and cliquewidth is a constant~\cite{HlinenyOSG08}.

 Oum and Seymour \cite{OumS06} proved that for any graph of rankwidth $k$, its cliquewidth is always between $k$ and $2^{k+1}-1$. Moreover, their proof gives an algorithm that in time $ 2^{\OO(k)} n^2$  
 converts a rank decomposition of width $k$ into a $(2^{k+1}-1)$-expression for the cliquewidth. By combining the construction of Oum and Seymour  with \Cref{theorem_rw_algo}, we derive the following corollary. 
 \begin{corollary} \label{cor_cw_appr}
 For  integer $k$ and $n$-vertex graph $G$, there is an algorithm that in time 
 $2^{2^{\OO(k)}} n^2$ either computes a
 $(2^{2k+1}-1)$-expression of $G$ or  correctly  concludes that the  cliquewidth of $G$ is more than $k$. 
 \end{corollary}

  The importance of cliquewidth is due to its algorithmic properties. 
  Courcelle, Makowsky, and Rotics in \cite{CourcelleMR00} identified a variation of $\MSO_2$, called $\MSO_1$, and showed how to extend Courcelle's theorem for this variation of $\MSO_2$ to graphs of bounded cliquewidth. 
 Informally, $\MSO_1$ is the class of  $\MSO_2$  formulas that allow quantification over subsets of vertices, but not of edges. While being less expressive than  $\MSO_2$, $\MSO_1$ still captures a broad class of optimization problems on graphs, including 
 vertex cover, dominating set, domatic number for fixed $k$, $k$-colorability for fixed $k$, partition into cliques for fixed $k$, clique, independent set, and induced path.  
 The meta-theorem of Courcelle, Makowsky, and Rotics \cite{CourcelleMR00} states that every $\MSO_1$-definable problem on graphs is solvable in time $\OO(f(k)\cdot (n+m))$ when a $k$-expression is provided with the input.
The theorem of Courcelle, Makowsky, and Rotics applies also to a more general class of problems, like optimization problems searching for sets of vertices that are optimal concerning some linear
evaluation function (for example,  a clique of the maximum weight) or counting \cite{CourcelleMR00,CourcelleMR01}. 

The applicability of the meta-theorem of Courcelle, Makowsky, and Rotics crucially depends on the efficiency of computing the cliquewidth of a graph and constructing the corresponding $k$-expression.
 The only known way of constructing (an approximate) $k$-expression for cliquewidth, as well as for related graph parameters like   NLC-width  \cite{Wanke94}  or  Boolean width \cite{Bui-XuanTV11}, is by making use of rankwidth.
 Combining \autoref{cor_cw_appr} with the meta-theorem of Courcelle, Makowsky, and Rotics implies that every $\MSO_1$-definable problem is solvable in quadratic $\OO(f(k)\cdot n^2)$  running time on graphs of cliquewidth at most $k$.
This is the first improvement on the time complexity of this meta-theorem since the $\OO(f(k) \cdot n^3)$ algorithm given by Oum in 2005~\cite{Oum08,oumthesis}.
    
\medskip\noindent\textbf{Exact rankwidth.}  
Courcelle and Oum \cite{CourcelleO07} proved that  for every fixed $k$, there are only finitely many graphs, such that a graph $G$ does not contain any of them as a vertex-minor if and only if   the rankwidth of $G$ at most $k$. They also proved that for every fixed graph $H$, the property that $H$  is isomorphic to a vertex-minor of an input graph $G$  is expressible in a variant of $\MSO_1$ that can be checked in time $\OO(f(k) \cdot (n+m))$ 
 if a $k$-expression is provided. 
Combined with \autoref{cor_cw_appr}, this implies the following.
\begin{corollary}\label{corollary_exactMSO}
For integer  $k$, deciding whether the rankwidth of an $n$-vertex graph $G$ is at most $k$ can be done in time $f(k)n^2$. 
\end{corollary}
Function $f(k)$ in \autoref{corollary_exactMSO} is huge, it is the tower   $2^{2^{\iddots^{k} }} $, where the height of the tower is bounded by some function of $k$. Also this approach does not provide the rank decomposition.

Jeong, Kim, and Oum \cite{Jeong0O18} developed an alternative  framework for computing branch decompositions of finite-dimensional vector spaces over a fixed finite field. As one of the applications of their method, they gave an FPT algorithm that for an input graph $G$ and integer $k$ in time $2^{2^{O(k^2)}}n^3$ either constructs a  rank decomposition of with $\leq k$ or concludes that the rankwidth of $G$ is more than $k$. The algorithm of 
Jeong, Kim, and Oum does not rely on vertex-minors and logic and can be seen as a (very nontrivial) adaptation to the space branchwidth of the dynamic programming algorithm of Bodlaender and Kloks for treewidth
\cite{BodlaenderK96}.  
Since the algorithm of Jeong, Kim, and Oum works for more general objects,  first they do a preprocessing step.  However for rankwidth computation this step is not required.  Then the 
 cubic running time of their algorithm is because due to the lack of a fast approximation algorithm, they have to use iterative compression. However, by combining with our approximation algorithm, the dynamic programming of Jeong, Kim, and Oum  directly implies 
 the following corollary. 
 \begin{corollary}\label{corollary_exactDP}
For integer $k$ and an $n$-vertex graph $G$, there is an algorithm that in time $2^{2^{O(k^2)}}n^2$ either constructs a rank decomposition of width at most $k$, or correctly concludes that the rankwidth of $G$ is more than $k$.
\end{corollary}

\medskip\noindent\textbf{Branchwidth  of a graph.}   As another application of the new algorithmic framework for connectivity functions, we obtain a 2-approximation algorithm for the branchwidth of a graph. In this case, the connectivity function $f$ is defined on the edge set $E(G)$ of a graph $G$. For $A \subseteq E(G)$, $f(A)$ is the number of vertices that are incident to both $A$ and $E(G) \setminus A$.
 The branchwidth of a graph is a ``close relative'' of the treewidth: For every graph  of branchwidth $k$, its treewidth is between $k-1$ and $3k/2 -1$ \cite{RobertsonS91}. 
Thus any $c$-approximation of treewidth is also a $3c/2$-approximation for branchwidth and vice versa.
  It appears that in certain situations, branchwidth could be more convenient to work with than treewidth, for example, when it comes to some dynamic programming algorithms \cite{CookS03,FastH17,FominT06}.  
The previously best-known approximation of branchwidth of running time $2^{\OO(k)}n$ is a 3-approximation that follows from the treewidth 2-approximation algorithm of Korhonen   \cite{Korhonen21}.
We prove the following theorem. 
 \begin{restatable}{theorem}{thmbranchwidthalgorithm}
 \label{theorem_vw_algo}
For  integer $k$ and an $n$-vertex graph $G$, there is an algorithm that in time 
 $2^{{\OO(k)}} n$ either computes a branch decomposition of $G$ of width at most $2k$, or correctly concludes that the branchwidth of $G$ is more than $k$.
 \end{restatable}

\subsection{Previous work}

Branchwidth of a connectivity function was introduced by Robertson and Seymour   in Graph Minors X \cite{RobertsonS91}.  Oum and Seymour \cite{Oum08} gave an algorithm of running time 
$\OO (|V|^{8k+12} \log |V|)$ deciding whether the branchwidth of a connectivity function is at most $k$. In \cite{OumS06}, 
Oum \& Seymour also gave a 3-approximation algorithm of running time  $O(|V|^6 \delta\log|V|)$, where $\delta$ is the time for each evaluation of an interpolation of $f$. (We refer  \cite{OumS06} for more details.)

\begin{table}
  \begin{center}
    \begin{tabular}{|c|c|c|c|}
      \hline
      Reference                                 & Approximation   &  TIME &  Remarks \\ \hline
       Oum \& Seymour \cite{OumS06}        & $3k+1$         & $8^k n^9\log{n} $ & Works for connectivity functions  \\
       Oum  \cite{Oum08} & $3k+1$ & $8^k n^4 $ & \\
        Oum  \cite{Oum08} & $3k-1$ & $ h(k) n^3 $ & \\
        Courcelle and Oum \cite{CourcelleO07} & exact &   $ h(k) n^3 $ & Does not provide decomposition\\ 
        Hlinen{\'{y}} \& Oum \cite{HlinenyO08}&   exact & $ h(k) n^3 $ & \\
        Jeong, Kim, \& Oum \cite{Jeong0O18}  &exact & $2^{2^{O(k^2)}} n^3$   & Works for spaces over finite fields \\ 
            This paper                                & $2k$      &  $2^{2^{\OO(k)}}n^2$ &   \\
            This paper                                & exact      &  $2^{2^{O(k^2)}}n^2$ & Our approximation +  \cite{Jeong0O18}\\
      \hline
    \end{tabular}
  \end{center}
  \caption{Overview of rankwidth  algorithms.    Here $k$ is the rankwidth
    and $n$ is the number of vertices of an input graph $G$.
    Unless otherwise specified, each of the algorithms outputs in $\OO(\text{TIME})$ 
      a
    decomposition of width given in the Approximation column. The function $h(k)$ is a huge function whose bound depends 
   on  huge list of forbidden minors in matroids or vertex-minors in graphs and the length of the logic formula describing such minors.
     }
  \label{table:cw_history}
\end{table}

Rankwidth was introduced by  Oum \& Seymour in \cite{OumS06} as a tool for \FPT-approximation of cliquewidth. 
As it was observed by Oum in \cite{Oum08},  deciding whether the rankwidth of a graph is at most $k$ is  \NP-complete. 
The first \FPT-approximation   algorithm for cliquewidth is due to  Oum \& Seymour \cite{OumS06}. 
 Its running time is  $\OO(8^k n^9\log{n})$ and it outputs a rank decomposition of width at most $3k+1$ or concludes that the rankwidth of the input graph is more than $k$.  Since then a number of \FPT  exact and approximation algorithms were developed \cite{CourcelleO07,OumS06,HlinenyO08,Oum08,OumS07,Jeong0O18}, see  \autoref{table:cw_history} for an overview. 
 The running times of all these algorithms is at least cubic in $n$. The existence of an \FPT-approximation in subcubic time was widely open, see \cite[Question~3]{Oum17}.
 Rankwidth $1$ graphs are known as 
distance-hereditary graphs \cite{Oum05}. There is a (non-trivial) linear time $\OO(n+m)$ recognition algorithm for distance-hereditary graphs  \cite{DamiandHP01}. To the best of our knowledge, already for rankwidth $k=2$ no algorithm faster than $\OO(n^3)$ was known. 
Algorithms for computing the branchwidth of certain matroids were studied by  Hlinen{\'{y}} in~\cite{Hlineny05,Hlineny06}.

Following the work of Courcelle, Makowsky, and Rotics in \cite{CourcelleMR00}, there is a lot of literature on developing algorithms (\FPT and \XP) on graphs of bounded cliquewidth   \cite{GerberK03,GodlinKM08,KoblerR01,KoblerR03,Rao07,SuchanT07}. The design of algorithms on rank decompositions  \cite{GanianH10,GanianHO13,DBLP:journals/dam/Bui-XuanTV10} and related graph parameters ``sandwiched'' by the rankwidth like NLC-width  \cite{Wanke94}  and Boolean width \cite{Bui-XuanTV11} is also an active research area in graph algorithms. For each of these width parameters, the only known way to compute the corresponding 
expression or decomposition is through rank decompositions. 
By   \Cref{theorem_rw_algo}, the running times of  all  \FPT algorithms on graphs of   $k$-cliquewidth, $k$-rankwidth, $k$-Boolean width or $k$-NLC width, improve from $f(k)n^3$ to $f(k)n^2$.


%
%

The branchwidth of a graph  was introduced by Robertson and Seymour   in Graph Minors~X~\cite{RobertsonS91}. 
The branchwidth of a planar graph is computable in polynomial time but in general the problem is \NP-hard \cite{SeymourTh94}. Bodlaender and Thilikos  in \cite{BodlaenderThil97} developed an algorithm computing the branchwidth of a graph in time $2^{\OO(k^3)} n$. 

  As we already mentioned, branchwidth of a graph is an invariant similar to treewidth and within a constant factor of treewidth. By the seminal result of Bodlaender \cite{Bodlaender96}, deciding whether the treewidth of a graph is at most $k$ can be done in time $2^{\OO(k^3)} n$. There are several   linear time  \FPT algorithms with single-exponential dependence in $k$ that reach a constant approximation ratio~\cite{BodlanderDDFLP13,Korhonen21}. In particular, Korhonen \cite{Korhonen21} obtained an algorithm that in time $2^{\OO(k)} n$ computes a tree decomposition of width at most $2k+1$ or concludes that the treewidth of the graph is more than $k$. When it comes to branchwidth, the result of  Korhonen \cite{Korhonen21} is incomparable  with \Cref{theorem_vw_algo}. 
\Cref{theorem_vw_algo} implies 3-approximation for treewidth, which is worse than Korhonen's algorithm, while Korhonen's algorithm yields a 3-approximation for branchwidth, which is worse than \Cref{theorem_vw_algo}.

\paragraph{Organization of the paper.}
We overview our framework for designing \FPT 2-approximation algorithms for computing branch decomposition in \Cref{subsec_framework}, outlining the main novel techniques.  \Cref{sec_notations} contains definitions and simple facts about branch decompositions. 
In \Cref{sec_comb_results} we develop the combinatorial tools of our framework which will be used in the next sections.    \Cref{section_algoprorefine} is devoted to the  algorithmic part of our framework. 
\Cref{sec_rankwidth_approx} implements our framework for rankwidth and \Cref{sec_gr_bw} for branchwidth of a graph. 

\section{Overview of our framework}\label{subsec_framework}
%
In this section, we overview our framework for designing \FPT 2-approximation algorithms for computing branch decompositions.
Our framework consists of two parts. The first part is combinatorial and the second part is algorithmic. 
We first overview the combinatorial part, then the algorithmic part, and then the specific applications for approximating rankwidth and branchwidth of graphs.


\subsection{Combinatorial framework}

Combinatorial arguments about how and when a branch decomposition can be improved form  the core of our framework. For a branch decomposition $(T,L)$ of a function $f$, we want to identify whether  $(T,L)$ could be refined into a  ``better'' branch decomposition $(T',L')$. 
By better, we mean the following. 
Let $k$ be the width of $(T,L)$ and $h$ be the number of heavy edges of $T$, that is, the edges $e$ where $f(e)$ is the width of  $(T,L)$.  
Then   $(T',L')$ is \emph{better}  than $(T,L)$ if the width of $(T',L')$ is at most $k$ and the number of edges of width $k$ in $T'$ is less than $h$.

Our central combinatorial insight is that if the width of $(T,L)$ is more than $2\bw(f)$, then for any  
heavy edge $e$, there is a partition of $V$ into three sets $(C_1, C_2, C_3)$ with some particular properties,  such that  the quadruple  $(e,C_1, C_2, C_3)$ ``enforces''  a better refinement of $T$. 
We start first with explaining how  edge $e$ and   tripartition $(C_1, C_2, C_3)$ enforce a refinement. 
\begin{figure}[!t]
\centering
\begin{tikzpicture}[every node/.style={draw, circle, fill=black, scale=0.4}]
\node (u) at (-0.5,0) {};
\node[fill=none, draw=none, scale=2.5] (uT) at (-0.4,0.3) {$u$};

\node (v) at (0.5,0) {};
\node[fill=none, draw=none, scale=2.5] (vT) at (0.4,0.3) {$v$};

\node (x) at (-1,0.7) {};

\node (a) at (-1.5,1.4) {};
\node[fill=none, draw=none, scale=2.5] (aT) at (-1.8,1.5) {$a$};

\node (b) at (-1.8,0.7) {};
\node[fill=none, draw=none, scale=2.5] (bT) at (-2.1,0.7) {$b$};

\node (y) at (-1,-0.7) {};

\node (c) at (-1.8,-0.7) {};
\node[fill=none, draw=none, scale=2.5] (cT) at (-2.1,-0.7) {$c$};

\node (d) at (-1.5,-1.4) {};
\node[fill=none, draw=none, scale=2.5] (dT) at (-1.8,-1.5) {$d$};

\node (z) at (1.2,-0.5) {};

\node (g) at (2.1,-0.2) {};
\node[fill=none, draw=none, scale=2.5] (gT) at (2.4,-0.1) {$g$};

\node (h) at (1, 0.7) {};
\node[fill=none, draw=none, scale=2.5] (hT) at (1.3,0.8) {$h$};

\node (w) at (1.8, -1) {};

\node (e) at (2.4, -1.6) {};
\node[fill=none, draw=none, scale=2.5] (eT) at (2.7,-1.7) {$e$};

\node (f) at (2.5, -0.9) {};
\node[fill=none, draw=none, scale=2.5] (fT) at (2.8,-0.9) {$f$};

\path (u) edge (v);
\path (u) edge (x);
\path (x) edge (a);
\path (x) edge (b);
\path (u) edge (y);
\path (y) edge (d);
\path (y) edge (c);
\path (v) edge (z);
\path (z) edge (g);
\path (v) edge (h);
\path (z) edge (w);
\path (w) edge (e);
\path (w) edge (f);

\node[fill=none, draw=none, scale=4] (ar) at (0,-1.8) {$\downarrow$};
\end{tikzpicture}
\vspace{-1em}
\begin{tabular}{lcr}
\begin{minipage}{.32\linewidth}
\begin{tikzpicture}[scale=0.8, every node/.style={draw, circle, fill=black, scale=0.4}]
\node (u) at (-0.5,0) {};
\node[fill=none, draw=none, scale=2.5] (uT) at (-0.4,0.3) {$u_1$};

\node (v) at (0.5,0) {};
\node[fill=none, draw=none, scale=2.5] (vT) at (0.4,0.3) {$v_1$};

\node (x) at (-1,0.7) {};

\node (a) at (-1.5,1.4) {};
\node[fill=none, draw=none, scale=2.5] (aT) at (-1.8,1.5) {$a$};

\node (b) at (-1.8,0.7) {};
\node[fill=none, draw=none, scale=2.5] (bT) at (-2.1,0.7) {$b$};

\node (y) at (-1,-0.7) {};

\node (c) at (-1.8,-0.7) {};

\node (d) at (-1.5,-1.4) {};

\node (z) at (1.2,-0.5) {};

\node (g) at (2.1,-0.2) {};
\node[fill=none, draw=none, scale=2.5] (gT) at (2.4,-0.1) {$g$};

\node (h) at (1, 0.7) {};

\node (w) at (1.8, -1) {};

\node (e) at (2.4, -1.6) {};

\node (f) at (2.5, -0.9) {};

\path (u) edge (v);
\path (u) edge (x);
\path (x) edge (a);
\path (x) edge (b);
\path (u) edge (y);
\path (y) edge (d);
\path (y) edge (c);
\path (v) edge (z);
\path (z) edge (g);
\path (v) edge (h);
\path (z) edge (w);
\path (w) edge (e);
\path (w) edge (f);
\end{tikzpicture}
\vspace{2.5em}
\end{minipage}
&
\begin{minipage}{.32\linewidth}
\begin{tikzpicture}[scale=0.8, every node/.style={draw, circle, fill=black, scale=0.4}]
\node (u) at (-0.5,0) {};
\node[fill=none, draw=none, scale=2.5] (uT) at (-0.4,0.3) {$u_2$};

\node (v) at (0.5,0) {};
\node[fill=none, draw=none, scale=2.5] (vT) at (0.4,0.3) {$v_2$};

\node (x) at (-1,0.7) {};

\node (a) at (-1.5,1.4) {};

\node (b) at (-1.8,0.7) {};

\node (y) at (-1,-0.7) {};

\node (c) at (-1.8,-0.7) {};
\node[fill=none, draw=none, scale=2.5] (cT) at (-2.1,-0.7) {$c$};

\node (d) at (-1.5,-1.4) {};

\node (z) at (1.2,-0.5) {};

\node (g) at (2.1,-0.2) {};

\node (h) at (1, 0.7) {};

\node (w) at (1.8, -1) {};

\node (e) at (2.4, -1.6) {};
\node[fill=none, draw=none, scale=2.5] (eT) at (2.7,-1.7) {$e$};

\node (f) at (2.5, -0.9) {};
\node[fill=none, draw=none, scale=2.5] (fT) at (2.8,-0.9) {$f$};

\path (u) edge (v);
\path (u) edge (x);
\path (x) edge (a);
\path (x) edge (b);
\path (u) edge (y);
\path (y) edge (d);
\path (y) edge (c);
\path (v) edge (z);
\path (z) edge (g);
\path (v) edge (h);
\path (z) edge (w);
\path (w) edge (e);
\path (w) edge (f);

\node[fill=none, draw=none, scale=4] (ar) at (0,-2.2) {$\downarrow$};
\end{tikzpicture}
\vspace{0.3em}
\end{minipage}
&
\begin{minipage}{.32\linewidth}
\begin{tikzpicture}[scale=0.8, every node/.style={draw, circle, fill=black, scale=0.4}]
\node (u) at (-0.5,0) {};
\node[fill=none, draw=none, scale=2.5] (uT) at (-0.4,0.3) {$u_3$};

\node (v) at (0.5,0) {};
\node[fill=none, draw=none, scale=2.5] (vT) at (0.4,0.3) {$v_3$};

\node (x) at (-1,0.7) {};

\node (a) at (-1.5,1.4) {};

\node (b) at (-1.8,0.7) {};

\node (y) at (-1,-0.7) {};

\node (c) at (-1.8,-0.7) {};

\node (d) at (-1.5,-1.4) {};
\node[fill=none, draw=none, scale=2.5] (dT) at (-1.8,-1.5) {$d$};

\node (z) at (1.2,-0.5) {};

\node (g) at (2.1,-0.2) {};

\node (h) at (1, 0.7) {};
\node[fill=none, draw=none, scale=2.5] (hT) at (1.3,0.8) {$h$};

\node (w) at (1.8, -1) {};

\node (e) at (2.4, -1.6) {};

\node (f) at (2.5, -0.9) {};

\path (u) edge (v);
\path (u) edge (x);
\path (x) edge (a);
\path (x) edge (b);
\path (u) edge (y);
\path (y) edge (d);
\path (y) edge (c);
\path (v) edge (z);
\path (z) edge (g);
\path (v) edge (h);
\path (z) edge (w);
\path (w) edge (e);
\path (w) edge (f);
\end{tikzpicture}
\end{minipage}
\end{tabular}

\begin{tikzpicture}[every node/.style={draw, circle, fill=black, scale=0.4}]
\node (cen) at (0,0) {};
\node[fill=none, draw=none, scale=2.5] (cenT) at (0.2,0.2) {$t$};

\node (x) at (0,0.7) {};
\node[fill=none, draw=none, scale=2.5] (xT) at (0,1) {$w_1$};

\node (y) at (-0.6,-0.6) {};
\node[fill=none, draw=none, scale=2.5] (yT) at (-0.7,-0.3) {$w_2$};

\node (w) at (0.6,-0.6) {};
\node[fill=none, draw=none, scale=2.5] (wT) at (0.7,-0.3) {$w_3$};

\node (v) at (-0.7,1) {};

\node (g) at (0.7,1) {};
\node[fill=none, draw=none, scale=2.5] (gT) at (1,1.1) {$g$};

\node (a) at (-1.4,1.2) {};
\node[fill=none, draw=none, scale=2.5] (aT) at (-1.7,1.3) {$a$};

\node (b) at (-1.4,0.5) {};
\node[fill=none, draw=none, scale=2.5] (bT) at (-1.7,0.5) {$b$};

\node (c) at (-1.4,-0.4) {};
\node[fill=none, draw=none, scale=2.5] (cT) at (-1.7,-0.4) {$c$};

\node (z) at (-1,-1.2) {};

\node (e) at (-1.8,-1.1) {};
\node[fill=none, draw=none, scale=2.5] (eT) at (-2.1,-1.1) {$e$};

\node (f) at (-1.6,-1.7) {};
\node[fill=none, draw=none, scale=2.5] (fT) at (-1.9,-1.7) {$f$};

\node (d) at (1.5,-1) {};
\node[fill=none, draw=none, scale=2.5] (dT) at (1.8,-1) {$d$};

\node (h) at (1.5,-0.1) {};
\node[fill=none, draw=none, scale=2.5] (hT) at (1.8,-0.1) {$h$};

\path (cen) edge (x);
\path (cen) edge (y);
\path (cen) edge (w);
\path (x) edge (g);
\path (x) edge (v);
\path (v) edge (a);
\path (v) edge (b);
\path (y) edge (c);
\path (y) edge (z);
\path (z) edge (e);
\path (z) edge (f);
\path (w) edge (h);
\path (w) edge (d);

\node[fill=none,draw=none] (centerer) at (2.8,0) {};
\end{tikzpicture}
\caption{Example of the refinement operation. A branch decomposition $(T, L)$ on a set $V = \{a,b,c,d,e,f,g,h\}$ (top).  For a tripartition
 $(C_1=\{a,b,g\}, C_2=\{c,e,f\}, C_3=\{d,h\})$, we have 
 the partial branch decompositions $(T_1, L_1) = (T,L\restriction_{\{a,b,g\}})$, $(T_2, L_2) = (T,L \restriction_{\{c,e,f\}})$, and $(T_3, L_3) = (T,L\restriction_{ \{d,h\}})$ (middle), and the refinement of $(T,L)$ with $(uv, C_1, C_2, C_3)$ (bottom).\label{fig:refex}}
\end{figure}
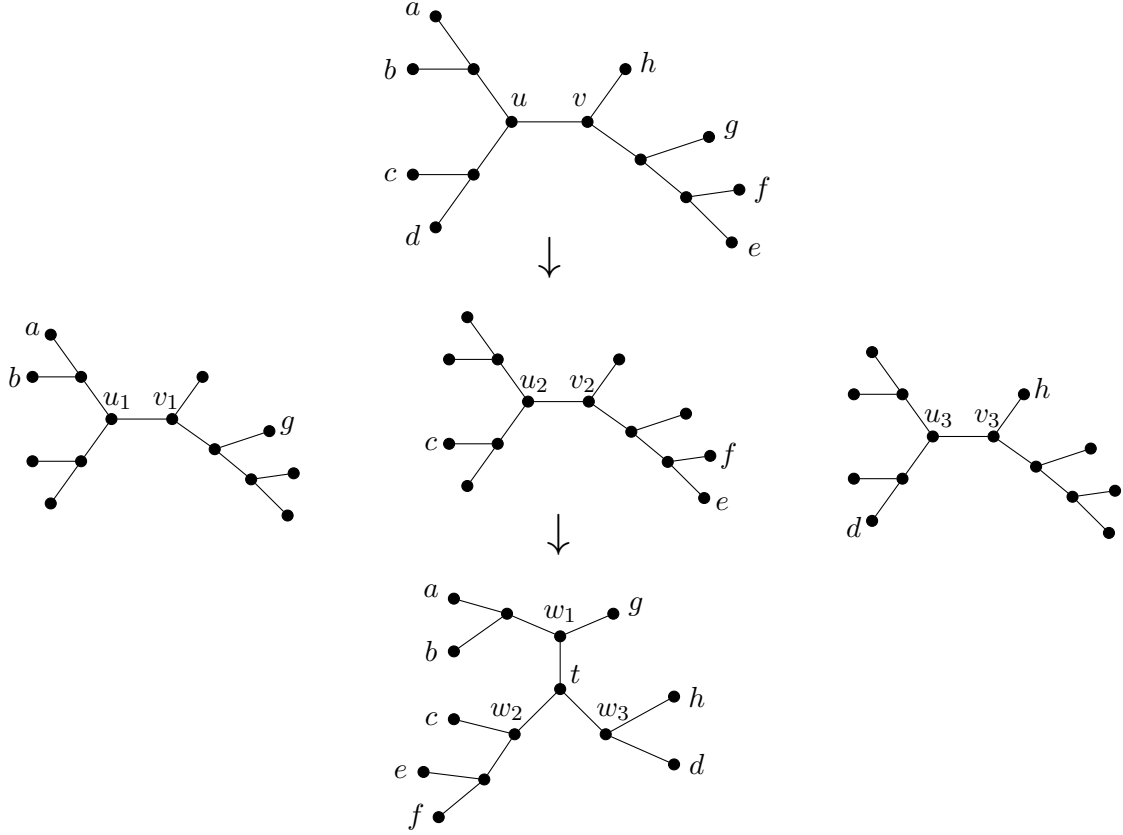



Let $(T,L)$ be a branch decomposition of a connectivity function $f : 2^V \rightarrow \mathbb{Z}_{\ge 0}$.
Let  $uv$ be an edge of $T$ and $(C_1, C_2, C_3)$ be a tripartition of $V$.  (One of the sets $C_i$ could be empty).
We denote by $(T,L\restriction_{C_i})$ the partial branch decomposition that has the same tree $T$ as $(T,L)$, but the mapping $L\restriction_{C_i}$ is a restriction of mapping $L$ to $C_i$. 
(We say that a partial branch decomposition is a branch decomposition where the labeling function $L$ is only required to be an injection.)

The \emph{refinement $(T',L')$ of $(T,L)$ with $(uv, C_1, C_2, C_3)$} is obtained by first taking the partial branch decompositions $(T_i, L_i) = (T, L\restriction_{C_i})$ for each $i \in \{1,2,3\}$.
Let $u_i v_i$ be the copy of the edge $uv$ in $T_i$.
The partial branch decompositions $(T_1, L_1)$, $(T_2, L_2)$, and $(T_3, L_3)$ are combined into a new partial branch decomposition by inserting a new node $w_i$ on each edge $u_i v_i$, and then connecting the nodes $w_1$, $w_2$, and $w_3$ to a new center node $t$.
Finally, the obtained partial branch decomposition is transformed into a branch decomposition by iteratively pruning leaves that are not labeled and suppressing degree-2 nodes. (See also \Cref{fig:refex}.)

How the  widths of edges change after the refinement?
First, note that if $C_i$ is non-empty, then there will be an edge $tw_i\in E(T')$ corresponding to a bipartition $(C_i, \overline{C_i})$ in the refinement, see \Cref{fig:refine}.  (For $X\subseteq V$, we use $\overline{X}$ to denote $V\setminus X$.) Thus the  width of $tw_i$ will be at least $f(C_i)$.
Let $(W, \oW)$ be the bipartition corresponding to the edge $uv$ in $T$.
Then in the refinement there will be edges $u_iw_i$ and $v_iw_i$ corresponding to bipartitions $(C_i \cap W, \overline{C_i \cap W})$ and $(C_i \cap \oW, \overline{C_i \cap \oW})$ for each $i$, provided that they are non-empty.
Therefore, the width of the refinement will be at least $\max(f(C_i \cap W), f(C_i \cap \oW))$.
Our first result is that if $f(uv) > 2\bw(f)$, then there exists a refinement that locally improves the branch decomposition around the edge $uv$.

\begin{figure} 
\begin{center}
\includegraphics[scale=4]{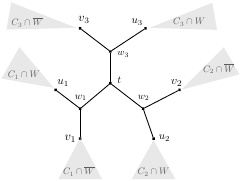} 
\caption{Changes of the width in a ``neighborhood'' of $uv$.}\label{fig:refine}
\end{center}
\end{figure}

\begin{theorem}
\label{the:comb_wimpr}
Let $f : 2^V \rightarrow \mathbb{Z}_{\ge 0}$ be a connectivity function.
If $W \subseteq V$ is a set with $f(W) > 2 \bw(f)$, then there exists a tripartition $(C_1, C_2, C_3)$ of $V$ so that for each $i \in \{1,2,3\}$ it holds that $f(C_i) < f(W) / 2$, $f(C_i \cap W) < f(W)$, and $f(C_i \cap \oW) < f(W)$.
\end{theorem}

We call such a tripartition $(C_1, C_2, C_3)$ satisfying $f(C_i) < f(W)/2$, $f(C_i \cap W) < f(W)$, and $f(C_i \cap \oW) < f(W)$, a \emph{\splitW}, signifying that the refinement operation enforced by this tripartition and the edge $uv$ corresponding to $(W, \oW)$ improves the branch decomposition around the edge $uv$. However, we cannot guarantee that   the new decomposition $T'$   would be better  than $T$. 
The reason is that  \splitW  can increase the widths of edges in branches $(T_i, L_i)$. It is a   non-trivial statement that the existence of a \splitW on a heavy edge 
implies the existence of a refinement that 
is better than $(T,L)$. More formally, 

\begin{restatable}{theorem}{combmaintheorem}
\label{the:comb_gloimpr}
Let $f : 2^V \rightarrow \mathbb{Z}_{\ge 0}$ be a connectivity function and $(T, L)$ be a branch decomposition of $f$ of width $k$, having $h\geq 1$ edges of width $k$.
Let $uv$ be an edge of $(T, L)$ corresponding to a partition $(W, \oW)$ and having width $f(uv) = k$.
If there exists a \splitW, then there exists a \splitW $(C_1, C_2, C_3)$ such that the refinement of $(T,L)$ with $(uv, C_1, C_2, C_3)$ has width at most $k$ and less than $h$ edges of width $k$.
\end{restatable}

The combination of \Cref{the:comb_wimpr,the:comb_gloimpr} shows that when we consider a heavy edge $uv$,
we either prove that $k = f(uv) \le 2 \bw(f)$ (if there is no  a \splitW then  $f(uv) \le 2 \bw(f)$), 
or that there exists a \splitW which can be used to globally improve $(T,L)$ by reducing the number of edges of width $k$.
We refer to a \splitW $(C_1,C_2,C_3)$ whose existence is guaranteed by \Cref{the:comb_gloimpr} as a \emph{\strongminsplit}.
The  statement of  \Cref{the:comb_gloimpr} is  non-constructive---it does not provide a clue how to compute a \strongminsplit. 
To make  \Cref{the:comb_gloimpr} constructive, we show that 
a \strongminsplit
can be found by selecting a \splitW that optimizes four explicit criteria over all \splitWs. 
More precisely, let $(C_1, C_2, C_3)$ be a \splitW.   
The first three conditions that $(C_1, C_2, C_3)$ optimizes are 
\begin{enumerate}
\item[$(i)$]  $\max\{f(C_1), f(C_2), f(C_3)\}\to \min$, 
\item[$(ii)$] subject to ($i$), the number of non-empty sets $C_i$ in  $(C_1, C_2, C_3)$ is minimum, and 
\item[$(iii)$]  subject to ($i$) and ($ii$), $f(C_1)+f( C_2)+ f( C_3) \to \min$.
\end{enumerate}
(In the proof we refer to a \splitWs satisfying $(i)$--$(iii)$ as to the \minsplitW.)

To state the fourth property we need the following  definition. 
Let $r=uv \in E(T)$ be the edge of $T$ corresponding to the partition $(W,\oW)$. 
By slightly abusing the notation, let us view $r$ as the root of $T$. Then for every node $w$ of $T$, we define $T_r[w]$ as the set of leaves of the subtree of $T$ rooted in $w$. We say that a   tripartition $(C_1, C_2, C_3)$ \emph{intersects}   a node $w\in V(T)$ if 
$T_r[w]$ contains elements from at least two sets $C_i$, $i\in \{1,2,3\}$. 
\begin{enumerate}
\item[$(iv)$]  subject to ($i$), ($ii$), and ($iii$), $(C_1, C_2, C_3)$ intersects the minimum number of nodes of $T$. 
\end{enumerate}

The main result of \Cref{sec_comb_results} is \Cref{the:glostro}. By this theorem,  any  \splitW satisfying $(i)$--$(iv)$ is a  \strongminsplit.
The properties  $(i)$--$(iv)$ are very useful from the algorithmic perspective: For several important classes of function $f$,  \splitW{s}   with properties  $(i)$--$(iv)$   can be computed by making use of dynamic programming on the branch decomposition $T$.

The combinatorial characterization of improving refinements brings us to the following generic algorithm for 
  2-approximating branchwidth.
The algorithm repeatedly takes a heavy edge $uv$ corresponding to a partition $(W, \oW)$.  Then it checks if there exists a \splitW, and if not, returns $(T,L)$.  If a \splitW exists,  the algorithm refines $(T,L)$ using a \strongminsplit.
In the next subsection we outline how to implement this generic algorithm  efficiently for branch decompositions that support efficient dynamic programming for finding \strongminsplits.

\subsection{Algorithmic framework}
In order to efficiently apply refinement operations to improve a branch decomposition, we need to find \strongminsplits efficiently.
For this, the main idea is that finding a \strongminsplit corresponds to a partitioning problem minimizing $(i)$--$(iv)$. Such type of problems can be solved efficiently  on branch decompositions that support dynamic programming, like rank decompositions or  branch decompositions of graphs. The application of  sequence of at most $n$ refinement operations guarantees that the width of the resulting decomposition would decrease by at least one. 
Thus starting with a branch decomposition of width $k$ and 
  repeatedly using dynamic programming to find a \strongminsplit in $\OO(f(k) \cdot n)$ time  would result in a total time complexity of $\OO(f(k) \cdot k \cdot n^2)$ to reduce the width down to $2\bw(f)$.
Of course, this is not sufficient to obtain our applications. Let us remind that we target a linear time algorithm for the branchwidth of a graph. For rankwidth, where we target for a quadratic running time, we also have to implement iterative compression,  resulting in cubic running time.


The algorithmic ingredient of our general framework strongly uses the combinatorial  properties of 
  \strongminsplits:  It appears, that  refinements with \strongminsplits ``interplay well'' with dynamic programming.  This allows us to implement a sequence of refinement operations that reduces the width from $k$ to $k-1$ in a total  of $t(k)  2^{\OO(k)}  n$ time, where 
   $t(k)$ depends on the time complexity of dynamic programming on the branch decomposition. This part is inspired by the amortization techniques from \cite{Korhonen21}.

We exploit the fact that when we refine a branch decomposition $(T,L)$ on an edge $uv$, only some connected subtree of $T$ around $uv$ changes. 
We call the nodes in this connected subtree the edit set of the refinement operation.
In particular, we show that a refinement operation can be implemented by removing its edit set $R$ and inserting a connected subtree of $|R|$ new nodes in its place.
Then we define a potential function $\phi$ of a branch decomposition so that $\phi(T,L)$ is initially bounded by $2^{\OO(k)} n$, where $k$ is the width of $(T,L)$.
We show that when refining with a \strongminsplit, a refinement operation with edit set $R$ decreases $\phi(T,L)$ by at least $|R|$, and therefore the total sum of $|R|$ across a series of refinement operations is bounded by $2^{\OO(k)} n$.

We exploit the bound on the sizes of edit sets by only re-computing the dynamic programming tables of the nodes in the edit set and re-using the rest.
We formalize the notion of performing dynamic programming operations in time $t(k)$ with a definition of a ``refinement data structure'' in \Cref{subsec:refds}.
This leads to the following theorem.

\begin{restatable}{theorem}{themainalg}
\label{the:main_alg}
Let $f$ be a connectivity function for which there exists a refinement data structure with time complexity $t(k)$.
There is an algorithm that given a branch decomposition $(T,L)$ of $f$ of width $k$, in time $t(k) 2^{\OO(k)} n$ either outputs a branch decomposition of $f$ of width at most $k-1$, or correctly concludes that $k \le 2\bw(f)$.
\end{restatable}


\subsection{Applications}

\paragraph{Rankwidth.}
%
We show that for rank decompositions, the 
refinement data structure from \Cref{the:main_alg}
%
can be implemented in $t(k) = 2^{2^{\OO(k)}}$ time, yielding a $2^{2^{\OO(k)}} n$ time implementation of the algorithm in 
 \Cref{the:main_alg} for rankwidth.
To obtain \autoref{theorem_rw_algo}, we then apply this algorithm with the technique of iterative compression of Reed, Smith, and Vetta \cite{ReedSmithVetta-OddCycle}, see also \cite[Chapter~4]{cygan2015parameterized}.  In particular,  by iteratively obtaining a rank decomposition of $G$ of width $\le 2 \rw(G) + 1$ from a rank decomposition of $G \setminus \{v\}$ of width $\le 2 \rw(G \setminus \{v\})$, and then applying  \Cref{the:main_alg} to decrease the width or to prove that it is already within a factor of 2 from the optimum.

A crucial detail of this iterative compression and the implementation of  \Cref{the:main_alg} 
 for rankwidth is that we always maintain an \emph{augmented rank decomposition}, i.e., a rank decomposition that for each edge corresponding to a partition $(A, B)$ of $V(G)$ stores representatives of the twin classes of the bipartite graph $G[A, B]$.
In particular, the implementation of \Cref{the:main_alg} for rank decompositions requires an augmented rank decomposition as an input and outputs an augmented rank decomposition.
The augmentation is required for implementing dynamic programming on rank decompositions, and the reason for maintaining the augmentation throughout the iterative compression is that we are not aware of any $\OO(f(k) n)$ time algorithms for making a given rank decomposition augmented.
Instead, in every step of the iterative compression we use a $2^{\OO(k)} n$ time algorithm for obtaining an augmented rank decomposition of $G$ from an augmented rank decomposition of $G \setminus \{v\}$.

\paragraph{Branchwidth of a graph.}
We show that for branch decompositions of graphs, the
refinement data structure from \Cref{the:main_alg}
can be implemented in
 $t(k) = 2^{\OO(k)}$ time, thus yielding a $2^{\OO(k)} n$ time implementation of the algorithm   for branchwidth of a graph.
By using a 2-approximation algorithm for treewidth with time complexity $2^{\OO(k)} n$~\cite{Korhonen21}, we can assume that we initially have a branch decomposition of width at most $3k$, where $k$ is the branchwidth.
We then obtain \autoref{theorem_vw_algo} by $k$ applications of the algorithm of \Cref{the:main_alg}.
We note that this algorithm can be made more self-contained, to only depend on the compression technique of Bodlaender~\cite{Bodlaender96} instead of treewidth algorithms, by applying the compression technique of Bodlaender in the same manner as it is applied in~\cite{BodlanderDDFLP13,Korhonen21}.

\section{Notation}\label{sec_notations}
Let $G$ be a graph, its vertices are $V(G)$ and edges $E(G)$.
For a vertex $v$, $N_G(v)$ denotes the neighborhood of $v$ in $G$, and for a vertex set $X$, $N_G(X) = \bigcup_{v \in X} N_G(v) \setminus X$.
We drop the subscript $G$ when it is clear from the context.
For a graph $G$ and a set $X \subseteq V(G)$, we denote by $G[X]$ the subgraph induced by $X$.
Let $v$ be a degree-2 vertex in $G$, with neighbors $N(v) = \{u, w\}$.
The suppression of $v$ deletes the vertex $v$ and edges incident to it and inserts the edge $uw$ if it does not exists.

For a graph $G$ and a pair $A,B$ of disjoint subsets of $V(G)$, $G[A,B]$ is the bipartite graph induced by edges between $A$ and $B$ and labeled with the bipartition $(A,B)$.
The notation $M_G[A,B]$ denotes the $|A| \times |B|$ 0-1 matrix representing $G[A,B]$.
All matrices in this paper are 0-1 matrices, and the GF(2) rank of matrix $M$ is denoted by $\rank(M)$.

A bipartition of a set $V$ is an ordered pair $(X,Y)$ of disjoint sets so that $X \cup Y = V$.
A tripartition of a set $V$ is an ordered triple $(X,Y,Z)$ of disjoint sets so that $X \cup Y \cup Z = V$.
Note that in our definitions a partition can contain empty parts.
We say that a tripartition $(X,Y,Z)$ of $V$ intersects a set $A \subseteq V$ if at least two of $X,Y,Z$ have a non-empty intersection with $A$.

Let $V$ be a finite set and $f : 2^V \rightarrow \mathbb{Z}_{\ge 0}$ be a function. Function 
$f$ is \emph{submodular} if 
\[
f(U\cup W)+f(U\cap W) \leq f(U)+f(W)
\]
for all $U,W\subseteq V$. 
 For a  set $W \subseteq V$,  we denote $\overline{W} = V \setminus W$. A function $f$ is \emph{symmetric} if for every $W \subseteq V$, $f(W)=f(\overline{W})$.  Function $f$ is a \emph{connectivity function} if (i) $f(\emptyset) = 0$, (ii) $f$  is symmetric,  and (iii) $f$ is submodular.

A cubic tree is a tree where each node has degree 1 or 3. 
A \emph{branch decomposition} of a connectivity function $f : 2^V \rightarrow \mathbb{Z}_{\ge 0}$ is a pair $(T, L)$, where  $T$ is a cubic tree and $L$ is a bijection mapping  $V$ to the leaves of $T$. If $|V|\leq 1$, then $f$ admits no branch decomposition. Very often, by slightly abusing the notation, we will treat the leaves of the branch decomposition as elements of $V$. Thus usually we do not distinguish a set of leaves $W$ of $T$ and the subset $L^{-1}(W)$ of $V$, and we usually denote a branch decomposition by only the tree $T$ instead of the pair $(T,L)$.

Let $e=uv \in E(T)$ be an edge of a branch decomposition $T$.
We denote by $T[uv] \subseteq V$ the set of leaves of $T$ that are closer to the node $u$ than to the node $v$, and by $T[vu] \subseteq V$ the set of leaves that are closer to $v$ than $u$.
Hence $(T[uv], T[vu])$ is the bipartition of $V$ induced by the connected components of the forest $T-e$ obtained from $T$ by deleting   $e$.
The \emph{width} of an edge $e=uv$ of $T$ is $f(e)=f(T[uv])=f(T[vu])$, and the \emph{width of a branch decomposition $(T,L)$}, denoted by $\bw(T,L)$, is a maximum width of an edge. Finally, the   \emph{branchwidth} of $f$, denoted by $\bw(f)$, is the minimum width of a branch decomposition of $f$.

Let $r \in E(T)$ be an edge of a branch decomposition $T$.
We introduce some notation with the intuition that $T$ is rooted at the edge $r$.
The $r$-subtree of a node $w \in V(T)$ is the subtree of $T$ induced by nodes $x \in V(T)$ such that $w$ is on the unique $x-r$ path (including $w$).
For any node $w \in V(T)$, we denote by $T_r[w] \subseteq V$ the set of leaves in the $r$-subtree of $w$.
Note that when $r = uv$, it holds that $T_r[u] = T[uv]$, $T_r[v] = T[vu]$, and for any $w \in V(T)$ it holds that either $T_r[w] \subseteq T_r[u]$ or $T_r[w] \subseteq T_r[v]$.
The $r$-parent of a node $w \in V(T) \setminus \{u,v\}$ is the node next to $w$ on the unique $w-r$ path.
If $p$ is the $r$-parent of $w$, then $w$ is an $r$-child of $p$.
Note that every non-leaf node has exactly two $r$-children, and every leaf node has no $r$-children.
A set $X \subseteq V$ $r$-intersects a node $w \in V(T)$ if $X \cap T_r[w]$ is non-empty.
Similarly, a tripartition $(C_1, C_2, C_3)$ of $V$ $r$-intersects a node $w \in V(T)$ if it intersects the set $T_r[w]$.


\section{Combinatorial results}\label{sec_comb_results}
In this section we prove our combinatorial results.
Throughout this section we use the convention that $f : 2^V \rightarrow \mathbb{Z}_{\ge 0}$ is a connectivity function.

\subsection{Refinement operation}
The central concept of our results is the refinement operation.
Before defining it, we need the definition of a \emph{partial branch decomposition}.

\begin{definition}[Partial branch decomposition]
A partial branch decomposition on a set $C$ is a pair $(T,L)$, where $T$ is a cubic tree and $L$ is an injection from $C$ to the leaves of $T$.
\end{definition}

Let $(T,L)$ be a branch decomposition of $f$ and let $C_i \subseteq V$.
We denote by $(T, L \restriction_{C_i})$ the partial branch decomposition on the set $C_i$ obtained by restricting the mapping $L$ to only $C_i$.
Now we can define the refinement operation (see also \autoref{fig:refex}).

\begin{definition}[Refinement]\label{def_refinement}
Let $(T,L)$ be a branch decomposition of $f$, $uv \in E(T)$ an edge of $T$, and $(C_1, C_2, C_3)$ a tripartition of $V$.
We define the refinement of $T$ with $(uv, C_1, C_2, C_3)$.

For each $i \in \{1,2,3\}$, let $(T_i, L_i) = (T, L \restriction_{C_i})$, and let $u_i v_i$ be the copy of the edge $uv$ in $T_i$.
Now, let $(T', L')$ be a partial branch decomposition on $V$ obtained by first inserting a node $w_i$ on the edge $u_i v_i$ of each $T_i$ (i.e. $V(T_i) \gets V(T_i) \cup \{w_i\}$ and $E(T_i) \gets E(T_i) \cup \{u_i w_i, w_i v_i\} \setminus \{u_i v_i\}$), then taking the disjoint union of $(T_1, L_1)$, $(T_2, L_2)$, $(T_3, L_3)$, and then inserting a node $t$ adjacent to $w_1$, $w_2$, $w_3$, connecting the disjoint union into a tree.
Finally, the refinement is obtained by simplifying $(T', L')$ by iteratively removing degree-1 nodes that are not labeled by $L'$, and suppressing degree-2 nodes.
\end{definition}

We observe that if $T'$ is a refinement of $T$ with $(r, C_1, C_2, C_3)$, then every edge of $T'$ either corresponds to a bipartition $(C_i, \overline{C_i})$ or to a bipartition $(T_r[w] \cap C_i, \overline{T_r[w] \cap C_i})$, where $w \in V(T)$.
More formally as follows.

\begin{observation}\label{observ_refinement}
Let $T$ be a branch decomposition, $r \in E(T)$, and $(C_1, C_2, C_3)$ a tripartition of $V$.
Let $T'$ be the refinement of $T$ with $(r, C_1, C_2, C_3)$.
It holds that
\begin{eqnarray*}
&\{\{T'[u'v'], T[v'u']\} \mid u'v' \in E(T')\} = \\
&\bigcup_{i \in \{1,2,3\}} \left( \{\{C_i, \overline{C_i}\}\} \cup \{\{T_r[w] \cap C_i, \overline{T_r[w] \cap C_i}\} \mid w \in V(T)\} \right) \setminus \{\{\emptyset, V\}\}.
\end{eqnarray*}
\end{observation}

Note that as a branch decomposition is a cubic tree, no two edges correspond to the same cut, i.e., the set $\{\{T[uv], T[vu]\} \mid uv \in E(T)\}$ has cardinality $|E(T)|$.

\subsection{Improving refinement\label{subsec:refi_intro}}
Now we provide the formal definitions and results of our combinatorial framework, postponing the proofs to \Cref{subsec:symsubmodlemmas,subsec:thefind,sec:comb_end}.
First we define a \splitW of a set $W \subseteq V$.

\begin{definition}[\splitW]
Let $W \subseteq V$.
A tripartition $(C_1, C_2, C_3)$ of $V$ is a \emph{\splitW}  
 if for every $i\in\{1,2, 3\}$, 
 \begin{itemize}
 \item 
 $f(C_i) < f(W)/2$, 
 \item $f(C_i \cap W) < f(W)$, and 
 \item $f(C_i \cap \overline{W}) < f(W)$. 
 \end{itemize}
The \emph{width of a \splitW} is $\max\{f(C_1), f(C_2), f(C_3)\}$, the \emph{sum-width} of a \splitW is $f(C_1) + f(C_2) + f(C_3)$, and the \emph{arity} of a \splitW is the number of non-empty parts in the partition.
\end{definition}
Let us note that  because of the condition  $f(C_i \cap W) < f(W)$, we have that for every $i\in\{1,2, 3\}$, $C_i\neq V$. Thus the arity of  every \splitW is always 2 or 3.
Also, \splitWs are symmetric in the sense that the ordering of the sets 
  $C_1, C_2, C_3$ as well as  replacing $W$ by $\overline{W}$,  do not change any  properties of a \splitW.

The following lemma is a re-statement of \autoref{the:comb_wimpr}.

\begin{lemma}
\label{the:find}
If $W \subseteq V$ such that $f(W) > 2\bw(f)$, then there exists a \splitW.
\end{lemma}

We postpone the proof of \autoref{the:find} to \Cref{subsec:thefind}.

If $uv$ is an edge of a branch decomposition $T$ with $(T[uv], T[vu]) = (W, \overline{W})$, then refining with $(uv, C_1, C_2, C_3)$ where $(C_1, C_2, C_3)$ is a \splitW locally improves the branch decomposition around the edge $uv$.
Towards proving that the existence of a \splitW implies the existence of a \splitW that globally improves the branch decomposition, we define \minsplitW.

\begin{definition}[\MinsplitW] For $W\subseteq V$, we say that 
a \splitW $(C_1, C_2, C_3)$  is  \emph{minimum} if
\begin{enumerate}
\item[(1)] $(C_1, C_2, C_3)$ is of minimum width among all \splitWs,
\item[(2)] subject to (1), $(C_1, C_2, C_3)$ is of minimum arity, and
\item[(3)]  subject to (1) and (2), $(C_1, C_2, C_3)$ is of minimum sum-width.
\end{enumerate}
\end{definition}

Let $uv = r$ be a an edge of a branch decomposition $T$.
Recall that for any node $w$ of $T$ it holds that either $T_r[w] \subseteq T[uv]$ or $T_r[w] \subseteq T[vu]$, and recall that by \autoref{observ_refinement}, when refining with $(uv, C_1, C_2, C_3)$, an edge corresponding to $T_r[w]$ will turn into edges corresponding to $T_r[w] \cap C_1$, $T_r[w] \cap C_2$, and $T_r[w] \cap C_3$.
The following lemma shows that when refining with a \minsplitW, the width of a branch decomposition does not increase.

\begin{lemma}
\label{the:glo}
Let $W\subseteq V$ and let $(C_1, C_2, C_3)$ be a \minsplitW. 
Then for each set  $W'$ such that $W' \subseteq W$ or $W' \subseteq \overline{W}$ and every  $i \in \{1,2,3\}$, it holds that $f(W' \cap C_i) \le f(W')$.
\end{lemma}

We postpone the proof of \autoref{the:glo} to \Cref{sec:comb_end}.

Next we define a \strongminsplit, which will be the type of refinement that we really use in our results.

\begin{definition}[\Strongminsplit]
Let $T$ be a branch decomposition, and $r = uv \in E(T)$ an edge of $T$.
A $T$-improvement on $r$ is a tuple $(r, C_1, C_2, C_3)$, where $W = T[uv]$, and $(C_1, C_2, C_3)$ is a \minsplitW.
We say that a $T$-improvement on $r$ \emph{intersects} a node $w \in V(T)$ if $(C_1, C_2, C_3)$ $r$-intersects it, i.e., the set $T_r[w]$ intersects at least two sets from $\{C_1, C_2, C_3\}$.
A $T$-improvement $(r, C_1, C_2, C_3)$ is a \emph{\strongminsplit} if it intersects the minimum number of nodes of $T$ among all $T$-improvements on $r$.
\end{definition}

The following theorem  is our main combinatorial result.
In particular, \autoref{the:comb_gloimpr} will be a straightforward application of it.

\begin{restatable}{theorem}{combmainlemma}
\label{the:glostro}
Let $T$ be a branch decomposition, $r \in E(T)$, and $(r, C_1, C_2, C_3)$ a \strongminsplit.
Then for every $i \in \{1,2,3\}$ and every node $w \in V(T)$, it holds that $f(T_r[w] \cap C_i) \le f(T_r[w])$. Moreover, if $T_r[w] \cap C_i \neq \emptyset$, then $f(T_r[w] \cap C_i) = f(T_r[w])$ if and only if $T_r[w] \subseteq C_i$.
\end{restatable}

We postpone the proof of \autoref{the:glostro} to \Cref{sec:comb_end}.
Let us prove \autoref{the:comb_gloimpr} using \autoref{the:glostro} and \autoref{observ_refinement}.

\combmaintheorem*
\begin{proof}
Denote $uv = r$.
Take a \strongminsplit $(r, C_1, C_2, C_3)$, which exists by the existence of a \splitW.
Let $T'$ be the refinement of $T$ with $(r, C_1, C_2, C_3)$.
By \autoref{observ_refinement}, each edge of $T'$ corresponds to a bipartition of form $(T_r[w] \cap C_i, \overline{T_r[w] \cap C_i})$, where $w \in V(T)$, or of form $(C_i, \overline{C_i})$.
Because $(C_1, C_2, C_3)$ is a \splitW, all edges of form $(C_i, \overline{C_i})$ have width $<k/2$.
By the first part of \autoref{the:glostro}, all edges of form $(T_r[w] \cap C_i, \overline{T_r[w] \cap C_i})$ have width $\le f(T_r[w]) \le k$, and therefore the width of $T'$ is at most $k$.

By the second part of \autoref{the:glostro},
 if $T'$ contains an edge of form $(T_r[w] \cap C_i, \overline{T_r[w] \cap C_i})$ with width $k$, then $T_r[w] \cap C_i = T_r[w]$, and thus the exact same edge also exists in $T$.
This gives an injective mapping from the edges of $T'$ of width $k$ to edges of $T$ of width $k$.
It remains to observe that by the definition of \splitW, this mapping maps no edge of $T'$ to the edge $uv$ of $T$, and therefore $T'$ contains strictly less edges of width $k$ than $T$.
\end{proof}

\subsection{Observations on connectivity functions\label{subsec:symsubmodlemmas}}
We need the following two observations.

\begin{observation}
\label{lem:compl}
For any $A_1, \ldots, A_n\subseteq V$, it holds that
  $f(A_1 \cap \cdots \cap A_n) = f(\overline{A_1} \cup \cdots \cup \overline{A_n})$.
\end{observation}
\begin{proof}
Note that $\overline{A_1 \cap \cdots \cap A_n} = \overline{A_1} \cup \cdots \cup \overline{A_n}$.
\end{proof}

\begin{observation} \label{lem:p1}For any $A,B\subseteq V$, it holds that
$f(A) \le f(A \cap B) + f(A \cap \overline{B}) \le f(A) + 2 f(B)$.
\end{observation}
\begin{proof}
For the first inequality, we have
\begin{eqnarray}f(A) = f((A \cap B) \cup (A \cap \overline{B})) \le f(A \cap B) + f(A \cap \overline{B}) - f(\emptyset).\label{eqobs2}
 \end{eqnarray}
For the second inequality, observe that
\begin{eqnarray*}
f(A \cap B) + f(A \cap \overline{B}) & \le &2f(A) + 2f(B) - f(A \cup B) - f(A \cup \overline{B}) \\
 &= &2f(A) + 2f(B) - f(\overline{A} \cap \overline{B}) - f(\overline{A} \cap B) \le 2f(A) + 2f(B) - f(A), 
 \end{eqnarray*}
where the last inequality follows from \eqref{eqobs2}.
\end{proof}

\subsection{Proof of \Cref{the:find}\label{subsec:thefind}}
The idea of the proof of \autoref{the:find} is that we take an optimal branch decomposition $T$ of $f$, i.e., a branch decomposition $T$ of width $\bw(f)$, and argue that either some bipartition corresponding to an edge of $T$ or some tripartition corresponding to a node of $T$ gives the \splitW.

We define an orientation of a set $C \subseteq V$ with respect to a set $W \subseteq V$.
This will be used for orienting the edges of the optimal branch decomposition based on $W$.

\begin{definition}[Orientation]
\label{def:orient}
Let $C,W\subseteq V$. We say that  
the set $W$ \emph{directly orients} $C$ if $f(C \cap W) < f(\overline{C} \cap W)$ and $f(C \cap \overline{W}) < f(\overline{C} \cap \overline{W})$.
The set $W$ \emph{ inversely orients} $C$ if it directly orients $\overline{C}$.
The set $W$ \emph{disorients} $C$ if it neither directly nor inversely orients $C$.
\end{definition}

Note that the definition of orienting is symmetric with respect to complementing $W$, i.e., $W$ (directly, inversely, dis)-orients $C$ if and only if $\overline{W}$ (directly, inversely, dis)-orients $C$.
Next we show that if there is an edge of an optimal branch decomposition that cannot be oriented according to \autoref{def:orient}, then it corresponds to a \splitW.

\begin{lemma}
\label{lem:disorient}
Let $C,W\subseteq V$. If 
  $W$ disorients $C$ and $f(C) < f(W)/2$, then $(C, \overline{C}, \emptyset)$ is a \splitW.
\end{lemma}
\begin{proof}
By possibly interchanging $C$ with  $\overline{C}$, without loss of generality    we can assume that $f(C \cap W) \le f(\overline{C} \cap W)$ and $f(C \cap \overline{W}) \ge f(\overline{C} \cap \overline{W})$. Since $f(C)= f(\overline{C}) < f(W)/2$, to show that $(C, \overline{C}, \emptyset)$ is a \splitW, 
it suffices to prove that $f(\overline{C} \cap W) < f(W)$ and $f(C \cap \overline{W}) < f(W)$.
Assume, by way of contradiction, that $f(\overline{C} \cap W) \ge f(W)$. Then by the submodularity of $f$, 
 $f(\overline{C} \cup W) \le f(C)$, hence $f(C \cap \overline{W}) \le f(C)$. Therefore,  $f(C \cap \overline{W}) + f(\overline{C} \cap \overline{W}) \le 2f(C) < f(W)$. But this contradicts  \autoref{lem:p1}.
The proof of  $f(C \cap \overline{W}) < f(W)$ is symmetric.
\end{proof}

Now, to prove \Cref{the:find}, we take an optimal branch decomposition $T$ of $f$ and orient each edge $uv$ with $(T[uv], T[vu]) = (C, \overline{C})$ towards $v$ if $W$ directly orients $C$ and towards $u$ if $W$ inversely orients $C$.
If no orientation can be found, then  \Cref{lem:disorient} shows that $(C, \overline{C}, \emptyset)$ is a \splitW and we are done.

\begin{claim}
No edge of $T$ is oriented towards a leaf.
\end{claim}
\begin{proof}
Suppose there is an edge of $T$ oriented towards a leaf $v \in V$.
This means that $V \setminus \{v\}$ directly orients $W$, implying that $f(W \setminus \{v\}) < f(W \cap \{v\})$ and $f(\oW \setminus \{v\}) < f(\oW \cap \{v\})$.
However, one of the sets $W \cap \{v\}$ or $\oW \cap \{v\}$ must be empty and therefore either $f(W \cap \{v\}) = 0$ or $f(\oW \cap \{v\}) = 0$, but $f$ cannot take values less than $0$, so we get a contradiction.
\end{proof}

Now, by walking in $T$ according to the orientation, we end up finding an internal node to which all adjacent edges are oriented towards.
The following lemma proves that this node indeed gives a \splitW, and therefore completes the proof of  \Cref{the:find}.

\begin{lemma}\label{lemma_3splitorientation}
Let $W \subseteq V$, and $(C_1, C_2, C_3)$ a tripartition of $V$ such that for each $i\in\{1, 2, 3\}$,  $f(C_i) < f(W)/2$   and $W$ directly orients $C_i$.
Then $(C_1, C_2, C_3)$ is a \splitW.
\end{lemma}
\begin{proof}
Since $W$ directly orients $C_i$, we have that 
 $f(C_i \cap W) < f(\overline{C}_i \cap W)$. By \autoref{lem:p1},  $f(C_i \cap W) + f(\overline{C}_i \cap W) \le f(W) + 2f(C_i)$. Therefore,  $f(C_i \cap W) \le f(W)/2 + f(C_i) < f(W)$.
The proof that  $f(C_i \cap \overline{W}) < f(W)$ is similar.
\end{proof}

\subsection{Proofs of \Cref{the:glo} and \Cref{the:glostro} \label{sec:comb_end}}

Let $W\subseteq V$ and let $(C_1, C_2, C_3)$ be a \minsplitW.
The main idea behind  the proofs of \Cref{the:glo} and \Cref{the:glostro} is that if $f(W' \cap C_1) \ge f(W')$ for some $W'\subseteq W$ or  $W'\subseteq \overline{W}$, then we prove that $(C'_1, C'_2, C'_3) = (C_1 \cup W', C_2 \setminus W', C_3 \setminus W')$ is also a \minsplitW. 
(And symmetrically for $C_2$ and $C_3$.)
In particular, this will be used to contradict the minimality of $(C_1, C_2, C_3)$ or the fact that $(r, C_1, C_2, C_3)$ is a \strongminsplit.

The proof is different for  \minsplitWs of arity 2   and  of  3 and we have to consider different cases.
The following lemma would be used in both cases.

\begin{lemma}
\label{lem:case1}
For $W\subseteq V$, 
let $(C_1, C_2, C_3)$ be a \splitW.
Suppose that for some $W' \subseteq W$,  $f(C_1 \cap W') \ge f(W')$.  Then for  $C'_1 = C_1 \cup W'$, it holds that
 $f(C'_1 \cap W) < f(W)$.
\end{lemma}
\begin{proof} 
%
\begin{align*}
    f(C'_1 \cap W)& = f((C_1 \cup W') \cap W)  =f(W' \cup (C_1 \cap W))&&  (W'\subseteq W)\\
 &\le  f(W') + f(C_1 \cap W) - f(W' \cap (C_1 \cap W)) && \text{(submodularity of $f$)}\\
      &=f(W') + f(C_1 \cap W) - f(W' \cap C_1)  &&    \\
     &\le  f(C_1 \cap W) < f(W).&&   \text{(def. of \splitW)}
\end{align*}

\end{proof}

Note that by the symmetry of \splitWs, \autoref{lem:case1} holds also for $C_2$ and $C_3$, and for $W' \subseteq \oW$.


\subsubsection{\splitWs of arity 2}
The following lemma completes the proof of \Cref{the:glo} for \splitWs of arity 2 (note symmetry).
It also sets up the proof of \Cref{the:glostro} for \splitWs of arity 2, which will be completed in \Cref{subsubsec:finish}.

\begin{lemma}
\label{lem:main2}
Let $W \subseteq V$ and $(C, \overline{C}, \emptyset)$ be a \minsplitW.  Then $f(C \cap W') \le f(W')$ for every 
 $W' \subseteq W$.
Moreover, if the equality $f(C \cap W') = f(W')$  holds, then $(C \cup W', \overline{C} \setminus W', \emptyset)$ is also a \minsplitW.
\end{lemma}
\begin{proof}
Let $(C, \overline{C}, \emptyset)$ be a \minsplitW.  Suppose that for some $W' \subseteq W$,  $f(C \cap W') \ge f(W')$. 
To prove the lemma, we show that in this case  
  $(C', \overline{C'}, \emptyset) = (C \cup W', \overline{C} \setminus W', \emptyset)$ is also a \splitW.  Moreover, when  
   $f(C \cap W') > f(W')$, the existence of \splitW  $(C', \overline{C'}, \emptyset)$  would contradict the minimality of $(C, \overline{C}, \emptyset)$. When $f(C \cap W') =f(W')$,   we show that $(C', \overline{C'}, \emptyset) $ is a \minsplitW.
  
Let us check that  $(C', \overline{C'}, \emptyset)$ satisfies all the conditions of a \splitW. 

First, by $f(C \cap W') \ge f(W')$ and submodularity of $f$, we have that 
 \begin{eqnarray}\label{eqar2}
 f(\overline{C'}) = f(C') = f(C \cup W') \le f(C).
 \end{eqnarray}
 Thus $f(\overline{C'}) = f(C') <f(W)/2$.

Then, we note that $f(C' \cap \overline{W}) = f(C \cap \overline{W}) < f(W)$ and $f(\overline{C'} \cap \overline{W}) = f(\overline{C} \cap \overline{W}) < f(W)$. By \autoref{lem:case1}, we have that 
 $f(C' \cap W) < f(W)$.
 Finally, 
\begin{align*}
 f(\overline{C'} \cap W) &= f(C' \cup \overline{W}) = f((C \cup W') \cup (C \cup \overline{W})) &&  \text{(by \autoref{lem:compl})}\\
 &\le f(C \cup W') + f(C \cup \overline{W}) - f((C \cup W') \cap (C \cup \overline{W})) && \text{(submodularity of $f$)}\\
 &= f(C \cup W') + f(C \cup \overline{W}) - f(C) &&\\
 &\le  f(C) + f(\overline{C} \cap W) - f(C) < f(W). &&  \text{(by \autoref{lem:compl} and \eqref{eqar2})} 
\end{align*}
This completes the proof that $(C', \overline{C'}, \emptyset)$ is a \splitW.

Now, if $f(C \cap W') > f(W')$, then by the submodularity of $f$,  $f(C') = f(C \cup W') < f(C)$. But this 
 contradicts the minimality of $(C, \overline{C}, \emptyset)$.
If $f(C \cap W') = f(W')$, then $f(C') \le f(C)$. In this case, since $(C', \overline{C'}, \emptyset)$ is a \splitW and $(C, \overline{C}, \emptyset)$ is a \minsplitW, we conclude that $(C', \overline{C'}, \emptyset)$  is also a \minsplitW.
\end{proof}

Note again that due to symmetry of \splitWs, \autoref{lem:main2} holds also for $\overline{C}$ and for $W' \subseteq \oW$.


\subsubsection{\splitWs of arity 3}
For \splitWs of arity 3, we will heavily exploit the minimality of arity, i.e., the condition on \minsplitWs that there should be no \splitW with smaller or equal width and with smaller arity.
We proceed with a sequence of auxiliary lemmas establishing properties of a \minsplitW of arity 3.

\begin{lemma}
\label{lem:3dir}
Let $(C_1, C_2, C_3)$ be a \minsplitW of  arity 3. Then for every 
  $i\in\{1,2,3\}$, $W$ directly orients $C_i$.
\end{lemma}
\begin{proof}
If $W$ disorients $C_i$, then by \autoref{lem:disorient},  $(C_i, \overline{C}_i, \emptyset)$ is a \splitW contradicting the minimality of $(C_1, C_2, C_3)$.
On the other hand,  $W$ cannot  inversely orient $C_i$. Indeed, the assumption $f(W \cap C_i) > f(W \cap \overline{C}_i)$ and $f(\overline{W} \cap C_i) > f(\overline{W} \cap \overline{C}_i)$ yields that $f(W \cap \overline{C}_i) < f(W)$ and $f(\overline{W} \cap \overline{C}_i) < f(W)$ by the fact that $(C_1, C_2, C_3)$ is a \splitW. In this case, again,  $(C_i, \overline{C}_i, \emptyset)$ is a \splitW contradicting the minimality of $(C_1, C_2, C_3)$.
\end{proof}

\begin{lemma}
\label{lem:bigside}
Let $(C_1, C_2, C_3)$ be a \minsplitW of arity 3. Then for every 
  $i\in\{1,2,3\}$, it holds that $f(W \cup C_i) > f(W)/2$ and $f(\overline{W} \cup C_i) > f(W)/2$.
\end{lemma}
\begin{proof}
By \autoref{lem:3dir}, $W$  directly orients $C_i$. Hence $f(\overline{W} \cap C_i) < f(\overline{W} \cap \overline{C}_i)$.
 By \autoref{lem:p1}, $f(W)=f(\overline{W})\leq f(\overline{W} \cap C_i) + f(\overline{W} \cap \overline{C}_i)$. Therefore, 
 $f(W \cup C_i) = f(\overline{W} \cap \overline{C}_i) > f(W)/2$.
%
The proof of $f(\overline{W}  \cup C_i) > f(W)/2$ is symmetric.
\end{proof}

The following lemma shows that any set $C_i$ in a \minsplitW is in some sense well-connected to any set $C_i \cup (W \cap C_j)$.

\begin{lemma}
\label{lem:linked}
Let $(C_1, C_2, C_3)$ be a \minsplitW of arity 3. Then for every 
 set $P$ and $i,j\in \{1,2,3\}$ such that  $C_i \subseteq P \subseteq C_i \cup (W \cap C_j)$,  it holds that $f(P) \ge f(C_i)$. 
\end{lemma}
\begin{proof}
For $i=j$ the lemma is trivial. For $i\neq j$,   without loss of generality, we can assume  that $i=1$ and $j=2$. 
Targeting towards a contradiction, let us assume that for some $P$ with  $C_1 \subseteq P \subseteq C_1 \cup (W \cap C_2)$, 
 $f(P) < f(C_1)$.
We claim that then $(C'_1, C'_2, C'_3) = (P, C_2 \setminus P, C_3)$ is a \splitW contradicting the minimality of $(C_1, C_2, C_3)$.

First, by our assumption,  $f(C'_1) = f(P) < f(C_1)<f(W)/2$. To verify that 
$f(C'_2)=f(C_2 \setminus P) = f(C_2 \cap \overline{P}) \leq  f(C_2)$, we observe that  $f(C_2 \cup \overline{P}) = f(C_2 \cup C_3) = f(C_1)$.  Then by the submodularity of $f$, 
\[
 f(C_2 \cap \overline{P})+ f(C_1)= 
f(C_2 \cap \overline{P}) + f(C_2 \cup \overline{P})
 \leq f(C_2)+f(P).
\]
 By our assumption, $f(P) < f(C_1)$. Thus $f(C'_2)=f(C_2 \cap \overline{P})<f(C_2)$.

It remains to show that $f(C'_i \cap W) < f(W)$ and 
 $f(C'_i \cap \overline{W}) < f(W)$.  Because $P \subseteq C_1 \cup (W \cap C_2)$, we have that $C'_1\cap \overline{W}=C_1\cap \overline{W}$ and $C'_2\cap \overline{W}=C_2\cap \overline{W}$. Thus $f(C'_i \cap \overline{W}) < f(W)$.
 
 To prove that $f(P \cap W) < f(W)$, first observe that $P \cup W = C_1 \cup W$. Then by \autoref{lem:bigside}, 
 \[
 f(P \cup W) = (C_1 \cup W)>\frac{f(W)}{2}.
 \]
By the submodularity of $f$, 
\[
f(P\cap {W}) + f(P \cup {W})
 \leq f(P)+f(W).
\]
Since $f(P) < f(C_1)<f(W)/2$, we have that  $f(C'_1 \cap W)=f(P\cap {W})<f(W)$.  

Similarly, to prove that   $f(C'_2 \cap W)<f(W)$, we note that $C'_2 \cup W = C_2 \cup W$.  Then by 
 \autoref{lem:bigside},  $f(C'_2 \cup W)>f(W)/2$. We already have proved that $f(C'_2)<f(C_2)$. Then by making use of the  submodularity of $f$, we have that  
 \[
\frac{f(W)}{2} +  f(C'_2 \cap W) <  f(C'_2 \cap W)+ f(C'_2 \cup W) \leq f(C'_2)+ f(W)<\frac{f(W)}{2}+ f(W).
 \]
 
This concludes the proof that $(C'_1, C'_2, C'_3)$ is a \splitW.  Finally, the width of $(C'_1, C'_2, C'_3)$ is at most the width of $(C_1, C_2, C_3)$, but its sum-width is strictly less. This contradicts the minimality of $(C_1, C_2, C_3)$. 
%
%
%
%
%
%
%
%
%
%
%
%
%
%
%
%
%
%
%
%
%
\end{proof}

The following lemma completes the proof of \autoref{the:glo} together with \autoref{lem:main2} (note symmetry).
It also sets up the proof of \autoref{the:glostro} for \splitWs of arity 3, which will be completed in \Cref{subsubsec:finish}.

\begin{lemma}
\label{lem:main3}
Let $(C_1, C_2, C_3)$ be a \minsplitW of arity $3$. Then for every
 $W' \subseteq W$,  $f(C_1 \cap W') \le f(W')$. Moreover, if the equality holds, then $(C_1 \cup W', C_2 \setminus W', C_3 \setminus W')$ is also a \minsplitW.
\end{lemma}
\begin{proof}
Let $(C_1, C_2, C_3)$ be a \minsplitW of arity $3$ and assume that for some $W' \subseteq W$, 
 $f(C_1 \cap W') \ge f(W')$. We define 
 $(C'_1, C'_2, C'_3) = (C_1 \cup W', C_2 \setminus W', C_3 \setminus W')$.
 First we will prove that $(C'_1, C'_2, C'_3)$ is a \splitW.  Due to the minimality of   $(C_1, C_2, C_3)$, this would imply that  
 $f(C_1 \cap W') > f(W')$ cannot occur.  Moreover, the condition  $f(C_1 \cap W') = f(W')$ would yield that 
 $(C'_1, C'_2, C'_3)$  is  a \minsplitW.

 We start the proof that  $(C'_1, C'_2, C'_3)$ is  \splitW  from showing that $f(C'_i) \le f(C_i)$ for all $i$.
The inequality 
\begin{eqnarray}\label{lemma_final}
f(C'_1) = f(C_1 \cup W') \le f(C_1)
\end{eqnarray}
follows directly from our assumption and the submodularity of $f$. Moreover,  \eqref{lemma_final} turns into an equality only if $f(C_1 \cap W') = f(W')$.

To prove that  
\begin{eqnarray}\label{lemma_finalB}
f(C'_2) \le f(C_2),
\end{eqnarray}  
we do the following
\begin{align*}
 f(C'_2) &= f(C_2 \setminus W') = f(C_2 \cap \overline{W'}) = f(\overline{C}_2 \cup W') = f(\overline{C}_2 \cup (C_1 \cup W'))&&\\
&\le f(C_2) + f(C_1 \cup W') - f(\overline{C}_2 \cap (C_1 \cup W'))&& \text{(submodularity of $f$)}\\
&\le f(C_2) + f(C_1) - f(\overline{C}_2 \cap (C_1 \cup W'))&&  \text{(by  \eqref{lemma_final})}\\
&= f(C_2) + f(C_1) - f(C_1 \cup (C_3 \cap W')) \le f(C_2). && \text{(by \autoref{lem:linked})}
%
\end{align*}
The proof of $f(C'_3) \le f(C_3)$ is symmetric.

Next we prove that $f(C'_i \cap W)<f(W)$ and that $f(C'_i \cap \overline{W})<f(W)$. 
First, note that for each $i\in\{1,2,3\}$,  $C'_i \cap \overline{W} = C_i \cap \overline{W}$.
So the cases with $\overline{W}$ are trivial.
The inequality $f(C'_1 \cap W) < f(W)$ is proved  in ~\autoref{lem:case1}.

The   bound on $f(C'_2 \cap W)$ follows from the following chain of inequalities
\begin{align*}
 f(C'_2 \cap W) &= f(C_2 \cap \overline{W'} \cap W) = f(\overline{C}_2 \cup W' \cup \overline{W}) && \text{(by \autoref{lem:compl})}\\
&\le f(\overline{C}_2 \cup W') + f(W) - f((\overline{C}_2 \cup W') \cap \overline{W}) &&  \text{(submodularity of $f$)}\\
&= f(C_2 \cap \overline{W'}) + f(W) - f(\overline{C}_2 \cap \overline{W}) &&\\
&= f(C_2 \cap \overline{W'}) + f(W) - f({C}_2 \cup {W}) && \text{(by \autoref{lem:compl})}\\
&< f(C_2 \setminus W') + f(W) - f(W)/2 = f(C'_2) + f(W)/2  &&  \text{(by  \autoref{lem:bigside})} \\
&<  f(C_2) + f(W)/2 < f(W).&& \text{(by \eqref{lemma_finalB})}
\end{align*}
%
%
%
The proof  of $f(C'_3 \cap W)< f(W)$ is symmetric.
This completes the proof that $(C'_1, C'_2, C'_3)$ is a \splitW with $f(C'_i) \le f(C_i)$ for all $i$.
 
Now, if $f(C_1 \cap W') > f(W')$, then $f(C'_1) < f(C_1)$ and $(C'_1, C'_2, C'_3)$ contradicts the minimality of $(C_1, C_2, C_3)$.
If $f(C_1 \cap W') = f(W')$, then $(C'_1, C'_2, C'_3)$ has the same width  and the sum-width as $(C_1, C_2, C_3)$, which is a \minsplitW. 
Thus in the case of $f(C_1 \cap W') = f(W')$,   $(C'_1, C'_2, C'_3)$ is also a \minsplitW.
\end{proof}

\subsubsection{Completing the proof of \autoref{the:glostro}\label{subsubsec:finish}}

The following lemma is just \Cref{lem:main2,lem:main3} put together.

\begin{lemma}
\label{lem:main}
Let $(C_1, C_2, C_3)$ be a \minsplitW. Then 
for every $W' \subseteq W$,  $f(C_1 \cap W') \le f(W')$. Moreover, if the equality holds, then $(C_1 \cup W', C_2 \setminus W', C_3 \setminus W')$ is also a \minsplitW.
\end{lemma}

As discussed above, the first part of \autoref{lem:main} directly gives \Cref{the:glo} by the symmetry of \splitWs.
Next we prove \Cref{the:glostro} by using also the second part of \autoref{lem:main}.

\combmainlemma*
\begin{proof}
The first part follows from combining the fact that either $T_r[w] \subseteq T[uv]$ or $T_r[w] \subseteq T[vu]$ with \autoref{the:glo}.

For the second part, suppose that for some node $w \in V(T)$ it holds that $T_r[w] \cap C_1 \neq \emptyset$, $T_r[w] \not\subseteq C_1$, and $f(T_r[w] \cap C_1) = f(T_r[w])$.
Take a tripartition $(C'_1, C'_2, C'_3) = (C_1 \cup T_r[w], C_2 \setminus T_r[w], C_3 \setminus T_r[w])$.
By  \Cref{lem:main}, this tripartition is a \minsplitW.
We argue that this \minsplitW contradicts the fact that $(r, C_1, C_2, C_3)$ is a \strongminsplit by $r$-intersecting less nodes of $T$.

The \splitW $(C_1, C_2, C_3)$ $r$-intersects the node $w$ but $(C'_1, C'_2, C'_3)$ does not, so it suffices to prove the implication that if $(C_1, C_2, C_3)$ does not $r$-intersect a node $w'$ then $(C'_1, C'_2, C'_3)$ does not $r$-intersect $w'$:

As both $T_r[w]$ and $T_r[w']$ represent leaves of $r$-subtrees of $T$, it follows that either $T_r[w'] \subseteq T_r[w]$, $T_r[w] \subseteq T_r[w']$, or that $T_r[w]$ and $T_r[w']$ are disjoint.
If $T_r[w'] \subseteq T_r[w]$, then $(C'_1, C'_2, C'_3)$ does not $r$-intersect $w'$.
If $T_r[w] \subseteq T_r[w']$, then $(C_1, C_2, C_3)$ $r$-intersects $w'$.
If $T_r[w']$ and $T_r[w]$ are disjoint, then the intersections of $(C'_1, C'_2, C'_3)$ with $T_r[w']$ are the same as the intersections of $(C_1, C_2, C_3)$ with $T_r[w']$, so $w'$ $r$-intersects $(C'_1, C'_2, C'_3)$ if and only if it $r$-intersects $(C_1, C_2, C_3)$.
\end{proof}

\section{Algorithmic properties of refinement}\label{section_algoprorefine}
In this section we present our algorithmic framework for designing fast \FPT 2-approximation algorithms for computing branch decompositions.
In particular, we show that a sequence of refinement operations decreasing the width of a branch decomposition from $k$ to $k-1$ or concluding $k \le 2 \bw(f)$ can be implemented in time $t(k) \cdot 2^{\OO(k)} \cdot n$  for connectivity functions $f$ whose branch decompositions support dynamic programming with time complexity $t(k)$ per node.
That is, we prove \autoref{the:main_alg}.
The concrete implementation of this framework for rankwidth, with $t(k) = 2^{2^{\OO(k)}}$, is provided in \Cref{sec_rankwidth_approx} and for graph branchwidth, with $t(k) = 2^{\OO(k)}$, in \Cref{sec_gr_bw}.

\subsection{Amortized analysis of refinement}
A naive implementation of the refinement operation, see  \autoref{def_refinement},  would use $\Omega(n)$ time on each refinement, which would result in the time complexity of $\Omega(n^2)$ over the course of $n$ refinements. 
In this section we show that the refinement operations can be implemented so that over any sequence of refinement operations using \strongminsplits on a branch decomposition of width at most $k$, the total work done in refining the branch decomposition amortizes to $2^{\OO(k)} n$.

The efficient implementation of   refinements is based on the notion of the edit set of  a \strongminsplit.  
Informally, the edit set  corresponding to a \strongminsplit $(r, C_1, C_2, C_3)$ are the nodes of the subtree of $T$ obtained after pruning all subtrees whose leaves are entirely from one of the sets $C_i$.
Formally, as follows.
\begin{definition}[Edit set]
Let $T$ be a branch decomposition, $r \in E(T)$, and $(r, C_1, C_2, C_3)$ a \strongminsplit.
The \emph{edit set of $(r, C_1, C_2, C_3)$} is the set $R \subseteq V(T)$ of nodes that $r$-intersect $(C_1, C_2, C_3)$, i.e., $R = \{w \in V(T) \mid T_r[w] \text{ intersects at least two sets from  } \{C_1, C_2, C_3\} \}$.
\end{definition}

Note that for a \strongminsplit $(uv, C_1, C_2, C_3)$, both $u$ and $v$ are necessarily in the edit set.
We formalize the intuition about edit sets in the following lemma.
It will be implicitly used in many of our arguments.

\begin{lemma}\label{lemma_propertiesR}
Let $T$ be a branch decomposition, $r = uv \in E(T)$, $(r, C_1, C_2, C_3)$ a \strongminsplit, $R$ the edit set of $(r, C_1, C_2, C_3)$, and $T'$ the refinement of $T$ with $(r, C_1, C_2, C_3)$. It holds that
\begin{enumerate}
\item[(1)] every node in $R$ is of degree 3,
\item[(2)] the nodes of $R$ induce a connected subtree $T[R]$ of $T$, and
\item[(3)] there exists an edge $r' \in E(T')$ so that for every $w \in V(T) \setminus R$ there is a node $w' \in V(T')$ with $T_r[w] = T'_{r'}[w']$.
\end{enumerate}
\end{lemma}
\begin{proof}
For (1), the set $T_r[w]$ of a leaf $w$  consists of one element. Thus it does not intersect $(C_1, C_2, C_3)$.
For (2), first note that $\{u,v\} \subseteq R$.
Then, consider a node $w \in R \setminus \{u,v\}$, and let $p$ be the $r$-parent of $w$.
It holds that $T_r[w] \subseteq T_r[p]$, so $p$ must also be in $R$.

For (3), observe that by the definition of edit set for every node $w \in V(T) \setminus R$ it holds that $T_r[w] \subseteq C_i$ for some $i$.
This implies that the $r$-subtree of $w$ appears identically in $T'$, and therefore $T'$ consists of the $r$-subtrees of all $w \in N_T(R)$ and a connected subtree inserted in the place of $R$ and connected to $N_T(R)$.
As $|R| \ge 2$, this inserted subtree contains at least one edge, which we can designate as the edge $r'$.
\end{proof}


Next we define the \emph{neighbor partition} of an edit set $R$.

\begin{definition}[Neighbor partition]
Let $(r, C_1, C_2, C_3)$ be a \strongminsplit and $R$ its edit set.
The neighbor partition of $R$ is the partition $(N_1, N_2, N_3)$ of the neighbors $N_T(R)$ of $R$, where $N_i = \{w \in N_T(R) \mid T_r[w] \subseteq C_i\}$.
\end{definition}

Note that the neighbor partition is indeed a partition of $N_T(R)$ by the definition of edit set.

Next we give an algorithm for performing the refinement operation in $\OO(|R|)$ time, given the edit set $R$ and its neighbor partition.

\begin{lemma}
\label{lem:refinealg}
Let $T$ be a branch decomposition, $r = uv \in E(T)$, and $(r, C_1, C_2, C_3)$ a \strongminsplit.
Given the edit set $R$ of $(r, C_1, C_2, C_3)$ and the neighbor partition $(N_1, N_2, N_3)$ of $R$, $T$ can be turned into the refinement of $T$ with $(r, C_1, C_2, C_3)$ in $\OO(|R|)$ time.
Moreover, all information stored in nodes $V(T) \setminus R$ is preserved and the other nodes are marked as new.
\end{lemma}
\begin{proof}
We create three copies $T_1$, $T_2$, $T_3$ of the induced subtree $T[R]$.
We denote the copy of a vertex $x \in R$ in $T_i$ by $x_i$, and denote $R_i = \{x_i \mid x \in R\}$.
To each $T_i$ we also insert a new node $w_i$ on the edge $u_i v_i$, i.e., let $V(T_i) \gets V(T_i) \cup \{w_i\}$ and $E(T_i) \gets E(T_i) \setminus \{u_i, v_i\} \cup \{u_i w_i, w_i v_i\}$.
We then insert a new center node $t$ and connect each $w_i$ to it.

For each node $w \in N_i$, let $p \in R$ be the $r$-parent of $w$.
We remove the edge $wp$ and insert the edge $w p_i$. 
It remains to remove all nodes of $R$, and then iteratively remove degree-1 nodes and suppress degree-2 nodes in $R_1 \cup R_2 \cup R_3 \cup \{t\}$.

For the time complexity, these operations can be done in time linear in $|R|+ |R_1| + |R_2| + |R_3| + |N_1| + |N_2| + |N_3| = \OO(|R|)$ because $T$ is represented as an adjacency list and the maximum degree of $T$ is $3$.
\end{proof}

The outline of the refinement operation in our framework is that the dynamic programming outputs the edit set $R$ and its neighbor partition in $\OO(t(k) |R|)$ time, then the algorithm of \autoref{lem:refinealg} computes the refinement in $\OO(|R|)$ time, and then the dynamic programming tables of the $|R|$ new nodes are computed in $\OO(t(k) |R|)$ time.
To bound the sum of the sizes of the edit sets $R$ over the course of the algorithm, we introduce the following potential function.
%
%
\begin{definition}[$k$-potential]
Let $T$ be a branch decomposition of $f$.
The \emph{$k$-potential of $T$} is
\[
\phi_k(T) = \sum_{\substack{e \in E(T)\\ f(e)<k}} f(e)\cdot 3^{f(e)} +\sum_{\substack{e \in E(T)\\ f(e)\geq k}} 3f(e)\cdot 3^{f(e)}.
\]
\end{definition}
When working with $k$-potentials, we will use the following notation. 
For $x\geq 0$, let  
\[
\phi_k(x) = \left\{
	\begin{array}{ll}
		x 3^{x}, & \text{ if } x < k\\
		3 x \cdot 3^{x}, & \text{ otherwise,}\\
	\end{array}
\right.
\]
For $W \subseteq V$, we will use  $\phi_k(W)$ to denote  $\phi_k(f(W))$. With this notation, the $k$-potential of $T$ is 
\[
\phi_k(T) = \sum_{uv \in E(T)} \phi_k(T[uv]).
\]

Note that for any $k$, the $k$-potential of a branch decomposition $T$ is at most $\OO(3^{\bw(T)} \bw(T) |E(T)|)$, which is $2^{\OO(k)} n$ when $\bw(T) = \OO(k)$.

Next we show that performing a refinement operation with an edit set $R$ decreases the $k$-potential by at least $|R|$.

\begin{lemma}
\label{lem:amortize}
Let $T$ be a branch decomposition with an edge $r = uv \in E(T)$ so that $f(uv) = \bw(T) = k$.  Let $(r, C_1, C_2, C_3)$ be a \strongminsplit, $R$ be the edit set of $(r, C_1, C_2, C_3)$, and $T'$ be the refinement of $T$ with $(r, C_1, C_2, C_3)$.
Then it holds that $\phi_k(T') \le \phi_k(T) - |R|$.
\end{lemma}
\begin{proof}
We use the notation that $W = T[uv]$. Note that 
\begin{eqnarray}
\label{lem:amortize:T}
\phi_k(T) = \phi_k(k) + \sum_{w \in (V(T) \setminus \{u,v\})} \phi_k(T_r[w])
\end{eqnarray}
and that by \autoref{observ_refinement} and the fact that $\phi_k(\emptyset) = 0$,
\begin{eqnarray}
\label{lem:amortize:Tp}
\phi_k(T')=\sum_{i \in \{1,2,3\}} \left( \phi_k(C_i) + \sum_{w \in V(T)} \phi_k(C_i \cap T_r[w]) \right).
\end{eqnarray}

Then
\begin{eqnarray*}
&&\phi_k(T) - \phi_k(T')
= \phi_k(T) - \sum_{i \in \{1,2,3\}} \left( \phi_k(C_i) + \sum_{w \in V(T)} \phi_k(C_i \cap T_r[w]) \right)
\end{eqnarray*}
and by taking $C_i \cap T_r[u] = C_i \cap W$ and $C_i \cap T_r[v] = C_i \cap \oW$ out of the sum we get
\begin{eqnarray*}
&\ge& \phi_k(T) - \sum_{i \in \{1,2,3\}} \left(\phi_k(C_i) + \phi_k(C_i \cap W) + \phi_k(C_i \cap \overline{W}) + \sum_{w \in (V(T) \setminus \{u,v\})} \phi_k(C_i \cap T_r[w]) \right)\\
\end{eqnarray*}
and by using (\ref{lem:amortize:T}), $\phi_k(C_i) \le \phi_k((k-1)/2)$, $\phi_k(C_i \cap W) \le \phi_k(k-1)$ by the definition of \splitW, and interleaving the sums (\ref{lem:amortize:T}) and (\ref{lem:amortize:Tp}) we get
\begin{eqnarray*}
&\ge& \phi_k(k) - 3 \phi_k((k-1)/2) - 6 \phi_k(k-1) + \sum_{w \in (V(T) \setminus \{u,v\})} \left(\phi_k(T_r[w]) - \sum_{i \in \{1,2,3\}} \phi_k(C_i \cap T_r[w]) \right)\\
\end{eqnarray*}
and by simplifying $\phi_k(k) - 3 \phi_k((k-1)/2) - 6 \phi_k(k-1) \ge 3 \cdot 3^k$
\begin{eqnarray*}
&\ge& 3 \cdot 3^k + \sum_{w \in (V(T) \setminus \{u,v\})} \left(\phi_k(T_r[w]) - \sum_{i \in \{1,2,3\}} \phi_k(C_i \cap T_r[w]) \right).
\end{eqnarray*}
The next inequality is by noticing that for $w \notin R$, $\phi_k(T_r[w]) = \sum_{i \in \{1,2,3\}} \phi_k(C_i \cap T_r[w])$ because $T_r[w]$ is a subset of some $C_i$ and $\phi_k(\emptyset) = 0$.
The final inequality is from \autoref{the:glostro}, which implies that for any node $w$ in the edit set and every $i$ it holds that $f(T_r[w] \cap C_i) < f(T_r[w])$.
\begin{eqnarray*}
&\ge& 2 + \sum_{w \in (R \setminus \{u,v\})} \left( \phi_k(T_r[w]) - \sum_{i \in \{1,2,3\}} \phi_k(C_i \cap T_r[w]) \right) \ge |R|.\label{lem:amortize:sl}
\end{eqnarray*}
\end{proof}

In particular, when performing a sequence of refinement operations with \strongminsplits on edges $r$ of width $f(r) = \bw(T) = k$, the sum of the sizes of edit sets is at most $\OO(3^k \cdot k \cdot |E(T)|)$.

\subsection{Refinement data structure\label{subsec:refds}}
We define a \emph{refinement data structure} to formally capture what is required from the underlying dynamic programming in our framework.

\begin{definition}[Refinement data structure]
Let $f$ be a connectivity function.
A refinement data structure of $f$ with time complexity $t(k)$ maintains a branch decomposition $T$ of $f$ with $\bw(T) \le k$ rooted on an edge $r = uv \in E(T)$ and supports the following operations:

\begin{enumerate}
\item Init($T$, $uv$): Given a branch decomposition $T$ of $f$ with $\bw(T) \le k$ and an edge $uv \in E(T)$, initialize the data structure in $\OO(t(k)|V(T)|)$ time.
\item Move($vw$): Move the root edge $r = uv$ to an incident edge $vw$, i.e., set $r \gets vw$.
Works in $\OO(t(k))$ time.
\item Width(): Return $f(uv)$ in $\OO(t(k))$ time.
\item CanRefine(): Returns true if there exists a \splitW where $W = T[uv]$ and false otherwise. Works in time $\OO(t(k))$. Once CanRefine() has returned true, the following can be invoked:
\begin{enumerate}
\item EditSet(): Let $(r, C_1, C_2, C_3)$ be a \strongminsplit, $R$ the edit set of $(r, C_1, C_2, C_3)$, and $(N_1, N_2, N_3)$ the neighbor partition of $R$. Returns $R$ and $(N_1, N_2, N_3)$. Works in $\OO(t(k)|R|)$ time.
\item Refine($R$, $(N_1, N_2, N_3)$)): Implements the refinement operation described in \autoref{lem:refinealg}, i.e., computes the refinement of $T$ with $(r, C_1, C_2, C_3)$ by removing the edit set $R$ and inserting a connected subtree of $|R|$ nodes in its place. Sets $r$ to an arbitrary edge between two newly inserted nodes (such an edge exists because $|R|\ge2$). Works in $\OO(t(k) |R|)$ time.
\end{enumerate}
\item Output(): Outputs $T$ in $\OO(t(k) |V(T)|)$ time.
\end{enumerate}
\end{definition}

We explain how our algorithm uses the refinement data structure in \Cref{section_general_algorithm}.
Let us here informally explain how the refinement data structure is typically implemented using dynamic programming.
The formal descriptions for rankwidth and graph branchwidth constitute \Cref{sec_rankwidth_approx,sec_gr_bw}.

For each node $w$ of $T$, the refinement data structure stores a dynamic programming table of size $\OO(t(k))$ that represents information of the $r$-subtree of $w$ in such a way that the dynamic programming tables of the nodes $u$ and $v$ combined together can be used to detect the existence of a \splitW on $W = T[uv]$.
Now, the Init($T$, $uv$) operation is to compute these dynamic programming tables in a bottom-up fashion for all nodes from the leafs towards the root $uv$, using $\OO(t(k))$ time per node.
The Move($vw$) operation changes the root edge $uv$ to an incident edge $vw$.
For implementing Move($vw$), we observe the following useful property.

\begin{observation}
\label{obs:reroot}
Let $T$ be a branch decomposition, $r = uv \in E(T)$ an edge of $T$, and $r' = vw \in E(T)$ another edge of $T$.
For all nodes $x \in V(T) \setminus \{v\}$, the $r$-subtree of $x$ is the same as the $r'$-subtree of $x$.
\end{observation}

In particular, as the dynamic programming table of a node depends only on its $r$-subtree, it suffices to re-compute only the dynamic programming table of the node $v$ in $\OO(t(k))$ time when using Move($vw$).
The Width() operation is typically implemented without dynamic programming, using some other auxiliary data structure.
The CanRefine() operation is implemented by combining the information of the dynamic programming tables of $u$ and $v$ in an appropriate way.
Then, the EditSet() operation is implemented by tracing the dynamic programming backwards, working with a representation of the \strongminsplit $(r, C_1, C_2, C_3)$ that allows for efficiently determining whether $C_i$ and $T_r[w]$ intersect.
The Refine($R$, $(N_1, N_2, N_3)$) operation is a direct application of \autoref{lem:refinealg} followed by computing the dynamic programming tables of the $|R|$ new nodes inserted by \autoref{lem:refinealg} in $\OO(t(k) |R|)$ time, and possibly also updating other auxiliary data structures.
The implementation of Output() is typically straightforward, as it just amounts to outputting the branch decomposition that the data structure has been maintaining.

\subsection{General algorithm}\label{section_general_algorithm}
We present a general algorithm that uses the refinement data structure to in time $t(k) 2^{\OO(k)} n$ either improve the width of a given branch decomposition from $k$ to $k-1$ or to conclude that it is already a 2-approximation.

\begin{algorithm}[!t]
\caption{Iterative refinement.\label{alg:detailed}}
\DontPrintSemicolon
\SetKwInOut{Input}{Input}\SetKwInOut{Output}{Output}\SetKwInOut{Time}{Runtime}
\SetKw{Conclude}{conclude}
\Input{A function $f : 2^V \rightarrow \mathbb{Z}_{\ge 0}$ and a branch decomposition $T$ of $f$.}
\Output{A branch decomposition of $f$ of width at most $\bw(T)-1$ or the conclusion that $\bw(T) \le 2\bw(f)$.}
Let $k \gets \bw(T)$\;
Let \state\xspace be an array initialized with the value \unseen for all nodes of $T$, including new nodes that will be created by refinement.\;
Let $s$ be an arbitrary leaf node of $T$\;
$v \gets s$\;
$u \gets $ the neighbor of $v$\;
\state[v] $\gets$ \open\;
\state[u] $\gets$ \open\label{alg:detailed:init}\;
\While{\textnormal{\state[u]} = \open}{
	\uIf{Exists $w \in N_T(u)$ with \textnormal{\state[w]} = \unseen\label{alg:detailed:casenb}}{
		$v \gets u$\;
		$u \gets w$\;
		\state[u] $\gets$ \open\;
	}
	\uElseIf{$f(uv)< k$\label{alg:detailed:widthcheck}}{
		\uIf{$v = s$}{
			\Return $T$\label{alg:detailed:ok}\;
		}
		\uElse{\label{alg:detailed:caseback}
			\state[u] $\gets$ \closed\;
			$u \gets v$\;
			$v \gets $ the node $v \in N_T(u)$ with $\state[v] = \open$\;
			\tcc{Such a node $v$ is unique.}
		}
	}
	\uElse{
		\uIf{Exists a \strongminsplit $(uv, C_1, C_2, C_3)$\label{alg:detailed:caseimp}}{
			$T \gets $ Refine($T$, $(uv, C_1, C_2, C_3)$)\label{alg:detailed:ref}\;
			\tcc{Where the refinement operation works as in \autoref{lem:refinealg}, i.e., by removing the edit set $R$ and inserting a connected subtree of $|R|$ nodes in its place.}
			$v \gets $ the node $v \in N_T(R)$ with $\state[v] = \open$\label{alg:detailed:after}\;
			$u \gets $ the node $u \in N_T(v)$ that was inserted by the refinement\label{alg:detailed:ext21}\;
			\tcc{Such nodes $v$ and $u$ are unique.}
			\state[u] $\gets$ \open\label{alg:detailed:ext22}\;
		}
		\uElse{
			\Conclude $\bw(T) \le 2\bw(f)$\label{alg:detailed:concl}\;
		}
	}
}
\end{algorithm}

Our algorithm is described as a pseudocode \Cref{alg:detailed}.
The algorithm is a depth-first-search on the given branch decomposition $T$, where whenever we return from a subtree via an edge $uv$ of width $f(uv) = k$, we check if there exists a \strongminsplit $(uv, C_1, C_2, C_3)$.
If there does not exists a such a \strongminsplit, then we conclude that $T$ is already a 2-approximation.
If there exists a such a \strongminsplit, then we refine $T$ using it.
We need to be careful to proceed so that the refinement does not break invariants of depth-first-search, and the extra work caused by refining with an edit set $R$ can be bounded by $\OO(|R|)$.

Let us explain how \Cref{alg:detailed} is implemented with the refinement data structure.
We always maintain that the root edge $uv$ in the refinement data structure corresponds to the edge $uv$ in \Cref{alg:detailed}.
We start by calling Init($T$, $uv$) after \cref{alg:detailed:init}.
In the cases of \autoref{alg:detailed:casenb} and \autoref{alg:detailed:caseback} the edge $uv$ is changed to an adjacent edge $vw$, which is done by the Move($vw$) operation.
The edge $uv$ is changed also after the refinement operation.
There we can move to the appropriate edge with $|R|$ Move($vw$) operations.
Now that the edge $uv$ of \Cref{alg:detailed} corresponds to the edge $uv$ of the refinement data structure, all non-elementary operations of \Cref{alg:detailed} can be performed with the refinement data structure.
In particular, checking $f(uv)$ on \cref{alg:detailed:widthcheck} is done by Width(), \cref{alg:detailed:caseimp} corresponds to CanRefine(), and \cref{alg:detailed:ref} corresponds to EditSet() and Refine().
The returned edit set is also used to determine the node $v$ on \cref{alg:detailed:after}.

The rest of this section is devoted to proving the correctness and the time complexity of \Cref{alg:detailed}.
The next lemma shows that adding the refinement operation does not significantly change the properties of depth-first-search and provides the key argument for proving the correctness.

\begin{lemma}
\label{lem:path}
 \Cref{alg:detailed} maintains the invariant that the nodes with state \open form a path $w_1, \ldots, w_l$ in $T$, where $l \ge 2$, $w_1 = s$, $w_{l-1} = v$, and $w_l = u$.
\end{lemma}
\begin{proof}
This invariant is satisfied at the beginning of the algorithm.
There are three cases in the if-else structure that do not terminate the algorithm and alter $u$, $v$, or the states, i.e., the cases of \autoref{alg:detailed:casenb}, \autoref{alg:detailed:caseback}, and \autoref{alg:detailed:caseimp}.
The case of \autoref{alg:detailed:casenb} maintains the invariant by extending the path by one node.
The case of \autoref{alg:detailed:caseback} maintains the invariant by removing the last node of the path.
In the case of \autoref{alg:detailed:caseimp}, recall that both $u$ and $v$ are in the edit set $R$ and the edit set is a connected subtree of $T$, so the refinement removes some suffix $w_j, \ldots, w_l$ of the path.
Together with the fact that $w_1$ is a leaf and thus $w_1 \notin R$ (see \autoref{lemma_propertiesR}), this implies that the node $v$ determined in \autoref{alg:detailed:after} must be the node $w_{j-1}$ of the path.
Finally, the path is extended by one node in \cref{alg:detailed:ext21,alg:detailed:ext22}.
\end{proof}

The next lemma will be used to prove the correctness of \Cref{alg:detailed} in the case when it returns an improved branch decomposition.

\begin{lemma}
\label{lem:closed}
If  \Cref{alg:detailed} reaches \autoref{alg:detailed:ok}, i.e., terminates by returning a branch decomposition, then all nodes except $v$ and $u$ have state \closed.
\end{lemma}
\begin{proof}
We show that  \Cref{alg:detailed} maintains the invariant that if a node $w$ is closed, then all nodes $w'$ in the $s$-subtree of $w$ are also closed.
This is trivially maintained by the case of \autoref{alg:detailed:casenb}.
In other cases all of the neighbors of $u$ except $v$ are closed due to \autoref{lem:path}.
This implies that the case of \autoref{alg:detailed:caseback} also maintains the invariant.
The case of \autoref{alg:detailed:caseimp}, i.e., refinement, maintains this because by \autoref{alg:detailed:after}, the union of the edit set $R$ and the path $w_1, \ldots, w_l$ defined in \autoref{lem:path} is a connected subtree that contains $u$, $v$, and $s$.

This invariant implies the conclusion of the lemma because at \autoref{alg:detailed:ok} all neighbors of $u$ except $v$ are closed.
\end{proof}



We are ready to prove the correctness and the time complexity of \Cref{alg:detailed}.
In particular, next we complete the proof of the main theorem of this section.

\themainalg* 
\begin{proof}
It suffices to prove that \Cref{alg:detailed} is correct and works in time $t(k) 2^{\OO(k)} n$ provided a refinement data structure with time complexity $t(k)$.

For the correctness. The algorithm terminates with the conclusion  $\bw(T) \le 2\bw(f)$ if and only if there is no \splitW of $W = T[uv]$, where $f(W) = \bw(T)$.
Therefore, by \autoref{the:find}, $\bw(T) = f(W) \le 2 \bw(f)$.
%
%
For the other case, let $w \neq s$ be a node of $T$ and $p$ be  the $s$-parent of $w$.
We note that the state of $w$ can be closed only if $f(wp) < k$.
Therefore, by \autoref{lem:closed}, when  \Cref{alg:detailed} reaches \autoref{alg:detailed:ok}, we have that $f(wp) < k$ for all $w \in V(T) \setminus \{u,v\}$, and by \cref{alg:detailed:widthcheck} we also have $f(uv) < k$.
Therefore we have that $f(e) < k$ for all edges $e$ of $T$, implying that \Cref{alg:detailed} is correct when it returns a branch decomposition.

For the running time.
By the definition of \splitW and \autoref{the:glostro}, the width of $T$ never increases.
By \autoref{lem:amortize}, with every refinement, the potential function drops by at least  $|R|$, the size of the edit set. While we cannot control the size of the edit set for each new refinement, the total sum of the sizes of the edit sets over all the sequence of refinements 
does not exceed  $\phi_k(T)=2^{\OO(k)} n$. 
Thus the total time complexity of the refinement operations is $t(k) 2^{\OO(k)} n$ and the total number of new nodes created over the course of the algorithm in refinement operations is $2^{\OO(k)} n$.
All cases of the algorithm advance the state of some node either from unseen to open or from open to closed, and therefore the total number of operations is $2^{\OO(k)} n$ and their total time $t(k) 2^{\OO(k)} n$.
\end{proof}

\section{Approximating rankwidth}\label{sec_rankwidth_approx}

In this section we prove \Cref{theorem_rw_algo}, that is, for integer $k$ and an $n$-vertex graph $G$, there is an algorithm that in time $2^{2^{\OO(k)}} n^2$ either computes a rank decomposition of width at most $2k$, or correctly concludes that the rankwidth of $G$ is more than $k$. To prove that theorem, we define ``augmented rank decompositions'', show how to implement the refinement data structure of \autoref{the:main_alg} for augmented rank decompositions with time complexity $t(k) = 2^{2^{\OO(k)}}$, and then apply the algorithm of \autoref{the:main_alg} with iterative compression.


\subsection{Preliminaries on rank decompositions\label{sec:rankdecomp}}
We start with preliminaries on rank decompositions and  representatives.

Most of the material is in this subsection is well-known and is commonly used  in dynamic programming over rank decompositions
\cite{DBLP:journals/dam/Bui-XuanTV10,GanianH10,GanianHO13}. 
For reader's convenience, we provide short proofs here. 

Let us remind that a rank decomposition of a graph $G$ is a branch decomposition  of the function $\cutrank_G : 2^{V(G)} \rightarrow \mathbb{Z}_{\ge 0}$ defined on vertex sets of a graph $G$.    
 For a graph
 $G$ and a subset $A$ of the vertex set $V(G)$, we define $\cutrank_G(A)$, 
as the rank of
the $|A| \times |\overline{A}|$  $0-1$ matrix $M_G[A, \overline{A}]$ over the binary field where the entry $m_{i,j}$ of $M_G[A, \overline{A}]$ on
 the $i$-th row and the $j$-th column is $1$ if and only if the $i$-th vertex in $A$ is adjacent to
 the $j$-th vertex in $\overline{A}$. That is,  $\cutrank_G(A)=\rank(M_G[A, \overline{A}])$. Then the rankwidth of a graph $G$, $\rw(G)$, is the minimum width of a rank decomposition.

%
%

\begin{proposition}[\cite{OumS06}]
The function $\cutrank_G$ is a connectivity function.
\end{proposition}

 In the rest of this section   we will always assume that we are computing the rankwidth of a graph $G$, and therefore we drop the subscript in $\cutrank_G$.
We also assume that a representation of $G$ that allows for checking the existence of an edge in $\OO(1)$ time (e.g. an adjacency matrix) is available.



\begin{definition}[Representative]
Let $A \subseteq V(G)$.
A set $R \subseteq A$ is a representative of $A$ if for every vertex $v \in A$ there is a vertex $u \in R$ with $N(v) \setminus A = N(u) \setminus A$.
The representative $R$ is minimal if for each such $v \in A$ there exists exactly one such $u \in R$.
\end{definition}

Note that there always exists a minimal representative.
The property of $\cutrank$ that we exploit in dynamic programming is that it bounds the size of any minimal representative.

\begin{proposition}[\cite{OumS06}]
\label{lem:cutrank_rep}
Let $A \subseteq V(G)$.
If $\cutrank(A) \le k$, then any minimal representative of $A$ has size at most $2^k$.
\end{proposition}

To do computations on minimal representatives, we usually need to work with representatives of cuts $(A, \overline{A})$.

\begin{definition}[Representative of $(A, \overline{A})$]
Let $A \subseteq V(G)$.
A set $(R, Q)$ with $R \subseteq A$ and $Q \subseteq \overline{A}$ is a (minimal) representative of $(A, \overline{A})$ if $R$ is a (minimal) representative of $A$ and $Q$ is a (minimal) representative of $\overline{A}$.
\end{definition}

Given some representative of $(A, \overline{A})$, a minimal representative can be found in polynomial time by the following lemma.

\begin{lemma}
\label{lem:minrep_computation}
Let $A \subseteq V(G)$.
Given a representative $(R, Q)$ of $(A, \overline{A})$, a minimal representative of $(A, \overline{A})$ can be computed in $(|R|+|Q|)^{\OO(1)}$ time.
\end{lemma}
\begin{proof}
One can use the partition refinement algorithm on the graph $G[R,Q]$.
\end{proof}

We will also need the following lemma to work with representatives.

\begin{lemma}
\label{lem:unirep}
Let $R_A$ be a representative of $A$ and $R_B$ a representative of $B$.
Then $R_A \cup R_B$ is a representative of $A \cup B$.
\end{lemma}
\begin{proof}
If $N(v) \setminus A = N(u) \setminus A$, then $N(v) \setminus (A \cup B) = N(u) \setminus (A \cup B)$.
\end{proof}

Next define the $A_R$-representative of a vertex.

\begin{definition}[$A_R$-representative of a vertex]
Let $A \subseteq V(G)$, $R$ a minimal representative of $A$, and $v$ a vertex $v \in A$.
The \emph{$A_R$-representative of $v$}, denoted by $\rep_{A_R}(v)$,  is the vertex $u \in R$ with $N(u) \setminus A = N(v) \setminus A$.
\end{definition}

By the definition of a minimal representative, there exists exactly one $A_R$-representative of a vertex, so the function $\rep_{A_R}(v)$ is well-defined.
Using a minimal representative of $(A, \overline{A})$ we can compute $\rep_{A_R}(v)$ efficiently.

\begin{lemma}
\label{lem:vrep}
Let $A \subseteq V(G)$.
Given a vertex $v \in A$ and a minimal representative $(R, Q)$ of $(A, \overline{A})$, $\rep_{A_R}(v)$ can be computed in $(|R|+|Q|)^{\OO(1)}$ time.
\end{lemma}
\begin{proof}
Test for each $u \in R$ if $N(u) \cap Q = N(v) \cap Q$.
\end{proof}

We also define the $A_R$-representative of a set.

\begin{definition}[$A_R$-representative of a set]
Let $A \subseteq V(G)$, $R$ a minimal $A$-representative, and $X \subseteq A$.
The \emph{$A_R$-representative of $X$}, denoted by $\rep_{A_R}(X)$ is $\rep_{A_R}(X) = \bigcup_{v \in X} \rep_{A_R}(v)$.
\end{definition}

Because $\rep_{A_R}(v)$ is well-defined, $\rep_{A_R}(X)$ is also well-defined.
By repeatedly applying \autoref{lem:vrep}, the $A_R$-representative of a set $X$ can be computed in $|X|(|R|+|Q|)^{\OO(1)}$ time.
As any $A_R$-representative of a set is a subset of $R$, there are at most $2^{|R|}$ different $A_R$-representatives of sets.
Therefore if $\cutrank(A) \le k$, there are at most $2^{2^k}$ different $A_R$-representatives of sets.

Many computations on $A_R$-representatives of sets rely on the following observation.

\begin{observation}
\label{lem:setrepuni}
Let $A \subseteq V(G)$, $X \subseteq A$, $Y \subseteq A$, and let $R_A$ be a minimal representative of $A$, $R_X$ be a minimal representative of $X$, and $R_Y$ be a minimal representative of $Y$.
Let also $X' \subseteq X$ and $Y' \subseteq Y$. Then
it holds that $\rep_{A_R}(X' \cup Y') = \rep_{A_R}(\rep_{X_{R_X}}(X') \cup \rep_{Y_{R_Y}}(Y'))$.
\end{observation}

\subsection{Augmented rank decompositions\label{subsec:augrankdecomp}}
In order to do dynamic programming efficiently on a rank decomposition, we define the notion of an \emph{augmented rank decomposition}.

\begin{definition}[Augmented rank decomposition]
An \emph{augmented rank decomposition} is a pair $(T, \reps)$, where $T$ is a rank decomposition, and for each edge $uv\in E(T)$,
a minimal representative $(\reps[uv], \reps[vu])$ of $(T[uv], T[vu])$ 
 corresponding to a cut $(T[uv], T[vu])$  is stored.
\end{definition}

We will also use a notation $\reps_r[w]$ to denote the minimal representative of $T_r[w]$ stored in $\reps$.
Note that by \autoref{lem:cutrank_rep} an augmented rank decomposition can be represented in $\OO(2^k n)$ space, where $k$ is the width.

Next we show that we can maintain an augmented rank decomposition in our iterative compression.
\begin{lemma}
\label{lem:iterative_add}
Let $v \in V(G)$.
Given an augmented rank decomposition $(T, \reps)$ of $G \setminus \{v\}$ of width $k$, an augmented rank decomposition of $G$ of width at most $k+1$ can be computed in $2^{\OO(k)} n$ time.
\end{lemma}
\begin{proof}
We obtain a rank decomposition $T'$ of $G$ by subdividing an arbitrary edge of $T$ and inserting $v$ as a leaf connected to the node created by subdividing.
The width of $T'$ is at most $k+1$ because adding one vertex to $A$ increases $\cutrank(A)$ by at most one.

For the new edge incident with $v$ a minimal representative is easy to compute in $\OO(n)$ time: $\{v\}$ is a minimal representative of $\{v\}$ and a minimal representative of $V(G) \setminus \{v\}$ has one vertex from $N(v)$ and one from $V(G) \setminus (N(v) \cup \{v\})$.

All other edges of $T'$ correspond to edges $uv \in E(T)$ in the sense that they correspond to some cut $(A', \overline{A'}) = (T[uv] \cup \{v\}, T[vu])$.
We start by setting for each such $A'$ the representative as $\reps[uv] \cup \{v\}$.
This is not necessarily a minimal representative of $A'$, but we turn it into minimal later.
Now we have representatives of size at most $2^k+1$ for the sides of cuts containing $v$.

For the sides of cuts not containing $v$, we compute minimal representatives by dynamic programming.
We root the decomposition at $v$ and proceed from leafs to roots.
For leafs, the minimal representatives have exactly one vertex, the leaf.
For non-leafs, the situation is that we want to compute a minimal representative of a set $B$ with $v \notin B$ such that $B = B_1 \cup B_2$, where we have already computed minimal representatives $R_1$ and $R_2$ of $B_1$ and $B_2$, and a representative $R_A$ of $\overline{B}$ of size $|R_A| \le 2^k+1$.
By \autoref{lem:unirep}, $R_1 \cup R_2$ is a representative of $B$, so $(R_1 \cup R_2, R_A)$ is a representative of $(B, \overline{B})$ of size $2^{\OO(k)}$, so we use \autoref{lem:minrep_computation} to compute a minimal representative of $(B, \overline{B})$ in time $2^{\OO(k)}$.
\end{proof}

\subsection{The algorithm for rankwidth\label{sec:algrank}}

This subsection is devoted to proving the following lemma.

\begin{lemma}
\label{the:rankwidth_refinement_ds}
There is a refinement data structure for rank decompositions with time complexity $t(k) = 2^{2^{\OO(k)}}$, where the init operation requires an augmented rank decomposition and the output operation outputs an augmented rank decomposition.
\end{lemma}

Combined with \autoref{the:main_alg} we get the following corollary.

\begin{corollary}
\label{cor:rankwidthcomp}
There is an algorithm, that given an augmented rank decomposition of $G$ of width $k$ outputs an augmented rank decomposition of $G$ of width at most $k-1$ or correctly concludes that $k \le 2\rw(G)$ in time $2^{2^{\OO(k)}} n$.
\end{corollary}

And by applying iterative compression with \autoref{cor:rankwidthcomp} and \Cref{lem:iterative_add} we get the main result of this section.

\thmrankwidthalgorithm*
%

Our refinement data structure is based on characterizing \strongminsplits by dynamic programming on the augmented rank decomposition $(T, \reps)$ directed towards the edge $r \in E(T)$.
In \Cref{subsec:embs,subsubsec:repsembs,subsubsec:imprembs,subsec:dpiembs} we introduce the objects manipulated in this dynamic programming and prove properties of them and in \Cref{subsec:rankref} we apply this dynamic programming to provide the refinement data structure.

\subsubsection{Embeddings\label{subsec:embs}}
Bui-Xuan, Telle, and Vatshelle in \cite{DBLP:journals/dam/Bui-XuanTV10} characterized the cut-rank of a cut $(A,\overline{A})$ by the existence of an embedding of $G[A, \overline{A}]$ into a certain graph $R_k$.
Next we define this notion of embedding.
In our definition the function describing the embedding is in some sense inversed.
This inversion will make manipulating embeddings in dynamic programming easier.

\begin{definition}[Embedding]
Let $G$ be a bipartite graph   with   bipartition $(A_G, B_G)$ of $V(G)$ and $H$ be a bipartite graph    with   bipartition $(A_H, B_H)$ of $V(H)$.
A function $f : V(H) \rightarrow 2^{V(G)}$ is an \emph{embedding of $G$ into $H$} if
\begin{itemize}\item
$f(u) \cap f(v) = \emptyset$ for $u \neq v$, \item $A_G = \bigcup_{v \in A_H} f(v)$,  $B_G = \bigcup_{v \in B_H} f(v)$, and \item for every pair $(a_H, b_H) \in A_H \times B_H$ and every $(a,b) \in f(a_H) \times f(b_H)$,  it holds that $a b \in E(G)$ if and only if $a_H b_H \in E(H)$.
\end{itemize}
\end{definition}

When using the notation $G[X,Y]$ to construct a bipartite graph, we always assume that the bipartition of $G[X,Y]$ is $(X,Y)$.
Note that the embedding completely defines the edges of $G[X,Y]$ in terms of the edges of $H$, and in particular gives a representative of $(X,Y)$ of size $|V(H)|$, as we formalize as follows.

\begin{observation}
\label{observation:embrep}
Let $G$ be a graph and $(A,\overline{A})$ be a bipartition of $V(G)$.
Let $H$ be a bipartite graph   with  bipartition $(A_H, B_H)$.
Let $f : V(H) \rightarrow 2^{V(G)}$ be an embedding of $G[A,\overline{A}]$ into $H$.
Let $g$ be a function mapping each $v \in V(H)$ to a subset of $f(v)$ as follows:
\[g(v) =
\left\{
\begin{array}{ll}
\{u\} \text{ where } u \in f(v)&\text{if } f(v) \text{ is non-empty,}\\
\emptyset&\text{otherwise.}\\
\end{array}
\right.
\]
Then $(\bigcup_{v \in A_H} g(v), \bigcup_{v \in B_H} g(v))$ is a representative of $(A, \overline{A})$ of size $|V(H)|$.
\end{observation}

Next we  define the graph $R_k$ that will be used to characterize cut-rank. 

\begin{definition}[Graph $R_k$ \cite{DBLP:journals/dam/Bui-XuanTV10}]
For each $k \ge 0$, we denote by $R_k$ the bipartite graph with a bipartition $(A,B)$, having for each subset $X \subseteq \{1,\dots, k\}$ a vertex $a_X \in A$ and a vertex $b_X \in B$, (in particular, having $|A| = 2^k$ and $|B| = 2^k$), and having an edge between $a_X$ and $b_Y$ if and only if $|X \cap Y|$ is odd.
\end{definition}

\begin{proposition}[\cite{DBLP:journals/dam/Bui-XuanTV10}]
\label{lem:cutrankemb}
Let $A \subseteq V(G)$.
It holds that $\cutrank(A) \le k$ if and only if there is an embedding of $G[A, \overline{A}]$ into $R_k$.
\end{proposition}

We will find \strongminsplits by computing embeddings into $R_k$ by dynamic programming.
In order to manipulate embeddings and objects related to embeddings we introduce some notation that naturally extends the definitions of intersections and unions of sets.

\begin{definition}[Intersection $f \cap X$]
Let $f : V(H) \rightarrow 2^A$ be a function and $X \subseteq A$ be a set.
We denote by $f \cap X$ the function $f \cap X : V(H) \rightarrow 2^X$ with $(f \cap X)(v) = f(v) \cap X$.
\end{definition}

We note that an intersection of an embedding and a set is again an embedding.

\begin{observation}
\label{obs:embint}
Let $A$ be a set, $C \subseteq A$, and $X \subseteq A$.
If $f : V(H) \rightarrow 2^A$ is an embedding of $G[A \cap C, A \setminus C]$ into $H$, then $f \cap X$ is an embedding of $G[X \cap C, X \setminus C]$ into $H$.
\end{observation}

Finally, we define the union of functions.

\begin{definition}[Union  $f \cup g$]
Let $f : V(H) \rightarrow 2^{X}$ and $g : V(H) \rightarrow 2^{Y}$ be functions.
We define function $f \cup g : V(H) \rightarrow 2^{X \cup Y}$ with $(f \cup g)(v) = f(v) \cup g(v)$.
\end{definition}

\subsubsection{Representatives of embeddings\label{subsubsec:repsembs}}
Next we define the $A_R$-representative of an embedding, extending the definition of the $A_R$-representative of a set.

\begin{definition}[$A_R$-representative of an embedding]
Let $G$ be a graph, $A \subseteq V(G)$, and $R$ be a minimal representative of $A$.
Let also $f : V(H) \rightarrow 2^A$ be an embedding.
The \emph{$A_R$-representative of $f$} is the function $g : V(H) \rightarrow 2^{R}$, where $g(v) = \rep_{A_R}(f(v))$.
\end{definition}

The $A_R$-representative of an embedding is well-defined because the $A_R$-representative of a set is well-defined.
If $\cutrank(A) \le k$, then the number of $A_R$-representatives of embeddings into $H$ is at most $(2^{2^k})^{|V(H)|}$. In particular, the number of $A_R$-representatives of embeddings into $R_k$ is at most $(2^{2^k})^{2 \cdot 2^k} = 2^{2^{\OO(k)}}$.

We will define the compatibility and the composition of two representatives of embeddings.
The intuition is that embeddings $f_X$ and $f_Y$ of disjoint subgraphs can be merged into an embedding $f_X \cup f_Y$ if and only if their representatives are compatible.
Moreover, the representative of $f_X \cup f_Y$ will be the composition of the representatives of $f_X$ and $f_Y$.

Let $X$ and $Y$ be disjoint subsets of $V(G)$ and $A = X \cup Y$.
Let also $R$ be a minimal representative of $A$, $R_X$ a minimal representative of $X$, and $R_Y$ a minimal representative of $Y$.
Let also $H$ be a   bipartite graph, $g_X : V(H) \rightarrow 2^{R_X}$ an $X_{R_X}$-representative an embedding, and $g_Y : V(H) \rightarrow 2^{R_Y}$ an $Y_{R_Y}$-representative of an embedding.

\begin{definition}[Compatibility]
Let $(A_H, B_H)$ be the bipartition of $H$.
The representatives $g_X$ and $g_Y$ are \emph{compatible} if for every pair $(a_H, b_H) \in A_H \times B_H$ it holds that
\begin{enumerate}
\item for every pair $(a,b) \in g_X(a_H) \times g_Y(b_H)$ it holds that $ab \in E(G) \Leftrightarrow a_H b_H \in E(H)$\label{def:comp:case1} and
\item for every pair $(a,b) \in g_Y(a_H) \times g_X(b_H)$ it holds that $ab \in E(G) \Leftrightarrow a_H b_H \in E(H)$\label{def:comp:case2}.
\end{enumerate}
\end{definition}

Note that compatibility can be tested in $(|V(H)|+|R_X|+|R_Y|)^{\OO(1)}$ time.
Next we show that two embeddings can be merged into an embedding only if their representatives are compatible.

\begin{lemma}
\label{lem:compatible_onlyif}
Let $C \subseteq A$ be a set and let $f_A$ be an embedding of $G[A \cap C, A \setminus C]$ into $H$.
If $g_X$ is the $X_{R_X}$-representative of $f_A \cap X$ and $g_Y$ is the $Y_{R_Y}$-representative of $f_A \cap Y$, then $g_X$ and $g_Y$ are compatible.
\end{lemma}
\begin{proof}
Let $f_X = f_A \cap X$, $f_Y = f_A \cap Y$, and let $(A_H, B_H)$ be the bipartition of $H$.
For the case (\ref{def:comp:case1}) of compatibility it suffices to prove for every pair $(a_H,b_H) \in A_H \times B_H$ and every $(a,b) \in g_X(a_H) \times g_Y(b_H)$ that $a b \in E(G)$ if and only if $a_H b_H \in E(H)$.

As $g_X(a_H)$ is an $X_{R_X}$-representative of $f_X(a_H)$, there exists $a' \in f_X(a_H)$ with $N(a') \setminus X = N(a) \setminus X$.
Similarly there exists $b' \in f_Y(b_H)$ with $N(b') \setminus Y = N(b) \setminus Y$.
Therefore $a b \in E(G)$ if and only if $a' b' \in E(G)$.
Since $f_X(a_H) = f_A(a_H) \cap X$ and $f_Y(b_H) = f_A(b_H) \cap Y$, we have that $a' \in f_A(a_H)$ and that $b' \in f_A(b_H)$.
Because $f_A$ is an embedding it holds that $a' b' \in E(G)$ if and only if $a_H b_H \in E(H)$.
This concludes the proof that the case (\ref{def:comp:case1}) of compatibility holds.

For the case (\ref{def:comp:case2}) it suffices to prove the same but for every $(a,b) \in g_Y(a_H) \times g_X(b_H)$.
This proof is symmetric.
\end{proof}

The composition is defined as the representative of the union.

\begin{definition}[Composition]
The composition of $g_X$ and $g_Y$ is the function $g_A : V(H) \rightarrow 2^A$ defined by $g_A(v) = \rep_{A_R}(g_X(v) \cup g_Y(v))$ for all $v \in V(H)$.
\end{definition}

By using a minimal $(A,\overline{A})$-representative $(R,Q)$ the composition can be computed in time $(|V(H)|+|R|+|Q|+|R_X|+|R_Y|)^{\OO(1)}$.

Next we prove that two embeddings can be merged into an embedding if their representatives are compatible, and in this case the composition gives the resulting representative.

\begin{lemma}
\label{lem:dpcompmain}
Let $C \subseteq A$ be a set.
Let $g_X$ be the $X_{R_X}$-representative of an embedding $f_X$ of $G[X \cap C, X \setminus C]$ into $H$ and $g_Y$ the $Y_{R_Y}$-representative of an embedding $f_Y$ of $G[Y \cap C, Y \setminus C]$ into $H$.
If $g_X$ and $g_Y$ are compatible, then $f_X \cup f_Y$ is an embedding of $G[A \cap C, A \setminus C]$ into $H$ so that the composition of $g_X$ and $g_Y$ is the $A_R$-representative of $f_X \cup f_Y$.
\end{lemma}
\begin{proof}
Let $(A_H, B_H)$ be the bipartition of $H$.
We let $f_A = f_X \cup f_Y$ and observe that $f_A$ is a function $f_A : V(H) \rightarrow 2^A$ that satisfies $f_A(u) \cap f_A(v) = \emptyset$ for $u \neq v$, $(X \cap C) \cup (Y \cap C) = (A \cap C) = \bigcup_{v \in A_H} f_A(v)$, and $(X \setminus C) \cup (Y \setminus C) = (A \setminus C) = \bigcup_{v \in B_H} f_A(v)$.
Therefore $f_A$ is an embedding of $G[A \cap C, A \setminus C]$ into $H$ if for every pair $(a_H, b_H) \in A_H \times B_H$ and every $(a,b) \in f_A(a_H) \times f_A(b_H)$ it holds that $ab \in E(G)$ if and only if $a_H b_H \in E(H)$.
We say that $f_A$ is \emph{good} for a pair $(a,b) \in (A \cap C) \times (A \setminus C)$ if this holds for this pair.
Now $f_A$ is an embedding of $G[A \cap C, A \setminus C]$ into $H$ if it is good for every pair $(a,b) \in (A \cap C) \times (A \setminus C)$.

Because $f_X$ is an embedding of $G[X \cap C, X \setminus C]$ into $H$ we have that $f_X$ is good for all pairs in $(X \cap C) \times (X \setminus C)$ and therefore $f_A$ is good for all pairs in $(X \cap C) \times (X \setminus C)$.
Analogously because $f_Y$ is an embedding of $G[Y \cap C, Y \setminus C]$ into $H$ we have that $f_A$ is good for all pairs in $(Y \cap C) \times (Y \setminus C)$.
It remains to prove that $f_A$ is good for all pairs in $(X \cap C) \times (Y \setminus C)$ and in $(Y \cap C) \times (X \setminus C)$.

Let $(a,b) \in (X \cap C) \times (Y \setminus C)$.
The set $g_X(a_H)$ is a $X_{R_X}$-representative of $f_X(a_H) \supseteq \{a\}$, so there exists $a' \in g_X(a_H)$ so that $N(a') \setminus X = N(a) \setminus X$.
Similarly there exists $b' \in g_Y(b_H)$ so that $N(b') \setminus Y = N(b) \setminus Y$.
Therefore $a b \in E(G)$ if and only if $a' b' \in E(G)$.
Now, by compatibility (\ref{def:comp:case1}) it holds that $a' b' \in E(G)$ if and only if $a_H b_H \in E(G)$.
Therefore $f_A$ is good for $(a,b)$.

The proof for $(a,b) \in (Y \cap C) \times (X \setminus C)$ is symmetric, using compatibility (\ref{def:comp:case2}) instead.
Therefore $f_A$ is an embedding of $G[A \cap C, A \setminus C]$ into $H$.
Finally, by \autoref{lem:setrepuni} we have that the composition of $g_X$ and $g_Y$ is the $A_R$-representative of $f_A$.
\end{proof}

\subsubsection{Improvement embeddings\label{subsubsec:imprembs}}
In order to construct a \minsplitW $(C_1, C_2, C_3)$, we build six embeddings simultaneously in the dynamic programming, three to bound $\cutrank(C_i)$ and three to bound $\cutrank(C_i \cap W)$.
We of course also need to bound $\cutrank(C_i \cap \oW)$, but note that $G[(C_i \cap \oW) \cap W, (\overline{C_i \cap \oW}) \cap W] = G[\emptyset, W]$, so dynamic programming is not required for building the side of $W$ of the embedding $G[C_i \cap \oW, \overline{C_i \cap \oW}]$.

\begin{definition}[$A$-restricted improvement embedding]
Let $A \subseteq V(G)$.
A 10-tuple
\[E = (f^C_1, f^C_2, f^C_3, f^W_1, f^W_2, f^W_3, k_1, k_2, k_3, l)\]
is an $A$-restricted improvement embedding if
there exists a tripartition $(C_1, C_2, C_3)$ of $A$ so that 
\begin{enumerate}
\item for each $i \in \{1,2,3\}$, $f^C_i$ is an embedding of $G[A \cap C_i, A \cap \overline{C_i}]$ into $R_{k_i}$ and
\item for each $i \in \{1,2,3\}$, $f^W_i$ is an embedding of $G[A \cap C_i, A \cap \overline{C_i}]$ into $R_l$.
\end{enumerate}
\end{definition}

Note that an $A$-restricted improvement embedding $E$ uniquely defines such a tripartition $(C_1, C_2, C_3)$ of $A$.
We call such a tripartition \emph{the tripartition of $E$}.
We say that $E$ \emph{intersects a set} if its tripartition $(C_1, C_2, C_3)$ intersects it.
We call the quadruple $(k_1, k_2, k_3, l)$ the \emph{shape of $E$}.
An $A$-restricted improvement embedding is $k$-bounded if $k_1,k_2,k_3,l \le k$.

In the following lemma we introduce the notation $E \cap X$, where $E$ is an $A$-restricted improvement embedding and $X \subseteq A$.

\begin{lemma}
\label{lem:imprtupleint}
Let $A \subseteq V(G)$, $X \subseteq A$, and 
\[E = (f^C_1, f^C_2, f^C_3, f^W_1, f^W_2, f^W_3, k_1, k_2, k_3, l)\] be an $A$-restricted improvement embedding with tripartition $(C_1, C_2, C_3)$.
The tuple 
\[E \cap X = (f^C_1 \cap X, f^C_2 \cap X, f^C_3 \cap X, f^W_1 \cap X, f^W_2 \cap X, f^W_3 \cap X, k_1, k_2, k_3, l)\]
is an $X$-restricted improvement embedding with tripartition $(C_1 \cap X, C_2 \cap X, C_3 \cap X)$.
\end{lemma}
\begin{proof}
By \autoref{obs:embint}, for each $i$ $f^C_i \cap X$ is an embedding of $G[X \cap C_i, X \cap \overline{C_i}]$ to $R_{k_i}$ and $f^W_i$ is an embedding of $G[X \cap C_i, X \cap \overline{C_i}]$ to $R_{l}$.
\end{proof}

We also define the union of improvement embeddings, extending the definition of the union of embeddings.

\begin{definition}
Let $X$ and $Y$ be disjoint subsets,
\[E_1 = (f^C_1, f^C_2, f^C_3, f^W_1, f^W_2, f^W_3, k_1, k_2, k_3, l)\] an $X$-restricted improvement embedding, and
\[E_2 = (g^C_1, g^C_2, g^C_3, g^W_1, g^W_2, g^W_3, k_1, k_2, k_3, l)\] an $Y$-restricted improvement embedding with the same shape.
We denote by $E_1 \cup E_2$ the 10-tuple
\[E_1 \cup E_2 = (f^C_1 \cup g^C_1, f^C_2 \cup g^C_2, f^C_3 \cup g^C_3, f^W_1 \cup g^W_1, f^W_2 \cup g^W_2, f^W_3 \cup g^W_3, k_1, k_2, k_3, l).\]
\end{definition}

Note that if the tripartition of $X$ is $(C_1 \cap X, C_2 \cap X, C_3 \cap X)$, the tripartition of $Y$ is $(C_1 \cap Y, C_2 \cap Y, C_3 \cap Y)$, and $E_1 \cup E_2$ is indeed an $X \cup Y$-restricted improvement embedding, then the tripartition of $E_1 \cup E_2$ is $(C_1, C_2, C_3)$.

We will use dynamic programming on the rank decomposition $T$ to compute improvement embeddings, minimizing the number of nodes of $T$ intersected.

\subsubsection{Representatives of improvement embeddings\label{subsec:dpiembs}}
The definition of the $A_R$-representative of an improvement embedding extends the definition of the $A_R$-representative of an embedding.

\begin{definition}[$A_R$-representative of an improvement embedding]
Let $A \subseteq V(G)$, and $R$ a minimal representative of $A$.
Let also 
\[E = (f^C_1, f^C_2, f^C_3, f^W_1, f^W_2, f^W_3, k_1, k_2, k_3, l)\] be an $A$-restricted improvement embedding.
The $A_R$-representative of $E$ is the tuple
\[(g^C_1, g^C_2, g^C_3, g^W_1, g^W_2, g^W_3, k_1, k_2, k_3, l),\] where each such $g$ is the $A_R$-representative of the corresponding embedding $f$.
\end{definition}

The shape of the $A_R$-representative of $E$ is the same as the shape of $E$.
When $\cutrank(A) \le k$, the number of $A_R$-representatives of $k$-bounded improvement embeddings is $(2^{2^{\OO(k)}})^6 k^4 = 2^{2^{\OO(k)}}$.
We naturally extend the definitions of composition and compatibility to representatives of improvement embeddings.

Let $G$ be a graph, $X$ and $Y$ disjoint subsets of $V(G)$, and $A = X \cup Y$.
Let also $R$ be a minimal representative of $A$, $R_X$ a minimal representative of $X$, and $R_Y$ a minimal representative of $Y$.
Let $E_X = (f^C_1, f^C_2, f^C_3, f^W_1, f^W_2, f^W_3, k_1, k_2, k_3, l)$ be the $X_{R_X}$-representative of an $X$-restricted improvement embedding and $E_Y  = (g^C_1, g^C_2, g^C_3, g^W_1, g^W_2, g^W_3, k'_1, k'_2, k'_3, l')$ the $Y_{R_Y}$-representative of an $Y$-restricted improvement embedding.

\begin{definition}[$C$-Compatibility]
$E_X$ and $E_Y$ are $C$-compatible if they have the same shape and for each $i \in \{1,2,3\}$, $f^C_i$ and $g^C_i$ are compatible.
\end{definition}

\begin{definition}[Compatibility]
$E_X$ and $E_Y$ are compatible if they are $C$-compatible and for each $i \in \{1,2,3\}$, $f^W_i$ and $g^W_i$ are compatible.
\end{definition}

We derive the following lemma directly from \autoref{lem:compatible_onlyif}.

\begin{lemma}
Let $E$ be an $A$-restricted improvement embedding.
If $E_X$ is the $X_{R_X}$-representative of $E \cap X$ and $E_Y$ is the $Y_{R_Y}$-representative of $E \cap Y$, then $E_X$ and $E_Y$ are compatible.
\end{lemma}
\begin{proof}
Apply \autoref{lem:compatible_onlyif} to each embedding in $E$.
\end{proof}

Next we define the composition of two representatives of improvement embeddings, extending the definition of the composition of representatives of embeddings.

\begin{definition}[Composition]
If $E_X$ and $E_Y$ are compatible, then the composition of $E_X$ and $E_Y$ is the pointwise composition of $E_X$ and $E_Y$, i.e., the composition is the 10-tuple
\[E_A = (h^C_1, h^C_2, h^C_3, h^W_1, h^W_2, h^W_3, k_1, k_2, k_3, l),\]
where for each $i \in \{1,2,3\}$ $h^C_i$ is the composition of $f^C_i$ and $g^C_i$, and $h^W_i$ is the composition of $f^W_i$ and $g^W_i$.
\end{definition}

Next we prove the main lemma for computing improvement embeddings by dynamic programming, which is an extension of \autoref{lem:dpcompmain}.

\begin{lemma}
\label{lem:dpmainif}
Let $X$ and $Y$ be disjoint sets, $A = X \cup Y$, $R$ a minimal representative of $A$, $R_X$ a minimal representative of $X$, and $R_Y$ a minimal representative of $Y$.
Let $E_1$ be an $X$-restricted improvement embedding and $E_2$ a $Y$-restricted improvement embedding.
Let $E_X$ be the $X_{R_X}$-representative of $E_1$ and $E_Y$ the $Y_{R_Y}$-representative of $E_2$.
If $E_X$ and $E_Y$ are compatible, then $E_1 \cup E_2$ is an $A$-restricted improvement embedding and the composition of $E_X$ and $E_Y$ is the $A_R$-representative of $E_1 \cup E_2$.
\end{lemma}
\begin{proof}
Denote \[E_1 = (f^C_1, f^C_2, f^C_3, f^W_1, f^W_2, f^W_3, k_1, k_2, k_3, l)\] and
\[E_2 = (g^C_1, g^C_2, g^C_3, g^W_1, g^W_2, g^W_3, k_1, k_2, k_3, l).\]
Let $(C^X_1, C^X_2, C^X_3)$ be the tripartition of $E_1$ and $(C^Y_1, C^Y_2, C^Y_3)$ the tripartition of $E_2$.
By \autoref{lem:dpcompmain} it holds for every $i$ that $f^C_i \cup g^C_i$ is an embedding of $G[C^X_i \cup C^Y_i, A \setminus (C^X_i \cup C^Y_i)]$ to $R_{k_i}$ and $f^W_i \cup g^W_i$ is an embedding of $G[C^X_i \cup C^Y_i, A \setminus (C^X_i \cup C^Y_i)]$ to $R_l$.
Therefore
\[E = (f^C_1 \cup g^C_1, f^C_2 \cup g^C_2, f^C_3 \cup g^C_3, f^W_1 \cup g^W_1, f^W_2 \cup g^W_2, f^W_3 \cup g^W_3, k_1, k_2, k_3, l)\]
is an improvement embedding with tripartition $(C^X_1 \cup C^Y_1, C^X_2 \cup C^Y_2, C^X_3 \cup C^Y_3)$.
By \autoref{lem:dpcompmain} the composition of $E_X$ and $E_Y$ is the $A_R$-representative of $E$.
\end{proof}

\paragraph{Finding the \splitW.}
In the dynamic programming we will determine if there exists a \splitW based on the representatives of $W$-restricted improvement embeddings and the representatives of $\oW$-restricted improvement embeddings.
For this purpose we will define the root-compatibility of two representatives of improvement embeddings, characterizing whether they can be merged to yield a \splitW.

Let $W \subseteq V(G)$, $R_W$ a minimal representative of $W$, and $R_{\oW}$ a minimal representative of $\oW$.
Let also $E_W = (f^C_1, f^C_2, f^C_3, f^W_1, f^W_2, f^W_3, k_1, k_2, k_3, l)$ be a $W_{R_W}$-representative of a $W$-restricted improvement embedding and $E_\oW = (g^C_1, g^C_2, g^C_3, g^W_1, g^W_2, g^W_3, k_1, k_2, k_3, l)$ a $\oW_{R_\oW}$-representative of an $\oW$-restricted improvement embedding so that $E_W$ and $E_\oW$ have the same shape.

\begin{definition}[Root-compatibility]
$E_W$ and $E_\oW$ are root-compatible if they are $C$-compatible and
\begin{enumerate}
\item for each $i$, there exists an embedding $f^\oW_i$ of $G[\emptyset, \oW]$ into $R_l$ so that $f^W_i$ is compatible with the $\oW_{R_\oW}$-representative of $f^\oW_i$\label{def:rootcomp:c1} and
\item for each $i$, there exists an embedding $g^\oW_i$ of $G[\emptyset, W]$ into $R_l$ so that $g^W_i$ is compatible with the $W_{R_W}$-representative of $g^\oW_i$.
\end{enumerate}
\end{definition}

The definition does not directly provide an efficient algorithm for checking root-compatibility, but a couple of observations and brute force yields a $2^{2^{\OO(k)}}$ time algorithm as follows.

\begin{lemma}
\label{lem:rootcompcheck}
Let $\cutrank(W) \le k$ and let the improvement embeddings be $k$-bounded.
Given $E_W$, $E_\oW$, $R_W$, and $R_\oW$, the root-compatibility of $E_W$ and $E_\oW$ can be checked in time $2^{2^{\OO(k)}}$.
\end{lemma}
\begin{proof}
We prove the case (\ref{def:rootcomp:c1}), the other case is symmetric.

Suppose there exists such an embedding $f^\oW_i$.
Let $(A_H, B_H)$ be the bipartition of $R_l$ and let $u_H, v_H \in B_H$ be a pair of distinct vertices in $B_H$.
Suppose that there exists $u,v \in \oW$ such that $N(u) \cap W = N(v) \cap W$, $u \in f^\oW_i(u_H)$, and $v \in f^\oW_i(v_H)$.
Note that now, if we remove $u$ from $f^\oW_i(u_H)$ and insert it into $f^\oW_i(v_H)$, we obtain an embedding whose $\oW_{R_\oW}$-representative is compatible with the same $W_{R_W}$-representatives as $f^\oW_i$.
Therefore, we can assume for any $u,v \in \oW$ that if $N(u) \cap W = N(v) \cap W$, then there exists $v_H \in B_H$ so that $\{u,v\} \subseteq f^\oW_i(v_H)$.

As for any $v \in \oW$ there is a vertex $v_R \in R_\oW$ with $N(v) \cap W = N(v_R) \cap W$, it is sufficient to enumerate all intersections $f^\oW_i \cap R_\oW$ and assign each $v$ to the same vertex of $B_H$ as $v_R$.
Note that in this case, the $\oW_{R_\oW}$-representative of $f^\oW_i$ depends only on the intersection $f^\oW_i \cap R_\oW$.
The number of the intersections is $\le \left(2^{|R_\oW|}\right)^{|B_H|} \le (2^{2^k})^{2^k} = 2^{2^{\OO(k)}}$ and each can be checked in time $2^{\OO(k)}$.
\end{proof}

We also introduce the definition of $C_i$-emptiness to denote whether none of the vertices have been assigned to $C_i$ in the tripartition.

\begin{definition}[$C_i$-empty]
Let $E_R = (f^C_1, f^C_2, f^C_3, \ldots)$ be the $A_R$-representative of an $A$-restricted improvement embedding.
Let $i \in \{1,2,3\}$, and let $(A_H, B_H)$ be the bipartition of $R_{k_i}$.
$E_R$ is $C_i$-empty if for every $v \in A_H$ it holds that $f^C_i(v) = \emptyset$.
\end{definition}

The definition of $C_i$-emptiness is for determining which of the parts of a corresponding tripartition intersect $A$.

\begin{observation}
\label{obs:cempty}
Let $E_R$ be the $A_R$-representative of an $A$-restricted improvement embedding $E$ and $(C_1, C_2, C_3)$ the tripartition of $E$.
It holds that $C_i = \emptyset$ if and only if $E_R$ is $C_i$-empty.
\end{observation}

Finally, we describe how to find \splitWs based on representatives of improvement embeddings.

\begin{lemma}
\label{lem:dpmainfinal}
Let $T$ be a rank decomposition, $uv = r \in E(T)$, and $W = T[uv]$.
Denote $X = W = T_r[u]$, $Y = \oW = T_r[v]$, and let $R_X$ be a minimal representative of $X$ and $R_Y$ a minimal representative of $Y$.
There exists a \splitW $(C_1, C_2, C_3)$ with arity $\alpha$ and $\cutrank(C_i) \le k_i$ for each $i \in \{1,2,3\}$  
if and only if there exists $E_u$ and $E_v$ such that  
\begin{enumerate}
\item $E_u$ is the $X_{R_X}$-representative of an $X$-restricted improvement embedding with shape $(k_1, k_2, k_3, l)$ whose tripartition $(C_1 \cap X, C_2 \cap X, C_3 \cap X)$ is\label{lem:dpmainfinal:c1},
\item $E_v$ is the $Y_{R_Y}$-representative of an $Y$-restricted improvement embedding with shape $(k_1, k_2, k_3, l)$ whose tripartition $(C_1 \cap Y, C_2 \cap Y, C_3 \cap Y)$ is\label{lem:dpmainfinal:c2},
\item $E_u$ and $E_v$ are root-compatible,
\item $k_i < \cutrank(W)/2$ for all $i$, $l < \cutrank(W)$, and
\item $\alpha = 2$ if there exists $i$ such that both $E_u$ and $E_v$ are $C_i$-empty and otherwise $\alpha = 3$.
\end{enumerate}
\end{lemma}
\begin{proof}
\textbf{If direction:}

Let \[E_X = (f^C_1, f^C_2, f^C_3, f^W_1, f^W_2, f^W_3, k_1, k_2, k_3, l)\]
 be a $X$-restricted improvement embedding whose $X_{R_X}$-representative $E_u$ is and whose tripartition $(C_1 \cap X, C_2 \cap X, C_3 \cap X)$ is.
Let also \[E_Y = (g^C_1, g^C_2, g^C_3, g^\oW_1, g^\oW_2, g^\oW_3, k_1, k_2, k_3, l)\]
 be a $Y$-restricted improvement embedding whose $Y_{R_Y}$-representative $E_u$ is and whose tripartition $(C_1 \cap Y, C_2 \cap Y, C_3 \cap Y)$ is.

As $E_u$ and $E_v$ are $C$-compatible, \autoref{lem:dpcompmain} implies that for each $i$ $f^C_i \cup g^C_i$ is an embedding of $G[C_i, \overline{C_i}]$ to $k_i$.
Therefore by \autoref{lem:cutrankemb} it holds that $\cutrank(C_i) \le k_i$.

As $E_u$ and $E_v$ are root-compatible, by the definition of root-compatibility and \autoref{lem:dpcompmain} for each $i$ there exists an embedding $f^\oW_i$ of $G[\emptyset, \oW]$ to $R_l$ so that $f^W_i \cup f^\oW_i$ is an embedding of $G[C_i \cap W, (\overline{C_i} \cap W) \cup \oW] = G[C_i \cap W, \overline{C_i \cap W}]$ to $R_l$.
Therefore by \autoref{lem:cutrankemb} it holds that $\cutrank(C_i \cap W) \le l$.

Symmetrically by root-compatibility for each $i$ there exists an embedding $g^W_i$ of $G[\emptyset, W]$ to $R_l$ so that $g^\oW_i \cup g^W_i$ is an embedding of $G[C_i \cap \oW, (\overline{C_i} \cap \oW) \cup W] = G[C_i \cap \oW, \overline{C_i \cap \oW}]$ to $R_l$.
Therefore by \autoref{lem:cutrankemb} it holds that $\cutrank(C_i \cap \oW) \le l$.

Finally note that $C_i = \emptyset$ if and only if both $E_u$ and $E_v$ are $C_i$-empty.

\textbf{Only if direction:}

Let $(C_1, C_2, C_3)$ be a \splitW of arity $\alpha$ where $\cutrank(C_1) = k_1$, $\cutrank(C_2) = k_2$, and $\cutrank(C_3) = k_3$.
Let also $l = \cutrank(W)-1$.
By \autoref{lem:cutrankemb}, for each $i$ there exists an embedding $f^C_i$ of $G[C_i, \overline{C_i}]$ to $R_{k_i}$, an embedding $f^W_i$ of $G[C_i \cap W, \overline{C_i \cap W}]$ to $R_l$, and an embedding $f^\oW_i$ of $G[C_i \cap \oW, \overline{C_i \cap \oW}]$ to $R_l$.

Let \[E_X = (f^C_1 \cap X, f^C_2 \cap X, f^C_3 \cap X, f^W_1 \cap X, f^W_2 \cap X, f^W_3 \cap X, k_1, k_2, k_3, l)\]
and note that the tripartition of $E_X$ is $(C_1 \cap X, C_2 \cap X, C_3 \cap X)$.
Let also \[E_Y = (f^C_1 \cap Y, f^C_2 \cap Y, f^C_3 \cap Y, f^\oW_1 \cap Y, f^\oW_2 \cap Y, f^\oW_3 \cap Y, k_1, k_2, k_3, l)\] and note that the tripartition of $E_Y$ is $(C_1 \cap Y, C_2 \cap Y, C_3 \cap Y)$.
Let $E_u$ be the $X_{R_X}$-representative of $E_X$ and $E_v$ the $Y_{R_Y}$-representative of $E_Y$.
Note that now, $E_u$ and $E_v$ satisfy (\ref{lem:dpmainfinal:c1}) and (\ref{lem:dpmainfinal:c2}).

By \autoref{lem:compatible_onlyif}, $E_u$ and $E_v$ are $C$-compatible.
Note that also by \autoref{lem:compatible_onlyif}, for each $i$, $f^W_i \cap Y$ is an embedding of $G[\emptyset, \oW]$ to $R_l$ whose $Y_{R_Y}$-representative is compatible with the $X_{R_X}$-representative of $f^W_i \cap X$, and $f^\oW_i \cap X$ is an embedding of $G[\emptyset, W]$ whose $X_{R_X}$-representative is compatible with the $Y_{R_Y}$-representative of $f^\oW_i \cap Y$.
Therefore $E_u$ and $E_v$ are root-compatible by the definition of root-compatibility.

If $(C_1, C_2, C_3)$ has arity 2, then there is $i$ such that $C_i = \emptyset$, which implies that both $E_u$ and $E_v$ are $C_i$-empty.
Otherwise each $C_i$ intersects with at least one of $X$ or $Y$, and therefore for each $i$ at least one of $E_u$ or $E_v$ is not $C_i$-empty.
\end{proof}

\subsubsection{Refinement data structure for augmented rank decompositions\label{subsec:rankref}}
In this subsection we provide the refinement data structure for augmented rank decompositions using dynamic programming building on the previous subsections.
Troughout this subsection, we assume that we are maintaining an augmented rank decomposition $(T, \reps)$ of width $\rw(T) \le k$ and that therefore we are only interested in $k$-bounded improvement embeddings.

\paragraph{Concrete representatives.}
In order to maintain an augmented rank decomposition in the refinement operation, we need to construct a minimal representative of each set $C_i$ of the \splitW $(C_1, C_2, C_3)$.
To do this, we maintain ``concrete representatives'' in the dynamic programming.

\begin{definition}[Concrete representative]
Let $H$ be a bipartite graph, $A \subseteq V(G)$, and $f : V(H) \rightarrow 2^{A}$ a function.
A concrete representative of $f$ is a function $g$ defined in \autoref{observation:embrep}, i.e., a function $g$ mapping each $v \in V(H)$ to a subset of $f(v)$ as follows:
\[g(v) =
\left\{
\begin{array}{ll}
\{u\} \text{ where } u \in f(v)&\text{if } f(v) \text{ is non-empty}\\
\emptyset&\text{otherwise.}\\
\end{array}
\right.
\]

Let $E = (f^C_1, f^C_2, f^C_3, \ldots)$ be an $A$-restricted improvement embedding.
A concrete representative of $E$ is a triple $(g^C_1, g^C_2, g^C_3)$, where each $g^C_i$ is a concrete representative of $f^C_i$.
\end{definition}

Note that a concrete representative can be represented in $\OO(|V(H)|)$ space.
By \autoref{observation:embrep}, a concrete representative of an embedding of $G[C_i, \overline{C_i}]$ into $H$ can be turned into a representative of $(C_i, \overline{C_i})$ of size $|V(H)|$.
In particular, using a concrete representative of a $V(G)$-restricted improvement embedding with tripartition $(C_1, C_2, C_3)$ we can compute a minimal representatives of each $C_i$ in time $2^{\OO(k)}$.

We define the union of two concrete representatives naturally.

\begin{definition}[Union of concrete representatives]
Let $f_R$ be a concrete representative of a function $f : V(H) \rightarrow 2^X$ and $g_R$ a concrete representative of a function $g : V(H) \rightarrow 2^Y$.
We denote by $f_R \cup g_R$ the function mapping each $v \in V(H)$ to a subset of $f(v) \cup g(v)$ as follows:
\[
(f_R \cup g_R)(v) = \left\{
\begin{array}{ll}
f_R(v), & \text{ if } f_R(v) \neq \emptyset\\
g_R(v), & \text{ otherwise.}\\
\end{array}
\right.
\]

The union of concrete representatives of improvement embeddings is the pointwise union of such triples of concrete representatives.
\end{definition}

Observe that such $f_R \cup g_R$ is a concrete representative of $f \cup g$.

\paragraph{Dynamic programming tables.}
In the refinement data structure we maintain a dynamic programming table for each node $w \in V(T)$.
We call this table the \emph{$r$-table} of the node $w$ to signify that this table contains information about the $r$-subtree of $w$.
Next we formally define an $r$-table.

\begin{definition}[$r$-table]
Let $(T, \reps)$ be an augmented rank decomposition, $r \in E(T)$, $w \in V(T)$, $A = T_r[w]$, and $R = \reps_r[w]$.
An $r$-table of $w$ is a triple $(\embs, \intcount, \crep)$, 
where 
\begin{enumerate}
\item $\embs$ is the set of all $A_R$-representatives of $k$-bounded $A$-restricted improvement embeddings, 
\item $\intcount$ is a function mapping each $E_R \in \embs$ to the least integer $i$ such that there exists an $A$-restricted improvement embedding $E$ whose $A_R$-representative $E_R$ is and which $r$-intersects $i$ nodes of the $r$-subtree of $w$, and 
\item $\crep$ is a function mapping each $E_R \in \embs$ to a concrete representative of an $A$-restricted improvement embedding $E$ such that $E_R$ is the $A_R$-representative of $E$ and $E$ $r$-intersects $\intcount(E_R)$ nodes of the $r$-subtree of $w$.
\end{enumerate}
\end{definition}

Note that an $r$-table can be represented by making use of  $2^{2^{\OO(k)}}$ space.

To correctly construct the refinement that matches the concrete representation obtained, we need to spell out some additional properties of $r$-tables, which will be naturally satisfied by the way the $r$-tables will be constructed.

\begin{definition}[Local and global representation]
Let $(T, \reps)$ be an augmented rank decomposition, $r \in E(T)$, $w \in V(T)$, $A = T_r[w]$, and $R = \reps_r[w]$.
An $r$-table $(\embs, \intcount, \crep)$ of $w$ locally represents an $A$-restricted improvement embedding $E$ if there exists $E_R \in \embs$ so that $E_R$ is the $A_R$-representative of $E$, $\intcount(E_R)$ is the number of nodes in the $r$-subtree of $w$ $r$-intersected by $E$, and $\crep(E_R)$ is a concrete representative of $E$.
The $r$-table of $w$ globally represents $E$ if the $r$-tables of all nodes $w'$ in the $r$-subtree of $w$ (including $w$ itself) locally represent $E \cap T_r[w']$.
\end{definition}

\begin{definition}[Linked $r$-table]
A linked $r$-table of $w$ is a 4-tuple $(\embs, \intcount, \crep, \dplink)$ so that $(\embs, \intcount, \crep)$ is an $r$-table of $w$ and for each $E_R \in \embs$ there exists an $A$-restricted improvement embedding $E$ so that 
\begin{enumerate}
\item $E_R$ is the $A_R$-representative of $E$,
\item $(\embs, \intcount, \crep)$ globally represents $E$, and
\item if $w$ is non-leaf and has $r$-children $w_1$ and $w_2$ with $r$-tables $(\embs_1, \intcount_1, \crep_1)$ and $(\embs_2, \intcount_2, \crep_2)$, then $\dplink$ is a function mapping $E_R$ to a pair $\dplink(E_R) = (E_1, E_2)$ such that $E_1 \in \embs_1$, $E_2 \in \embs_2$, $E_1$ is the $T_r[w_1]_{\reps_r[w_1]}$-representative of $E \cap T_r[w_1]$, and $E_2$ is the $T_r[w_2]_{\reps_r[w_2]}$-representative of $E \cap T_r[w_2]$.
\end{enumerate}
\end{definition}

The existence of a linked $r$-table will be formally proved in \autoref{lem:nonleafrtable} when also its construction is given.
Note that a linked $r$-table of a node $w$ can be represented in $2^{2^{\OO(k)}}$ space.
We start implementing the dynamic programming from the leaves.

\begin{lemma}
Let $(T, \reps)$ be an augmented rank decomposition, $r \in E(T)$, and $w$ a leaf node of $T$.
A linked $r$-table of $w$ can be constructed in time $2^{\OO(k)}$.
\end{lemma}
\begin{proof}
Let $A = T_r[w]$.
As $|A| = 1$, the number of $A$-restricted improvement embeddings is $2^{\OO(k)}$ and we can iterate over all of them and construct the $r$-table directly by definition.
Moreover, the $r$-table is by definition linked because $A$ is the minimal representative of itself and so there is bijection between $A$-restricted improvement embeddings and their $A_A$-representatives.
\end{proof}

The next lemma is the main dynamic programming lemma.
It specifies the computation of linked $r$-tables for non-leaf nodes.

\begin{lemma}
\label{lem:nonleafrtable}
Let $(T, \reps)$ be an augmented rank decomposition, $r \in E(T)$, $w$ a non-leaf node of $T$, and $w_1$ and $w_2$ the $r$-children of $w$.
Given linked $r$-tables of the nodes $w_1$ and $w_2$, a linked $r$-table $(\embs, \intcount, \crep, \dplink)$ of $w$ can be constructed in time $2^{2^{\OO(k)}}$.
\end{lemma}
\begin{proof}
Let $A = T_r[w]$, $X = T_r[w_1]$, $Y = T_r[w_2]$, $R = \reps_r[w]$, $R_X = \reps_r[w_1]$, and $R_Y = \reps_r[w_2]$.
Let $(\embs_1, \intcount_1, \crep_1)$ be the given $r$-table of $w_1$ and $(\embs_2, \intcount_2, \crep_2)$ the given $r$-table of $w_2$.

Let $E$ be a $k$-bounded $A$-restricted improvement embedding.
Note that the number of nodes in the $r$-subtree of $w$ that $E$ $r$-intersects is $i_1 + i_2 + i_w$, where $i_1$ is the number of nodes in the $r$-subtree of $w_1$ that $E \cap X$ $r$-intersects, $i_2$ is the number of nodes in the $r$-subtree of $w_2$ that $E \cap Y$ $r$-intersects, and $i_w = 1$ if $E$ $r$-intersects $w$ and $0$ otherwise.
Moreover, the fact whether $E$ $r$-intersects $w$ can be determined only by considering the $A_R$-representative of $E$, in particular by whether it is $C_i$-empty for at most one $i$.

We enumerate all pairs $(E_1, E_2) \in \embs_1 \times \embs_2$.
If $E_1$ is compatible with $E_2$, then by \autoref{lem:dpmainif}, for any $E_X$ and $E_Y$ such that $E_1$ is a $X_{R_X}$-representative of $E_X$ and $E_2$ is a $Y_{R_Y}$-representative of $E_Y$ it holds that $E = E_X \cup E_Y$ is an $A$-restricted improvement embedding such that $E \cap X = E_X$, $E \cap Y = E_Y$, and the composition $E_R$ of $E_1$ and $E_2$ is the $A_R$-representative of $E$.
Let us fix such $E_X$ and $E_Y$ so that they are globally represented by $(\embs_1, \intcount_1, \crep_1)$ and $(\embs_2, \intcount_2, \crep_2)$, respectively.
If $E_R \notin \embs$, or $E_R \in \embs$ but $\intcount(E_R) > \intcount_1(E_1) + \intcount_2(E_2) + i_w$, we insert $E_R$ to $\embs$ and set $\intcount(E_R) \gets \intcount_1(E_1) + \intcount_2(E_2) + i_w$, $\crep(E_R) \gets \crep_1(E_1) \cup \crep_2(E_2)$ and $\dplink(E_R) \gets (E_1, E_2)$.
Now $(\embs, \intcount, \crep)$ globally represents $E$.

The fact that for all $A$-restricted improvement embeddings $E$ we actually considered a pair $(E_1, E_2) \in \embs_1 \times \embs_2$ so that $E_1$ is a $X_{R_X}$-representative of $E \cap X$ and $E_2$ is a $Y_{R_Y}$-representative of $E \cap Y$ follows from the definition of $r$-state and \autoref{lem:imprtupleint}, i.e., the fact that $E \cap X$ is an $X$-restricted improvement embedding and $E \cap Y$ is an $Y$-restricted improvement embedding.
\end{proof}

Now the Init($T$, $r$) operation can be implemented in $|V(T)| 2^{2^{\OO(k)}}$ time by constructing a linked $r$-table of each node in the order from leafs to root with $|V(T)|$ applications of \autoref{lem:nonleafrtable}.
For the Move($vw$) operation, observe that the linked $r$-table of a node $x \in V(T)$ depends only on the $r$-subtree of $x$, and therefore by \autoref{obs:reroot} it suffices to recompute only the linked $r$-table of $v$ when using Move($vw$).
Therefore Move($vw$) can be implemented in $2^{2^{\OO(k)}}$ time by a single application of \autoref{lem:nonleafrtable}.
The operation Width() is implemented without $r$-tables.
It amounts to finding an embedding of $G[\reps[uv], \reps[vu]]$ to $R_k$ by brute-force in time $2^{2^{\OO(k)}}$.
The Output() operation also is straightforward as we just output the augmented rank decomposition we are maintaining.

It remains to implement the operations CanRefine(), EditSet(), and Refine($R$, $(N_1, N_2, N_3)$).
The following is the main lemma for them, providing the implementations of CanRefine() and EditSet() and setting the stage for Refine($R$, $(N_1, N_2, N_3)$).

\begin{lemma}
\label{lem:rankrefinement}
Let $(T, \reps)$ be an augmented rank decomposition, $uv = r \in E(T)$, and $W = T[uv]$.
For each node $w \in V(T)$, let $(\embs_w, \intcount_w, \reps_w, \dplink_w)$ be the linked $r$-table of $w$.
There is an algorithm that returns $\bot$ if there is no \splitW, and otherwise returns a tuple $(R, N_1, N_2, N_3, C^R_1, C^R_2, C^R_3)$, where $(r, C_1, C_2, C_3)$ is a \strongminsplit, $R$ is the edit set of $(r, C_1, C_2, C_3)$, $(N_1, N_2, N_3)$ is the neighbor partition of $R$, and for each $i$ $C^R_i$ is a minimal representative of $C_i$.
The algorithm works in time $2^{2^{\OO(k)}} (|R|+1)$.
\end{lemma}
\begin{proof}
Denote $X = W = T_r[u]$, $Y = \oW = T_r[v]$, let $R_X = \reps_r[u]$ and $R_Y = \reps_r[v]$.

We iterate over all pairs $(E_u, E_v) \in \embs_u \times \embs_v$ satisfying the conditions in \autoref{lem:dpmainfinal} and determine the width, the arity, and the sum-width of a corresponding \splitW as in \autoref{lem:dpmainfinal}.
Also, the minimum number of nodes of $T$ $r$-intersected by a \splitW corresponding to $(E_u, E_v)$ can be determined as $\intcount_u(E_u) + \intcount_v(E_v)$.
If no such pair is found we return $\bot$.
Otherwise we find a pair $(E_u, E_v)$ so that $E_u$ is the $X_{R_X}$-representative of an $X$-restricted improvement embedding $E_X$ that is globally represented by $(\embs_u, \intcount_u, \crep_u)$ and whose tripartition $(C_1 \cap X, C_2 \cap X, C_3 \cap X)$ is, $E_v$ is the $Y_{R_Y}$-representative of an $Y$-restricted improvement embedding $E_Y$ that is globally represented by $(\embs_v, \intcount_v, \crep_v)$ and whose tripartition $(C_1 \cap Y, C_2 \cap Y, C_3 \cap Y)$ is, and $(r, C_1, C_2, C_3)$ is a \strongminsplit of $T$.
As $|\embs_u||\embs_v| = 2^{2^{\OO(k)}}$ and each pair can be checked in time $2^{2^{\OO(k)}}$ (root-compatibility by \autoref{lem:rootcompcheck}), this phase has time complexity $2^{2^{\OO(k)}}$.

Let $(f^C_1, f^C_2, f^C_3) = \crep_u(E_u) \cup \crep_v(E_v)$.
Now $f^C_i$ is a concrete representative of an embedding of $G[(X \cap C_i) \cup (Y \cap C_i), (X \setminus C_i) \cup (Y \setminus C_i)] = G[C_i, \overline{C_i}]$ to $R_{k_i}$, and therefore by \autoref{observation:embrep} we obtain a representative of $(C_i, \overline{C_i})$ of size at most $2 \cdot 2^{k_i}$ from $f^C_i$.
We compute a minimal representative $C^R_i$ of $C_i$ in time $2^{\OO(k)}$ by \autoref{lem:minrep_computation}

We compute the edit set and the neighbor partition with a BFS-type algorithm that maintains a queue $Q$ containing pairs $(w, E_w)$, where $w \in V(T)$, with the invariant that if $w$ is in the $r$-subtree of $u$, then $E_w$ is the $T_r[w]_{\reps_r[w]}$-representative of $E_X \cap T_r[w]$ and if $w$ is in the $r$-subtree of $v$, then $E_w$ is the $T_r[w]_{\reps_r[w]}$-representative of $E_Y \cap T_r[w]$.
We start by inserting the pairs $(u, E_u)$ and $(v, E_v)$ to $Q$.
Then, we iteratively pop a pair $(w, E_w)$ from the queue.
If there is $i$ such that $E_w$ is $C_j$-empty for all $j \neq i$, then it holds that $T_r[w] \subseteq C_i$, and therefore we insert $w$ to $N_i$.
Otherwise, we insert $w$ to $R$, let $w_1$ and $w_2$ be the $r$-children of $w$, and let $(E_1, E_2) = \dplink_w(E_w)$.
We insert $(E_1, w_1)$ and $(E_2, w_2)$ into $Q$.
This maintains the invariant by the definition of a linked $r$-table.

For each node $w \in R \cup N_1 \cup N_2 \cup N_3$ we just access tables indexed by $E_w$ and determine $C_i$-emptiness, so the amount of work is bounded by $|R \cup N_1 \cup N_2 \cup N_3| 2^{2^{\OO(k)}} = |R| 2^{2^{\OO(k)}}$.
\end{proof}

What is left is the Refine($R$, $(N_1, N_2, N_3)$) operation.
Computing the refinement $T'$ of $T$ itself is a direct application of \autoref{lem:refinealg} and computing the linked $r$-tables of the new nodes is done by $|R|$ applications of \autoref{lem:nonleafrtable}, but in order to compute the linked $r$-tables we must first make $T'$ augmented, that is, for each new edge $uv$ compute a minimal representative of $(T'[uv], T'[vu])$.
We use the minimal representatives of $C_1$, $C_2$, and $C_3$ returned by the algorithm of \autoref{lem:rankrefinement} for this.


\begin{lemma}
\label{lem:reprefinemethod}
Let $T$ be a rank decomposition, $r \in E(T)$, and $(r, C_1, C_2, C_3)$ a \strongminsplit.
Let $w$ be a non-leaf node of $T$ and let $w_1$ and $w_2$ be the $r$-children of $w$.
Let $i$,$j$, and $l$ be such that $\{i,j,l\} = \{1,2,3\}$.
Given minimal representatives of $T_r[w_1] \cap C_i$, $T_r[w_2] \cap C_i$, $\overline{T_r[w]}$, $C_j$, and $C_l$, a minimal representative of $(T_r[w] \cap C_i, \overline{T_r[w] \cap C_i})$ can be computed in $2^{\OO(k)}$ time.
\end{lemma}
\begin{proof}
As $(r, C_1, C_2, C_3)$ is a \strongminsplit of $T$, by \autoref{the:glostro} and \autoref{lem:cutrank_rep} each of the given minimal representatives has size at most $2^k$.

It holds that \[T_r[w] \cap C_i = (T_r[w_1] \cap C_i) \cup (T_r[w_2] \cap C_i),\] so by \autoref{lem:unirep} we obtain a representative of $T_r[w] \cap C_i$ as the union of the given minimal representatives of $T_r[w_1] \cap C_i$ and $T_r[w_2] \cap C_i$.
It also holds that 
\[\overline{T_r[w] \cap C_i} = \overline{T_r[w]} \cup \overline{C_i} = \overline{T_r[w]} \cup C_j \cup C_l,\]
so again by \autoref{lem:unirep} we obtain a representative of $\overline{T_r[w] \cap C_i}$ as the union of the given minimal representatives of $\overline{T_r[w]}$, $C_j$, and $C_l$.
Now we have a representative of $(T_r[w] \cap C_i, \overline{T_r[w] \cap C_i})$ of size $2^{\OO(k)}$, so we can use \autoref{lem:minrep_computation} to compute a minimal representative in time $2^{\OO(k)}$.
\end{proof}

In particular, a minimal representative of $\overline{T_r[w]}$ is available in $(T, \reps)$ as $\reps[pw]$, where $p$ is the $r$-parent of $w$, and minimal representatives of $T_r[w_1] \cap C_i$ and $T_r[w_2] \cap C_i$ are available by doing the construction in an order towards the root $r$.

This completes the description of the refinement data structure for augmented rank decompositions and thus also the proof of \autoref{the:rankwidth_refinement_ds}.

\section{Approximating graph branchwidth}\label{sec_gr_bw}
In this section we prove the following theorem.

\thmbranchwidthalgorithm*

We start with the definition of branch decomposition of a graph. 
\begin{definition}[Border]
Let $G$ be a graph and $X \subseteq E(G)$.
The border $\bd_G(X)$ of $X$ is the set of vertices of $G$ that are incident to an edge in $X$ and to an edge in $E(G) \setminus X$.
\end{definition}

The following result is well-known. 
\begin{proposition}
The function $|\bd_G| : 2^{E(G)} \rightarrow \mathbb{Z}_{\ge 0}$ assigning for each $X \subseteq E(G)$ the cardinality  $|\bd_G(X)|$ of the border of $X$,  is a connectivity function.
\end{proposition}

A branch decomposition of a graph  $G$ is a branch decomposition of function $|\bd_G|$, and the branchwidth of $G$, denoted by $\bw(G)$ is the branchwidth of $|\bd_G|$.
In the rest of this section we assume that we are computing the branchwidth of a graph $G$ and drop the subscript.

For branchwidth we do not need iterative compression as we can use the algorithm of Korhonen~\cite{Korhonen21} and the connection between branchwidth and treewidth~\cite{SeymourTh94} to obtain a branch decomposition of $G$ of width at most $3k$ in $2^{\OO(k)} n$ time.

The following is the key lemma, and the whole section is devoted to its proof. 

\begin{lemma}
\label{the:branchwidth_refinement_ds}
There is a refinement data structure for branch decompositions of graphs with time complexity $t(k) = 2^{\OO(k)}$.
\end{lemma}

With \autoref{the:main_alg} and the aforementioned connections to treewidth,  this will prove \Cref{theorem_vw_algo}.
%

The remaining part of this section is organized as follows. 
In  \Cref{subsec:bwaug} we define augmented branch decompositions, in \Cref{subsec:bwdp} we introduce the objects manipulated in our dynamic programming and prove some properties of them. Then \Cref{subsec:bwrefimpl} we give the refinement data structure for branch decompositions using this dynamic programming.

\subsection{Augmented branch decompositions\label{subsec:bwaug}}
In our refinement data structure we maintain an augmented branch decomposition.
An augmented branch decompositions stores the \emph{border description} of each cut.

\begin{definition}[Border description]
Let $X \subseteq E(G)$.
The border description of $X$ is the pair $(\bd(X), f)$, where $f : \bd(X) \rightarrow \mathbb{Z}_{\ge 0}$ is the function so that $f(v)$ is the number of edges in $X$ incident to $v$.
\end{definition}

An augmented branch decomposition is a branch decomposition $T$ where for each edge $uv \in E(T)$ the border descriptions of $T[uv]$ and $T[vu]$ are stored.
Note that an augmented branch decomposition can be represented in $\OO(|V(T)|\bw(T))$ space.

The following lemma leads to an algorithm for computing the border descriptions.

\begin{lemma}
\label{lem:bord_desc_comp}
Let $X$ and $Y$ be disjoint subsets of $E(G)$.
Given the border descriptions of $X$ and $Y$, the border description of $X \cup Y$ can be computed in $\OO(|\bd(X)| + |\bd(Y)|)$ time.
\end{lemma}
\begin{proof}
Let $(\bd(X), f)$ be the border description of $X$ and $(\bd(Y), g)$ the border description of $Y$.
It holds that $\bd(X \cup Y) \subseteq \bd(X) \cup \bd(Y)$, where $(\bd(X) \cup \bd(Y)) \setminus \bd(X \cup Y)$ can be identified as the vertices $v$ for which $f(v) + g(v)$ is the degree of $v$.
Now $(\bd(X \cup Y), f+g)$ is the border description of $X \cup Y$.
\end{proof}

It follows that given a branch decomposition $T$, a corresponding augmented branch decomposition can be computed in $\OO(|V(T)|\bw(T))$ time by applying \autoref{lem:bord_desc_comp} $2|V(T)|$ times.

\subsection{Borders of tripartitions\label{subsec:bwdp}}
We define the partial solutions stored in dynamic programming tables.

\begin{definition}[Border of a tripartition]
Let $A \subseteq E(G)$ and $(C_1, C_2, C_3)$ a tripartition of $A$.
The border of $(C_1, C_2, C_3)$ is the 9-tuple $(R_1, R_2, R_3, r_1, r_2, r_3, k_1, k_2, k_3)$,
where for each $i \in \{1,2,3\}$ it holds that $R_i = \bd(C_i) \cap \bd(A)$, $r_i = 0$ if $C_i = \emptyset$ and otherwise $r_i = 1$, and $k_i$ is the number of vertices $v \in V(G) \setminus \bd(A)$ such that there exists an edge $e_1 \in C_i$ incident to $v$ and an edge $e_2 \in A \setminus C_i$ incident to $v$, i.e., $k_i = |\bd(C_i) \cap \bd(A \setminus C_i) \setminus \bd(A)|$.
\end{definition}

We call a border of tripartition $k$-bounded if for each $i$ it holds that $k_i \le k$.
Note that if $|\bd(A)| \le k$, then the number of $k$-bounded borders of tripartitions of $A$ is $\le (2^k)^3 2^3 k^3 = 2^{\OO(k)}$, and each of them can be represented in $\OO(k)$ space.

We define the \emph{composition} of borders of tripartitions to combine partial solutions.

\begin{definition}[Composition]
Let $X$ and $Y$ be disjoint subsets of $E(G)$ and $A = X \cup Y$.
Let $R_X = (R^X_1, R^X_2, R^X_3, r^X_1, r^X_2, r^X_3, k^X_1, k^X_2, k^X_3)$ be the border of a tripartition of $X$ and $R_Y = (R^Y_1, R^Y_2, R^Y_3, r^Y_1, r^Y_2, r^Y_3, k^Y_1, k^Y_2, k^Y_3)$ the border of a tripartition of $Y$.
Denote $F = (\bd(X) \cup \bd(Y)) \setminus \bd(A)$.
The composition of $R_Y$ and $R_Y$ is the 9-tuple $(R_1, R_2, R_3, r_1, r_2, r_3, k_1, k_2, k_3)$, where for each $i \in \{1,2,3\}$,
\begin{enumerate}
\item $R_i = \bd(A) \cap (R^X_i \cup R^Y_i)$,
\item $r_i = \max(r^X_i, r^Y_i)$, and
\item $k_i = k^X_i + k^Y_i + |F \cap (R^X_i \cup R^Y_i) \cap (R^X_j \cup R^X_l \cup R^Y_j \cup R^Y_l)|$, where $\{i,j,l\} = \{1,2,3\}$.
\end{enumerate}
\end{definition}

Note that when the sets $\bd(X)$, $\bd(Y)$, and $\bd(A)$ are given and have size $\le k$, the composition can be computed in $\OO(k)$ time.

Next we prove that the composition operation really combines partial solutions as expected.

\begin{lemma}
\label{lem:bwdpmain}
Let $X$ and $Y$ be disjoint subsets of $E(G)$ and $A = X \cup Y$.
If $R_X$ is the border of a tripartition $(C^X_1, C^X_2, C^X_3)$ and $R_Y$ is the border of a tripartition $(C^Y_1, C^Y_2, C^Y_3)$, then the composition of $R_X$ and $R_Y$ is the border of the tripartition $(C^X_1 \cup C^Y_1, C^X_2 \cup C^Y_2, C^X_3 \cup C^Y_3)$.
\end{lemma}
\begin{proof}
Let $V^X$ be the set of vertices incident to $X$ and $V^Y$ the set of vertices incident to $Y$.
Let also $F = (\bd(X) \cup \bd(Y)) \setminus \bd(A)$.
Let $i \in \{1,2,3\}$ and let $j$ and $l$ be such that $\{i,j,l\} = \{1,2,3\}$.
It holds that
\[\bd(C^X_i \cup C^Y_i) = (\bd(C^X_i) \cup \bd(C^Y_i)) \cap (\bd(C^X_j) \cup \bd(C^X_l) \cup \bd(C^Y_j) \cup \bd(C^Y_l) \cup \bd(\overline{A})).\]
Therefore by observing that $\bd(\overline{A}) = \bd(A)$, $\bd(A) \cap V^X \subseteq \bd(X)$, and $\bd(A) \cap V^Y \subseteq \bd(Y)$ we get
\[\bd(C^X_i \cup C^Y_i) \cap \bd(A) = (\bd(C^X_i) \cup \bd(C^Y_i)) \cap \bd(A) = (R^X_i \cup R^Y_i) \cap \bd(A),\]
so the sets $R_1$, $R_2$, and $R_3$ in the composition are correct.

By observing that $V^X \cap F \subseteq \bd(X)$ and $V^Y \cap F \subseteq \bd(Y)$,  we have that

\begin{equation*}
\begin{split}
\bd(C^X_i \cup C^Y_i) \cap F =& F \cap (\bd(C^X_i) \cup \bd(C^Y_i)) \cap (\bd(C^X_j) \cup \bd(C^X_l) \cup \bd(C^Y_j) \cup \bd(C^Y_l))\\
=& F \cap (R^X_i \cup R^Y_i) \cap (R^X_j \cup R^X_l \cup R^Y_j \cup R^Y_l).
\end{split}
\end{equation*}
Since $V^X \setminus (F \cup \bd(A)) = V^X \setminus \bd(X)$ and $V^Y \setminus (F \cup \bd(A)) = V^Y \setminus \bd(Y)$ are disjoint,  we get that
\begin{equation*}
\begin{split}
 |\bd(C^X_i \cup &C^Y_i) \setminus (F \cup \bd(A))|\\
=& |(\bd(C^X_i) \cap (\bd(C^X_j) \cup \bd(C^X_l)) \setminus \bd(X)| + |(\bd(C^Y_i) \cap (\bd(C^Y_j) \cup \bd(C^Y_l)) \setminus \bd(Y)|\\
=& |\bd(C^X_i) \cap \bd(X \setminus C^X_i) \setminus \bd(X)| + |\bd(C^Y_i) \cap \bd(Y \setminus C^Y_i) \setminus \bd(Y)|.
\end{split}
\end{equation*}
Hence the numbers $k_1$, $k_2$, and $k_3$ in the composition are correct.
The numbers $r_1$, $r_2$, and $r_3$ are correct by observing that $C^X_i \cup C^Y_i$ is empty if and only if both $C^X_i$ and $C^Y_i$ are empty.
\end{proof}

The next lemma gives the method for determining if there exists a \splitW based on borders of tripartitions.

\begin{lemma}
\label{lem:dpbwroot}
Let $T$ be a rank decomposition, $uv = r \in E(T)$, and $W = T[uv]$.
Denote $X = W$ and $Y = \oW$.
There exists a \splitW $(C_1, C_2, C_3)$ with arity $\alpha$ and $|\bd(C_i)| = k_i$ for each $i \in \{1,2,3\}$ if and only if there exists $R_X$ and $R_Y$ such that
\begin{enumerate}
\item $R_X = (R^X_1, R^X_2, R^X_3, r^X_1, r^X_2, r^X_3, k^X_1, k^X_2, k^X_3)$ is the border of $(C_1 \cap X, C_2 \cap X, C_3 \cap X)$,
\item $R_Y = (R^Y_1, R^Y_2, R^Y_3, r^Y_1, r^Y_2, r^Y_3, k^Y_1, k^Y_2, k^Y_3)$ is the border of $(C_1 \cap Y, C_2 \cap Y, C_3 \cap Y)$,
\item the composition of $R_X$ and $R_Y$ is $(\emptyset, \emptyset, \emptyset, r_1, r_2, r_3, k_1, k_2, k_3)$, where $r_1 + r_2 + r_3 = \alpha$ and $k_i < |\bd(W)|/2$ for each $i$, and
\item for each $i$ it holds that $k^X_i + |R^X_i| < |\bd(W)|$ and $k^Y_i + |R^Y_i| < |\bd(W)|$.
\end{enumerate}
\end{lemma}
\begin{proof}
Suppose that such $R_X$ and $R_Y$ exists.
By \autoref{lem:bwdpmain} and the fact that $\bd(E(G)) = \emptyset$, $(\emptyset, \emptyset, \emptyset, r_1, r_2, r_3, k_1, k_2, k_3)$ is the border of $(C_1, C_2, C_3)$.
By the definition of border we have that $k_i = |\bd(C_i)|$ and $r_1 + r_2 + r_3$ is the number of non-empty sets in $(C_1, C_2, C_3)$.
It remains to prove that $|\bd(C_i \cap W)| = k^X_i + |R^X_i|$ and that $|\bd(C_i \cap \oW)| = k^Y_i + |R^Y_i|$ for each $i$.
We have that 
 \begin{equation*}
\begin{split}
\bd(W \cap C_i)=& (\bd(W \cap C_i) \cap \bd(\oW)) \cup (\bd(W \cap C_i) \cap \bd(W \setminus C_i))\\
=& (\bd(X \cap C_i) \cap \bd(X)) \cup (\bd(X \cap C_i) \cap \bd(X \setminus C_i))\\
=& (\bd(X \cap C_i) \cap \bd(X)) \cup (\bd(X \cap C_i) \cap \bd(X \setminus C_i) \cap \bd(X)) \\
&~~~~\cup (\bd(X \cap C_i) \cap \bd(X \setminus C_i) \setminus \bd(X))\\
=& (\bd(X \cap C_i) \cap \bd(X)) \cup (\bd(X \cap C_i) \cap \bd(X \setminus C_i) \setminus \bd(X)).
\end{split}
\end{equation*}

%
Therefore, by the definition of border, 
\begin{multline*}
|(\bd(X \cap C_i) \cap \bd(X)) \cup (\bd(X \cap C_i) \cap \bd(X \setminus C_i) \setminus \bd(X))|\\
= |\bd(X \cap C_i) \cap \bd(X)| + |\bd(X \cap C_i) \cap \bd(X \setminus C_i) \setminus \bd(X)| 
= |R^X_i| + k^X_i.
\end{multline*}
The other case is symmetric.

The above is the proof of the if direction.
The proof for the only if direction is the same but starting from supposing that such \splitW $(C_1, C_2, C_3)$ exists and letting $R_X$ be the border of $(C_1 \cap X, C_2 \cap X, C_3 \cap X)$ and $R_Y$ the border of $(C_1 \cap Y, C_2 \cap Y, C_3 \cap Y)$.
\end{proof}

\subsection{Refinement data structure for graph branch decompositions\label{subsec:bwrefimpl}}
In the refinement data structure for branch decompositions we maintain an augmented branch decomposition $T$ rooted at edge $r \in E(T)$ and a dynamic programming table that stores for each node $w \in V(T)$ all $k$-bounded borders of tripartitions of $T_r[w]$ and information about how many nodes in the $r$-subtree of $w$ they intersect.
We call this dynamic programming table an $r$-table, to signify that it is directed towards $r$.

In this subsection we always assume that $k$ is an integer such that $\bw(T) \le k$, and therefore we only care about $k$-bounded borders of tripartitions.
Next we formally define the contents of an $r$-table.

\begin{definition}[$r$-table]
Let $T$ be a branch decomposition, $r \in E(T)$, $w \in V(T)$, and $A = T_r[w]$.
The $r$-table of $w$ is the pair $(\bds, \intcount)$, where $\bds$ is the set of all $k$-bounded borders of tripartitions of $A$, and $\intcount$ is a function mapping each $R \in \bds$ to the least integer $i$ such that there exists a tripartition of $A$ whose border $R$ is and that $r$-intersects $i$ nodes of the $r$-subtree of $w$.
\end{definition}

As there are $2^{\OO(k)}$ $k$-bounded tripartitions of $T_r[w]$, the $r$-table of $w$ can be represented in $2^{\OO(k)}$ space.
  
\begin{lemma}
\label{lem:bwdpnonleafcomp}
Let $T$ be an augmented branch decomposition, $r \in E(T)$, and $w$ a non-leaf node of $T$ with $r$-children $w_1$ and $w_2$.
Given the $r$-tables of $w_1$ and $w_2$, the $r$-table of $w$ can be constructed in $2^{\OO(k)}$ time.
\end{lemma}
\begin{proof}
Denote $A = T_r[w]$, $X = T_r[w_1]$, and $Y = T_r[w_2]$.
Let $(\bds_{w_1}, \intcount_{w_1})$ and $(\bds_{w_2}, \intcount_{w_2})$ be the $r$-tables of $w_1$ and $w_2$.

We construct the $r$-table $(\bds, \intcount)$ of $w$ as follows. 
We iterate over all pairs $(R_X, R_Y) \in \bds_{w_1} \times \bds_{w_2}$ and let $\bds$ be the set of compositions of those pairs.
This correctly constructs $\bds$ by \autoref{lem:bwdpmain} and the observation that if $(C_1, C_2, C_3)$ is a tripartition of $A$ whose border is $k$-bounded, then $(C_1 \cap X, C_2 \cap X, C_3 \cap X)$ is a tripartition of $X$ whose border is $k$-bounded and $(C_1 \cap Y, C_2 \cap Y, C_3 \cap Y)$ is a tripartition of $Y$ whose border is $k$-bounded.
For each $R \in \bds$ we set $\intcount(R)$ as the minimum value of $\intcount_{w_1}(R_X) + \intcount_{w_2}(R_Y) + i_w$ over such pairs $R_X$, $R_Y$ whose composition $R$ is, where $i_w = 1$ if $R$ is of form $(\ldots, r_1, r_2, r_3, \ldots)$ where $r_1 + r_2 + r_3 \ge 2$, and $i_w = 0$ otherwise.
This correctly constructs $\intcount$ by the observations that $(C_1, C_2, C_3)$ $r$-intersects a node $w'$ in the $r$-subtree of $w_1$ if and only if $(C_1 \cap X, C_2 \cap X, C_3 \cap X)$ intersects $w'$, $(C_1, C_2, C_3)$ $r$-intersects a node $w'$ in the $r$-subtree of $w_2$ if and only if $(C_1 \cap Y, C_2 \cap Y, C_3 \cap Y)$ intersects $w'$, and that $(C_1, C_2, C_3)$ intersects $w$ if and only if $r_1 + r_2 + r_3 \ge 2$.

As $|\bds_{w_1}| |\bds_{w_2}| = 2^{\OO(k)}$, the time complexity is $2^{\OO(k)}$.
\end{proof}

Now the Init($T$, $r$) operation can be implemented in $|V(T)| 2^{\OO(k)}$ time by first making $T$ augmented by $2|V(T)|$ applications of \autoref{lem:bord_desc_comp} and then constructing the $k$-bounded $r$-tables of all nodes in the order from leafs to root by $|V(T)|$ applications of \autoref{lem:bwdpnonleafcomp}.
For the Move($vw$) operation, we note the $r$-table of a node $x \in V(T)$ depends only on the $r$-subtree of $x$, and therefore by \autoref{obs:reroot} it suffices to recompute only the $r$-table of the node $v$ when using Move($vw$).
Therefore Move($vw$) can be implemented in $2^{\OO(k)}$ time by a single application of \autoref{lem:bwdpnonleafcomp}.
The Width() operation returns $|\bd(T[uv])|$, which is available because $T$ is augmented.
The Output() operation is also straighforward as it just returns the branch decomposition we are maintaining.

The following lemma implements the operations CanRefine() and EditSet() based on \autoref{lem:dpbwroot}.

\begin{lemma}
Let $T$ be a branch decomposition, $uv = r \in E(T)$, and $W = T[uv]$.
For each node $w \in V(T)$ let $(\bds_w, \intcount_w)$ be the $r$-table of $w$.
There is an algorithm that returns $\bot$ if there is no \splitW, and otherwise returns a tuple $(R, N_1, N_2, N_3)$, where $(r, C_1, C_2, C_3)$ is a \strongminsplit, $R$ is the edit set of $(r, C_1, C_2, C_3)$ and $(N_1, N_2, N_3)$ is the neighbor partition of $R$.
The algorithm works in time $2^{\OO(k)} (|R|+1)$.
\end{lemma}
\begin{proof}
Denote $X = W = T_r[u]$, $Y = \oW = T_r[v]$.
We iterate over all pairs $(R_X, R_Y) \in \bds_u \times \bds_v$, using \autoref{lem:dpbwroot} to either determine that there is no \splitW or to find a pair $(R_X, R_Y)$ such that there is a \strongminsplit $(r, C_1, C_2, C_3)$ so that $R_X$ is the border of $(C_1 \cap X, C_2 \cap X, C_3 \cap X)$, $R_Y$ is the border of $(C_1 \cap Y, C_2 \cap Y, C_3 \cap Y)$, and the number of nodes of $T$ $r$-intersected by $(C_1, C_2, C_3)$ is $\intcount_u(R_X) + \intcount_v(R_Y)$.
In time $2^{\OO(k)}$ we either find such a pair or conclude that there is no \splitW.

We compute the edit set and the neighbor partition with a BFS-type algorithm that maintains a queue $Q$ containing pairs $(w, R_w)$, where $w \in V(T)$, with the invariant that there exists a \strongminsplit $(r, C_1, C_2, C_3)$ so that for all pairs $(w, R_w)$ that appear in the queue, $R_w$ is the border of $(C_1 \cap T_r[w], C_2 \cap T_r[w], C_3 \cap T_r[w])$.
We start by inserting the pairs $(u, R_X)$ and $(v, R_Y)$ to $Q$.
We iteratively pop a pair $(w, R_w)$ from the queue.
Denote $R_w = (\ldots, r_1, r_2, r_3, \ldots)$.
If there is $i$ such that $r_j = 0$ for both $j \neq i$, then it holds that $T_r[w] \subseteq C_i$, and therefore we insert $w$ to $N_i$.
Otherwise, we insert $w$ to $R$, let $w_1$ and $w_2$ be the $r$-children of $w$, and find a pair $(R_{w_1}, R_{w_2}) \in \bds_{w_1} \times \bds_{w_2}$ so that the composition of $R_{w_1}$ and $R_{w_2}$ is $R_w$ and $\intcount_{w_1}(R_{w_1}) + \intcount_{w_2}(R_{w_2}) + 1 = \intcount_w(R_w)$.
This kind of pair exists and maintains the invariant by the definition of $r$-table and \autoref{lem:bwdpmain}.

For each node $w \in R \cup N_1 \cup N_2 \cup N_3$ we iterate over $\bds_{w_1} \times \bds_{w_2}$ and access some tables, so the total amount of work is bounded by $|R \cup N_1 \cup N_2 \cup N_3| 2^{\OO(k)} = |R| 2^{\OO(k)}$.
\end{proof}

What is left is the Refine($R$, $(N_1, N_2, N_3)$) operation.
Computing the refinement of $T$ is done in $\OO(|R|)$ time by applying \autoref{lem:refinealg}.
Computing the border descriptions of the newly inserted edges can be done in a bottom-up fashion starting from $N_1 \cup N_2 \cup N_3$ by $2 |R|$ applications of \autoref{lem:bord_desc_comp}.
Then, computing the $r$-tables of the newly inserted nodes can be done also in a similar fashion in $|R| 2^{\OO(k)}$ time by $|R|$ applications of \autoref{lem:bwdpnonleafcomp}.

This completes the description of the refinement data structure for branch decompositions and thus also the proof of \autoref{the:branchwidth_refinement_ds}.

\section*{Acknowledgements}
We thank Eunjung Kim and Sang-il Oum for answering our questions
about their paper.
We also thank Michał Pilipczuk for suggesting the problem of approximating rankwidth.

\bibliographystyle{siam}
\bibliography{book_kernels_fvf}
\newpage
\section*{Appendix. Definitions of treewidth and cliquewidth}
\paragraph{Treewidth.}
A {\em tree decomposition} of a  graph $G$ is a pair
$(X,T)$ where $T$ is a tree whose vertices we will call {\em
nodes} and $X=(\{X_{i} \mid i\in V(T)\})$ is a collection of
subsets of $V(G)$ such that
\begin{enumerate}
\item $\bigcup_{i \in V(T)} X_{i} = V(G)$,

\item for each edge $(v,w) \in E(G)$, there is an $i\in V(T)$
such that $v,w\in X_{i}$, and

\item for each $v\in V(G)$ the set of nodes $\{ i \mid v \in X_{i}
\}$ forms a subtree of $T$.
\end{enumerate}
The {\em width} of a tree decomposition $(\{ X_{i} \mid i \in V(T)
\}, T)$ equals $\max_{i \in V(T)} \{|X_{i}| - 1\}$. The {\em
treewidth} of a graph $G$ is the minimum width over all tree
decompositions of $G$.

\paragraph{Cliquewidth.} Let $G$ be a graph, and $k$ be a positive integer.
A \emph{$k$-graph} is a graph whose vertices are labeled by integers from
$\{1,2,\dots,k\}$. We call the $k$-graph consisting of exactly one vertex
labeled by some integer from $\{1,2,\dots,k\}$ an initial $k$-graph.
The \emph{cliquewidth}  is the smallest integer $k$ such that $G$
can be constructed by means of repeated application of the following four
operations on $k$-graphs: 
\begin{itemize}\item 
{\em   introduce}: construction of an initial $k$-graph
labeled by $i$ and denoted by $i(v)$ 
(that is, $i(v)$ is a $k$-graph with a single vertex),
\item    {\em  disjoint union}
(denoted by $\oplus$),
\item {\em  relabel}: changing all labels $i$ to $j$
(denoted by $\rho_{i\to j}$), and
\item {\em   join}: connecting all vertices
labeled by $i$ with all vertices labeled by $j$ by edges
(denoted by $\eta_{i,j}$). 
\end{itemize}

Using the symbols of these operation, we can construct well-formed expressions.  An expression is called \emph{$k$-expression} for $G$ if the graph produced by performing 
these operations, in the order defined by the expression, is isomorphic to $G$ when labels are removed, and the \emph{cliquewidth of $G$}  is the minimum $k$ such that there is a $k$-expression for $G$.

\end{document}